\documentclass[11pt,letterpaper]{article}
\usepackage[letterpaper,margin=1in]{geometry}
\usepackage{abstract,setspace,leftidx,indentfirst}
\usepackage{authblk}    % allow multiple authors
\usepackage{graphicx}% Include figure files
\usepackage{dcolumn}% Align table columns on decimal point
\usepackage{bm}% bold math
\usepackage[colorlinks]{hyperref}% add hypertext capabilities
\usepackage{physics}
\usepackage[squaren]{SIunits}
\usepackage{mathtools}
\usepackage{amsthm, amssymb, amsfonts}
\usepackage{xcolor}
\usepackage{mathrsfs}
\usepackage[normalem]{ulem}
\usepackage[ruled,vlined]{algorithm2e}
\usepackage{xurl}
\usepackage{comment}
\usepackage{enumitem}

\usepackage[
    backend=biber,
    style=alphabetic,
    maxalphanames=3,
    maxnames=999,
    giveninits=false,
    doi=false,
    url=false,
    isbn=false,
    eprint=false
]{biblatex}
\DeclareFieldFormat[article,inproceedings]{title}{\mkbibemph{#1}}
\AtEveryBibitem{
  \clearfield{month}
  \clearfield{issn}
}

\newtheorem{theorem}{Theorem}
\newtheorem{definition}[theorem]{Definition}
\newtheorem{lemma}[theorem]{Lemma}
\newtheorem{corollary}[theorem]{Corollary}
\newtheorem{fact}[theorem]{Fact}

\newtheorem*{remark*}{Remark}

\hypersetup{
	colorlinks=true,
	linkcolor=blue,
	filecolor=blue,      
	urlcolor=blue,
	citecolor=blue
}
\setlist[enumerate,1]{label=\textnormal{(\arabic*)}}

\title{Learning and Generating Mixed States Prepared by Shallow Channel Circuits}
\author[1]{Fangjun Hu\thanks{\href{mailto:fhu@quera.com}{fhu@quera.com}}}
\author[1]{Christian Kokail}
\author[1]{Milan Kornja\v ca}
\author[1]{Pedro L. S. Lopes}
\author[2]{Weiyuan Gong}
\author[1]{Sheng-Tao Wang}
\author[3]{Xun Gao}
\author[1]{Stefan Ostermann}
\affil[1]{QuEra Computing Inc.}
\affil[2]{School of Engineering and Applied Sciences, Harvard University}
\affil[3]{JILA and Department of Physics, CU Boulder}
\date{\vspace{-15mm}}

\begin{document}
\maketitle

\begin{abstract}
    Learning quantum states from measurement data is a central problem in quantum information and computational complexity.
    In this work, we study the problem of learning to generate mixed states on a finite-dimensional lattice. Motivated by recent developments in mixed state phases of matter, we focus on arbitrary states in the trivial phase.
    A state belongs to the trivial phase if there exists a shallow preparation channel circuit under which local reversibility is preserved throughout the preparation.
    We prove that any mixed state in this class can be efficiently learned from measurement access alone. Specifically, given copies of an unknown trivial phase mixed state, our algorithm outputs a shallow local channel circuit that approximately generates this state in trace distance. The sample complexity and runtime are polynomial (or quasi-polynomial) in the number of qubits, assuming constant (or polylogarithmic) circuit depth and gate locality. Importantly, the learner is not given the original preparation circuit and relies only on its existence.
    Our results provide a structural foundation for quantum generative models based on shallow channel circuits. In the classical limit, our framework also inspires an efficient algorithm for classical diffusion models using only a polynomial overhead of training and generation.
\end{abstract}

\section{Introduction}

Understanding the structure of efficiently preparable quantum states is a central goal in quantum information theory, computational complexity, and quantum many-body physics.
While significant progress has been made for pure states prepared by shallow unitary circuits \cite{kim_2024_learning, landau_2025_learning}, much less is known about the learnability of mixed states prepared by shallow channel circuits. 
In this work, we study the problem of efficiently \textit{learning and generating} such mixed states from their measurement data.

We consider the following concrete learning task. Suppose we are given independent copies of an unknown $n$-qubit mixed state $\rho$ on a $k$-dimensional lattice. 
We may perform measurements on these copies and collect classical data.
The goal is to output a description of a channel circuit $\mathcal{W}$ -- namely, a circuit of completely positive trace-preserving (CPTP) maps -- that generates a state $\mathcal{W}(\ket{0}\!\bra{0}^{\otimes n})$ approximating $\rho$ in trace distance.

Learning mixed states looks more challenging than learning pure states due to their higher amount of degrees of freedom.
In the absence of any structural assumptions, quantum state tomography requires resources exponential in $n$ in the worst case.
One possible way to overcome this exponential barrier is to impose structure on the target state. The most natural assumption is that the target state is prepared by a shallow channel circuit with bounded depth and gate locality. 
% This is a natural operational model for many noisy physical evolutions and open-system dynamics.
A question naturally arises: can mixed states prepared by shallow channel circuits be learned efficiently, both in sample complexity and computational time?

At first glance, this problem remains difficult. In fact, only the shallowness of the preparation circuit of mixed states does not guarantee ease of learning.
There exist shallow channel circuits that prepare states with long-range conditional mutual information (CMI) \cite{lee_2024_universal}, and a mixed state with long-range CMI is generically hard to learn \cite{kumar_2026_unlearnable}.
% indicating that a shallow circuit structure alone does not guarantee efficient learnability.

To better understand when efficient learning may be possible, we turn to insights from the theory of quantum phases of matter.
For pure states, a robust notion of topological phase was proposed in terms of equivalence under mutual shallow local unitary circuits \cite{chen_2010_local} and has been developed extensively in broad research spanning from topological phase classification to quantum information theory \cite{zeng_2019_quantum}.
Under this definition, a pure state that can be prepared using a shallow and local unitary circuit is classified as belonging to the \textit{trivial phase}, and is consequently regarded as exhibiting \textit{short-range entanglement}.
In recent years, substantial effort has also been devoted to extending phase classifications to mixed states \cite{sang_2024_mixed, sang_2025_statbility}. However, unlike unitary circuits, general CPTP gates are not invertible, which complicates attempts to define mixed-state phases via mutual shallow-and-local circuit equivalence.

\begin{figure}
    \centering
    \includegraphics[width=0.9\linewidth]{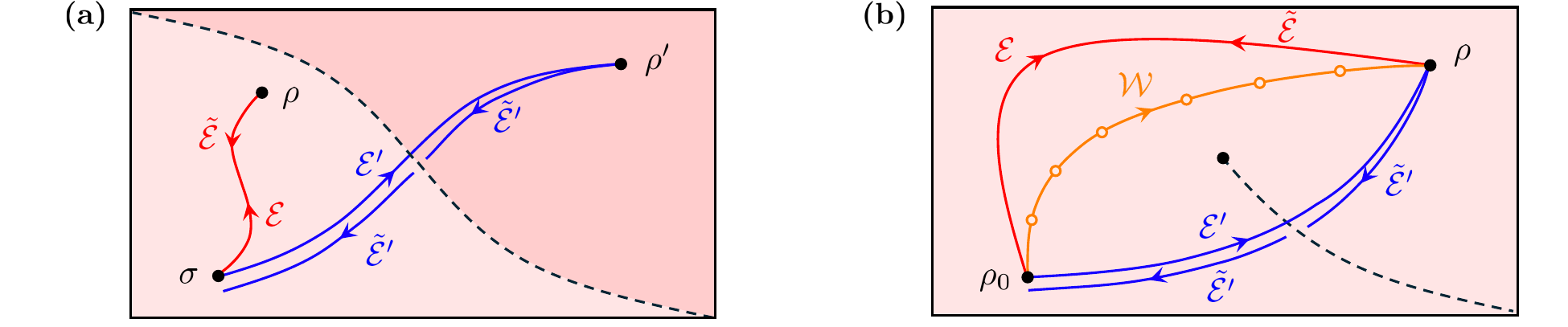}
    \caption{Phase diagrams in the space of mixed states.
    (a) For states $\rho$ and $\sigma$ in the same phase, one can transform $\rho$ and $\sigma$ mutually using shallow-and-local channel circuits $\mathcal{E}$ and $\tilde{\mathcal{E}}$ along one single path (red).
    For $\rho'$ in a different phase, any circuit $\mathcal{E}'$ preparing $\rho'$ from $\sigma$ (blue) must contain some step in the middle where local reversibility fails, denoted as the phase boundary (dashed line).
    (b) Trivial phase mixed state learnability. As long as a shallow-and-local channel circuits $\mathcal{E}$ (red) that prepares $\rho$ from $\rho_0=\ket{0}\!\bra{0}^{\otimes n}$ with local reversibility, our work always produces a new generation channel $\mathcal{W}$ (orange) that generates $\rho$, without knowing anything about $\mathcal{E}$.
    Notice that we do not exclude the possibility that there is some other preparation channel $\mathcal{E}'$ (blue) containing a phase transition, but our result only relies on the existence of a transition-free path (red), hence it still works in this scenario. 
    We remark that $\mathcal{W}$ may not exhibit local reversibility and can cross the phase boundary (not shown in the figure). Also, the preparation dynamics $\mathcal{E}$ can be time-continuous, but the dynamics of $\mathcal{W}$ in our construction is always time-discrete with $k+1$ steps.
    }
    \label{fig:schematic}
\end{figure}

To address the longstanding challenge of characterizing short-range entanglement in mixed states, the notion of \textit{local reversibility} has recently been introduced and studied in condensed matter theory.
Informally, local reversibility requires that information lost in a small region can be approximately recovered using operations supported on the same region \cite{sang_2025_mixed, ma_2025_circuit}, such that we can eliminate the effects caused by the gates in the preparation channel gate-by-gate.
% This condition serves as a quantitative proxy for short-range entanglement in mixed states.
Two mixed states are said to lie in the same phase if one can be transformed into the other via a shallow local channel circuit along a path in the space of mixed states that preserves local reversibility throughout the evolution (Figure \ref{fig:schematic}a).
We will see that this definition automatically implies that the second state can also be transformed back to the first state, with another shallow local channel circuit preserving local reversibility (Section \ref{sec:MSP_via_LR}).
Conversely, if any such transformation necessarily passes through intermediate states that violate local reversibility, the two states belong to distinct phases, and we say that a \textit{phase transition} occurs when local reversibility is violated during evolution.
In this sense, local reversibility provides a structural property that is stable under shallow channel circuits and characterizes a special class of \textit{trivial phase mixed states} -- the phase contains all product states.
For pure states in the trivial phase (i.e., prepared by shallow unitary circuits), recent works have established efficient learning algorithms \cite{kim_2024_learning, landau_2025_learning}, however, the analogous question for mixed states has remained open.

\subsection{Main result}
\label{sec:main_result}

In this paper, we give a positive result for mixed states prepared by shallow channel circuits with local reversibility.
Specifically, we show that mixed states in the trivial phase can be learned efficiently in both sample complexity and runtime.
We informally state our main result for learning a generation circuit of trivial phase mixed states as follows (see formal statement in Theorem \ref{thm:main}),

\begin{theorem}[Main Theorem, informal]
    \label{thm:main_informal}
    Suppose a mixed state $\rho$ in the trivial phase. Namely, $\rho$ is prepared by a shallow channel circuit $\rho =\mathcal{E}_d \circ \cdots \circ \mathcal{E}_2 \circ \mathcal{E}_1 ( \ket{0} \! \bra{0}^{\otimes n} )$ where each $\mathcal{E}_{\ell} = \prod_x \mathcal{E}_{\ell, x}$ is a layer of non-overlapping local channel gates.
    Each gate $\mathcal{E}_{\ell, x}$ acts with support size at most $c$, and there exists local channel gates $\{ \tilde{\mathcal{E}}_{\ell, x} \}$ that eliminate the effects caused by $\{ \mathcal{E}_{\ell, x} \}$ gate-by-gate. 
    Then, given trace distance error $\varepsilon$ and success probability $1 - \delta$, there exists an algorithm that runs in $\mathrm{poly} (n, 2^{(c d)^k}, 1/\varepsilon, \log ( 1/\delta ))$ time, learns from $\mathrm{poly} ( n, 2^{(c d)^k}, 1/\varepsilon, \log ( 1/\delta ))$ copies of $\rho$, and generates $\rho$ using a $(k + 1)$-layer channel circuit $\mathcal{W}$, where each gate size is at most $c'=O(c d)^k$.
\end{theorem}

Our result claims that learning and generating arbitrary trivial phase mixed states is efficient in the sense of both sample and computational complexity.
In particular, if $c = O(1)$ and $d = O(1)$, then our Theorem \ref{thm:main_informal} asserts that we can use polynomial time and samples to learn and generate $\rho$. If both $c = \mathrm{polylog}(n)$ and $d = \mathrm{polylog}(n)$, then $(c d)^k = \mathrm{polylog}(n)$ and $2^{(c d)^k}$ is a quasi-polynomial in $n$.
In quantum many body physics, the gate size of preparation circuit is typically assumed to scale as either $O(1)$ or $\mathrm{polylog}(n)$ \cite{coser_2019_classification, sang_2025_statbility}, and the parameters $c$ and $c'=O(c d)^k$ in Theorem \ref{thm:main_informal} satisfies the requirement.
Here, we allow a much larger channel gate size $c'$ in the generation circuit $\mathcal{W}$. In fact, we can show that each gate in $\mathcal{W}$ can be further decomposed with at most $(2d + 2)$ layers of gates, each of size at most $c$ (see Theorem \ref{thm:AM} and Theorem \ref{thm:LE}).
Therefore, if the circuit structure of preparation circuit $\mathcal{E}$ is known and $c=O(1)$, we can apply an epsilon-net searching to learn the gates in $\mathcal{W}$.
We will give a more detailed comment in Section \ref{sec:analysis}.

When taking all channels in our setting to be unitary gates, our result immediately resembles the efficient learning and generation algorithm given in Ref.\,\cite{kim_2024_learning}, where $\rho = \ket{\psi}\!\bra{\psi}$ with $\ket{\psi} = U \ket{0}^{\otimes n}$ for some shallow unitary circuit $U$.
Furthermore, the time and sample complexity scalings stated above are as good as the ones in the case of trivial phase pure state learning \cite{landau_2025_learning}.

An important feature of our result is that it relies only on the existence of a shallow local preparation path that preserves local reversibility. Concretely, suppose there exists a shallow local channel circuit $\mathcal{E}$ that prepares $\rho$ from the product state $\rho_0 = \ket{0}\!\bra{0}^{\otimes n}$ and such that all intermediate states along this preparation satisfy local reversibility. Then our algorithm, given only measurement access to $\rho$, outputs a (possibly different) shallow local generation circuit $\mathcal{W}$ that approximately generates $\rho$ (Figure \ref{fig:schematic}b). In particular, the learner is not given $\mathcal{E}$ or any description (for example, the circuit structure) of the preparation circuit.

We emphasize that our guarantee does not require all preparation paths to preserve local reversibility.
It is possible that there exists another shallow circuit $\mathcal{E}'$ preparing a trivial phase mixed state $\rho$ that passes through intermediate states violating local reversibility \cite{chen_2010_local, wu_2012_phase, hu_2025_local, kumar_2026_unlearnable}.
% -- for example, a uniform mixture over all closed loop configurations with trivial homology on a torus.
Our result is unaffected by such an alternative preparation path: the learning algorithm succeeds as long as there exists at least one transition-free path. Thus, learnability depends only on membership in the trivial phase, rather than on a specific preparation history.

% Our algorithm is based on two structural properties of such states: (i) \textit{approximate Markovianity}, a well-established consequence of short-range correlations; and (ii) \textit{local extendibility}, initially proposed in Ref.\,\cite{kim_2026_topological}, and was first leveraged to learning different phases of pure states in Ref.\,\cite{kim_2024_learning}.
% We then develop a learning procedure that leverages these two properties to iteratively construct a shallow channel circuit approximating the unknown state. We will give a more detailed outline of our techniques in Section \ref{sec:technical_overview}.

\subsection{Applications}
\label{sec:applications}

\textbf{Quantum generative models}. 
Quantum generative models aim to use quantum devices to generate target quantum states whose measurement outcomes obey probability distributions that are difficult to sample from classically \cite{bharti_2022_noisy}.
Prior work has shown that shallow unitary quantum circuits may offer generative advantages, under standard computational assumptions \cite{bremner_2010_classical, gao_2017_quantum, bermejovega_2018_architectures}, which can be leveraged to realize certain quantum generative models \cite{huang_2025_generative}.
% Some recent work leverages these results to encode an unknown distribution into a pure state, which can be efficiently learned and generated by using a shallow unitary quantum circuit, but it is hard to generate in a classical computer \cite{huang_2025_generative}.

% Even though most existing works focus on pure states generated by noiseless shallow unitary circuits, near-term quantum devices are inherently noisy and are more accurately modeled by shallow channel circuits rather than unitary circuits. 
Mixed states provide greater flexibility for encoding probability distributions and modeling stochastic processes.
One typical example of mixed state quantum generative models is the \textit{quantum diffusion model} \cite{zhang_2024_generative, hu_2025_local, liu_2025_measurement}. 
Quantum diffusion models generalize classical diffusion models by defining a forward noise (diffusion) process that gradually transforms the target state into a simple product state, together with a learned reverse (denoising) process that approximately reverses this evolution.
However, efficient learning of the reverse dynamics typically relies on a strong assumption: the learner is given explicit knowledge of a forward diffusion path that is guaranteed to avoid any phase transition \cite{hu_2025_local, liu_2025_measurement}.
In realistic generative settings, the target state is observed only through samples or measurement access, and whether the underlying preparation path exhibits a transition or not is unknown.

Our results provide a novel paradigm for quantum generative models. 
For states in the trivial phase, explicit knowledge of the forward noise path is unnecessary. It suffices that there exists a shallow local channel circuit preparing the state from a product state, without exhaustively searching for all possible phase-transition-free preparation circuits.
In contrast to prior approaches using Haar random unitaries as specific diffusion paths, whose complexity scales exponentially in system size \cite{zhang_2024_generative}, our algorithm runs in low time and sample complexity in the number of qubits.

\textbf{Channel circuit compilation}.
Our results provide a method for simplifying the preparation of mixed states by compiling potentially complex generation procedures into shallow local channel circuits.
This perspective is particularly relevant in the context of quantum Gibbs sampling.
Under rapid mixing assumptions, the target Gibbs state can be prepared by evolving an initial state under a local Lindbladian until convergence \cite{chen_2025_efficient}.
While such dynamics are efficient in principle, directly simulating the continuous-time evolution may still be computationally costly in practice.
In contrast, one can approximate the local behavior of the dynamics on small patches and then apply our learning algorithm to reconstruct a global generation circuit.
The resulting circuit has depth scaling only linearly with the lattice dimension $k$, offering a potentially significant reduction in implementation complexity.

% https://arxiv.org/abs/1803.00095
% https://arxiv.org/abs/2601.03426

\textbf{Learning phases of matter}. 
Circuit shallowness and gate locality play important roles in the field of topological order for pure states \cite{wen_1989_vacuum, wen_1990_topological, wen_2007_quantum, kitaev_2003_fault, chen_2010_local, haah_2011_local, haah_2016_invariant}, which was recently generalized to mixed state phases \cite{sang_2024_mixed, sang_2025_statbility}.
Our learning algorithm applies to arbitrary mixed states in the trivial phase. An important refinement of this class arises when additional symmetry constraints are imposed, leading to \textit{symmetry-protected topological} (SPT) mixed states \cite{coser_2019_classification, ma_2023_average, deGroot_2022_symmetry, ma_2025_symmetry}.
Let $G$ be a symmetry group with unitary representation $\{U_g\}_{g \in G}$ acting on the system. 
Two trivial phase mixed states $\rho$ and $\sigma$ are said to be in the same \textit{(weakly)-SPT} phase if $\rho = \mathcal{E}(\sigma)$ with a local channel circuit $\mathcal{E}$ that respects the symmetry action. Concretely, each channel gate $\mathcal{G}$ in the circuit satisfies the covariance condition $\mathcal{G}(U_g X U_g^\dagger) = U_g \mathcal{G}(X) U_g^\dagger$ for all $g \in G$.
Such symmetry constraints enrich the structure of trivial phase mixed states and play an important role in the study of symmetry-protected phases and phenomena such as strong-to-weak spontaneous symmetry breaking (SWSSB) \cite{lee_2025_symmetry, lee_2023_quantum, lessa_2025_strong, kuno_2025_strong}.
Our results imply that, given copies of an unknown SPT mixed state, we can always output a shallow local channel circuit that generates this state. This provides a systematic procedure for learning low-complexity generative descriptions of SPT mixed phases.

% Our results imply that these states can also be efficiently learned from measurement access alone. In particular, given copies of an unknown SPT mixed state, our algorithm reconstructs a shallow local channel circuit that generates the state, without requiring any prior knowledge of the underlying preparation dynamics. This provides a systematic procedure for learning low-complexity generative descriptions of SPT mixed phases.

\textbf{One-way test of trivial phase}. Our learning algorithm succeeds for all mixed states in the trivial phase.
This property naturally leads to a testing procedure for trivial phase mixed states: given measurement access to an unknown state $\rho$, we run the learning algorithm and check whether it produces a generation circuit. If the algorithm fails to produce such a circuit, this certifies that $\rho$ cannot be generated by any shallow local channel circuit satisfying local reversibility, and therefore $\rho$ does not belong to the trivial phase.

The guarantee is inherently one-way: if the algorithm succeeds, we cannot formally conclude that the state must lie in the trivial phase, since technically the learned generation circuit $\mathcal{W}$ may not satisfy local reversibility (see Figure \ref{fig:schematic}b).
Nevertheless, our algorithm can still effectively rule out trivial phase structures for many states of interest. For example, in quantum error-correcting codes, states corresponding to codes below the noise threshold are expected to violate the short-range entanglement constraints in our trivial phase mixed state learning algorithm \cite{kitaev_2003_fault, kim_2013_long, sang_2025_statbility}, and thus can be detected through failure of the learning procedure.

\subsection{Related works}
\label{sec:related_works}

\textbf{Quantum learning theory}.
In the pure state setting, several works have established efficient learning guarantees for states generated by shallow unitary circuits \cite{landau_2025_learning, kim_2024_learning}.
% In particular, Huang et al.~\cite{huang_2024_learning} showed that shallow quantum circuits can be efficiently learned. 
% Subsequent works developed algorithms for learning pure states prepared by shallow local unitary circuits \cite{landau_2025_learning, kim_2024_learning}.
In addition to learning quantum states, Ref.\,\cite{huang_2024_learning} also provided an algorithm that can learn a shallow quantum circuit given query access.
The more recent algorithm using this query model can further learn a unitary constructed out of a quantum circuit with depth double-logarithmic in the system size and an arbitrary Clifford circuit in any sequence~\cite{grewal2025query}.
Beyond this depth, it is known that learning circuits with depth logarithmic in the system size are hard in the worst case~\cite{huang_2024_learning}, and learning random circuits with depth super-logarithmic in the system size are hard under learning theory assumptions~\cite{nietner2025average,chen2025information} or cryptographic assumptions~\cite{fefferman2025hardness,schuster2025random}.

In contrast, the mixed state setting is less well understood. 
For general unstructured mixed states, learning is information-theoretically as hard as learning arbitrary pure states using random purification methods~\cite{pelecanos_2025_mixed}.
For structured mixed states, some positive results are known. 
% The steady states of local Lindbladians allow for efficient reconstruction of the Lindbladians \cite{bairey_2020_learning}.
In one dimension, quantum states admitting a matrix product operator (MPO) representation with constant bond dimension can be efficiently reconstructed from local measurements~\cite{votto_2025_learning}.
Moreover, an efficient algorithm to find the closest matrix product state is also known for one-dimensional quantum states that have constant fidelity with some matrix product states with bond dimensions scaling polynomially in the system size~\cite{bakshi2025learning}.
However, it remains unclear how to efficiently convert such descriptions into explicit shallow channel circuits that generate the states.

% \textbf{Topological phases for mixed states}. 
% Ref.\,\cite{ma_2025_circuit} constructs explicit shallow local channel circuits with local reversibility that connect these two Gibbs states of the same \textit{known} Hamiltonian, which lie in the same mixed state phase (not necessarily a trivial phase). This result has also been leveraged in the study of self-correcting quantum memories and related fault-tolerant constructions.

\textbf{Classical diffusion models}. Our structural framework also admits a natural classical limit. When restricted to commuting (diagonal) states, quantum channel gates reduce to classical noisy channel gates.
Under this limit, the preparation and its recovery process in our work exactly reduces to the denoising and diffusion process in the well-known classical diffusion models \cite{sohl_2015_deep, song_2019_generative, song_2021_sde, hu_2025_local}.
In particular, diffusion models typically begin from a simple product distribution (e.g., pixel-wise independent noise), which lies in the classical analogue of the trivial phase. Our results imply that any classical distribution in the trivial phase can be efficiently learned and generated via local noisy channels, with low overhead of sample and computational complexity.

\textbf{Unlearnability results}. 
Prior work has also investigated the problem of recognizing phases of matter from data.
In particular, it has been shown that for phases defined via shallow local circuits, either pure or mixed, identifying the phase of a state can require time exponential in the correlation length \cite{schuster_2025_hardness}.
On the other hand, from the perspective of state reconstruction, Ref.\,\cite{kumar_2026_unlearnable} shows that states exhibiting long-range CMI are generically hard to learn.
These results highlight intrinsic barriers to efficient learning in the absence of additional structural assumptions. 
% In contrast, our work shows that when local reversibility holds, mixed states in the trivial phase admit efficient learning algorithms.

\subsection{Discussions}

Our algorithm can be interpreted as providing only a one-way test for the trivial phase mixed state.
% successful output of the algorithm does not formally imply that the state must belong to the trivial phase.
An interesting direction for future work is to refine the constructive framework developed in this paper to obtain a two-way test that can conclusively determine whether a given mixed state belongs to the trivial phase.

More broadly, our work raises the question of identifying the fundamental limits of learnability for quantum states. While trivial phase mixed states admit efficient learning algorithms, many physically relevant states exhibit long-range correlations that may obstruct efficient reconstruction. It would be desirable to develop general criteria that characterize which classes of quantum states are inherently hard to learn from measurement data, and how such hardness relates to structural properties such as entanglement range, conditional mutual information, or phase structure.

In the application of quantum generative models, although our algorithm provides an efficient method for learning a quantum state that encodes a desired distribution, it remains unclear how to systematically encode classical information into quantum states and how to select appropriate measurement bases to sample from the target distribution. Developing methods for encoding classical data and extracting samples in this framework is an interesting direction for future work.

Finally, it is natural to ask whether the learning framework developed here can be extended beyond the trivial phase to more general topological mixed-state phases. Achieving this would have important consequences for quantum information processing. In particular, efficient learning of generative descriptions for such phases may lead to efficient decoding procedures for quantum error-correcting codes if the physical noise level remains below the error-correction threshold.

% \textbf{Note added}. During the preparation of this paper, we are aware of two forthcoming works that study local extendibility in a more systematic manner in the context of the entanglement bootstrap \cite{kim_2026_topological, sang_2026}.
% Their results establish the generic structural properties of locally extendible states, which are complementary to the learning perspective considered in this work.

% \section{Summary of Results}
% \label{sec:summary}

\section{Technical Overview}
\label{sec:technical_overview}

\subsection{Trivial phase mixed states via local reversibility}
\label{sec:problem_setting}

\begin{figure}[t]
    \centering    
    \includegraphics[width=0.5\linewidth]{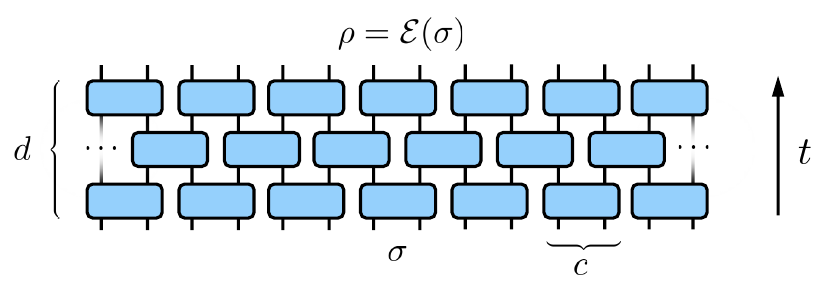}
    \caption{
    Schematic of the general circuit setting. We consider a channel circuit $\mathcal{E}$ of depth $d$, where each gate acts on at most $c$ qubits. The resulting lightcone extends over a region of radius at most $s = c d$.
    }
    \label{fig:circuit}
\end{figure}

We suppose a $k$-dimensional lattice $\Lambda$ with linear size $L$ and qubit number $n = O(L^k)$.
In the traditional study of phases of matter for quantum pure states, two states $\ket{\psi}$ and $\ket{\phi}$ are in the same phase if they are connected by a shallow (\textit{i.e.}, constant or $\mathrm{polylog} (n)$ depth) local unitary gate circuit $U$, namely $\ket{\psi} = U \ket{\phi}$ \cite{chen_2010_local}.
Consequently, the trivial phase is defined as the collection of states $\ket{\psi} = U \ket{0}^{\otimes n}$ that is prepared by a constant or $\mathrm{polylog} (n)$ depth local unitary gate circuit $U$. This definition is consistent, since $\ket{\psi} = U \ket{\phi}$ always automatically implies that $\ket{\phi} = U^{\dagger} \ket{\psi}$ and $U^{\dagger}$ is also a shallow and local circuit.

Now we consider a mixed state $\rho$ prepared by a shallow circuit with $d$ layers of local channel gates with trivial state $\rho_0 = \ket{0} \! \bra{0}^{\otimes n}$ as the input state:
\begin{equation}
    \rho =\mathcal{E} (\rho_0), \quad \mathcal{E}=\mathcal{E}_d \circ \cdots \circ \mathcal{E}_2 \circ \mathcal{E}_1,
\end{equation}
where each $\mathcal{E}_{\ell} = \prod_x \mathcal{E}_{\ell, x}$ is a layer of non-overlapping (quasi)-local channel gates $\mathcal{E}_{\ell, x}$ and $x$ is the position label of $\mathcal{E}_{\ell, x} $ (see Figure \ref{fig:circuit}). 
Here, \textit{shallow} means the depth of the circuit is constant $d=O(1)$ or polylogarithmic $d = \mathrm{polylog} (n)$; \textit{local} means the size of support of each $\mathcal{E}_{\ell, x}$ is at most $c = O(1)$; \textit{quasi-local} means the size of support of each $\mathcal{E}_{\ell, x}$ is at most $c = \mathrm{polylog} (n)$.
Since our main theorem does not distinguish between $c = O(1)$ and $c = \mathrm{polylog} (n)$, we will just use the terminology ``local'' to refer to both of these two cases. 

\begin{figure}[t]
    \centering
    \includegraphics[width=\linewidth]{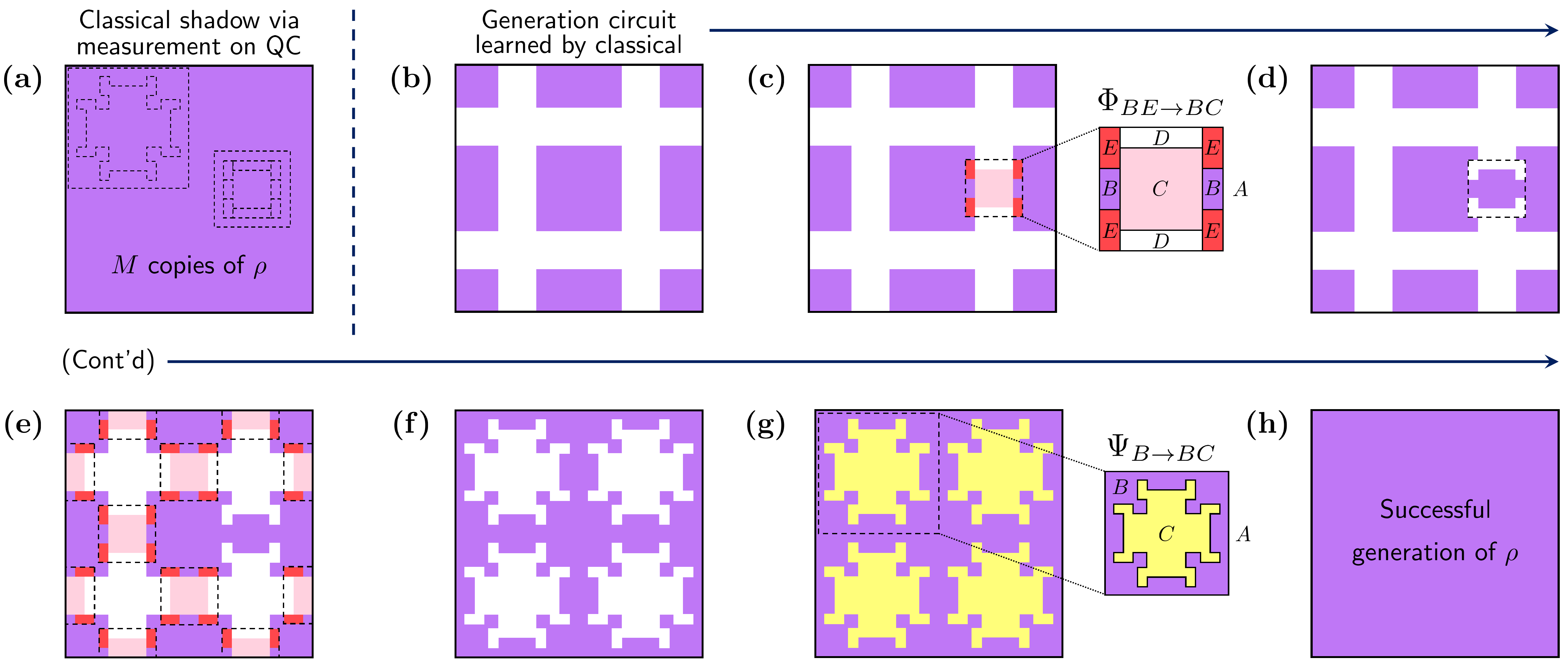}
    \caption{
    Learning scheme of a 2D trivial phase mixed state
    (generalized from Figure 13 of Ref.\,\cite{kim_2024_learning}).
    (a) Classical shadow by measurement for copies of $\rho$ on a quantum computer (QC). The snapshots of classical shadow on different local regions will be repeatedly used for learning and generation.
    (b-h) Learning and generating $\rho$ by using a $3$-layer circuit.
    (b) In layer 1, we learn and generate all local states with spatial support in purple.
    (c-f) Schematic of layer 2. All channels in layer 2 are local extension maps $\Phi_{BE \to BC}$ (Section \ref{sec:TO-LE}), exemplified in (c), where we learn new regions $C$ and discard $E$, leading to the state in (d).
    All local extension maps in layer 2 act in parallel, shown in (e), and the learned state has spatial support on (f).
    (g) In layer 3, we extend to the remaining hole-shaped region $C$ by using a local recovery map $\Psi_{B \to BC}$ supported on $BC$ (Section \ref{sec:TO-AM}). (h) We obtain the overall state $\rho$.
    All channels act locally and are efficiently learnable from classical shadow snapshots (Section \ref{sec:learning_locally}).
    }
    \label{fig:covering}
\end{figure}

For the phase of mixed states, the generalization of phase definition is not straightforward because if one can transfer from one mixed state $\sigma$ to another mixed state $\rho$ using some local shallow channel circuit $\mathcal{E} (\sigma) = \rho$, it is not obvious under what circumstance that there exists a local shallow channel circuit that transform from $\rho$ to $\sigma$.
The following observation for pure state phase enables us to make a reasonable generalization of phase definition for mixed states: if $\ket{\psi} = U \ket{\phi}$ for some shallow unitary circuits $U = U_d \cdots U_2 U_1$ where each $U_{\ell} = \prod_x U_{\ell, x}$ is a layer of non-overlapping local unitary gates $U_{\ell, x}$, then one can recover $\ket{\phi}$ from $\ket{\psi}$ by applying $\{ U^{\dagger}_{\ell, x} \}$ gate-by-gate from $\ell = d$ to $\ell = 1$. This motivates the following phase definition for mixed states: for any $\rho$ and $\sigma$, we say they are in the same phase if and only if $\rho =\mathcal{E} (\sigma)$ with local and shallow $\mathcal{E}$, and one can approximately recovery $\sigma$ by applying a sequence of local channel gates that cancel the effect of $\{ \mathcal{E}_{\ell, x} \}$ gate-by-gate. More specifically, for any $\ell, x$, there exists $\tilde{\mathcal{E}}_{\ell, x}$ with same support as $\mathcal{E}_{\ell, x}$ such that the following \textit{local reversibility} holds (see Definition \ref{def:LR} formally):
\begin{equation}
    \tilde{\mathcal{E}}_{\ell, x} \circ \mathcal{E}_{\ell, x} \circ \mathcal{E}_{\leq \ell - 1} (\sigma) \approx \mathcal{E}_{\leq \ell - 1} (\sigma),
\end{equation}
where $\mathcal{E}_{\leq \ell - 1} := \mathcal{E}_{\ell - 1} \circ \cdots \circ \mathcal{E}_2 \circ \mathcal{E}_1$ is the first $\ell - 1$ layers of $\mathcal{E}$ (see Figure \ref{fig:LR}).
We emphasize that $\tilde{\mathcal{E}}_{\ell, x} \circ \mathcal{E}_{\ell, x}$ is not necessarily an identity map, but it approximately stabilizes the particular state $\mathcal{E}_{\leq \ell - 1} (\sigma)$.
A phase transition occurs when such a $\tilde{\mathcal{E}}_{\ell, x}$ does not exists for some gate $\mathcal{E}_{\ell, x}$.
We will give a formal definition of local reversibility mixed state phases in Section \ref{sec:MSP_via_LR}. 
Especially, a \textit{trivial phase mixed state} is a state in the same phase as the trivial product state $\ket{0} \! \bra{0}^{\otimes n}$.

\begin{figure}[t]
    \centering    
    \includegraphics[width=0.95\linewidth]{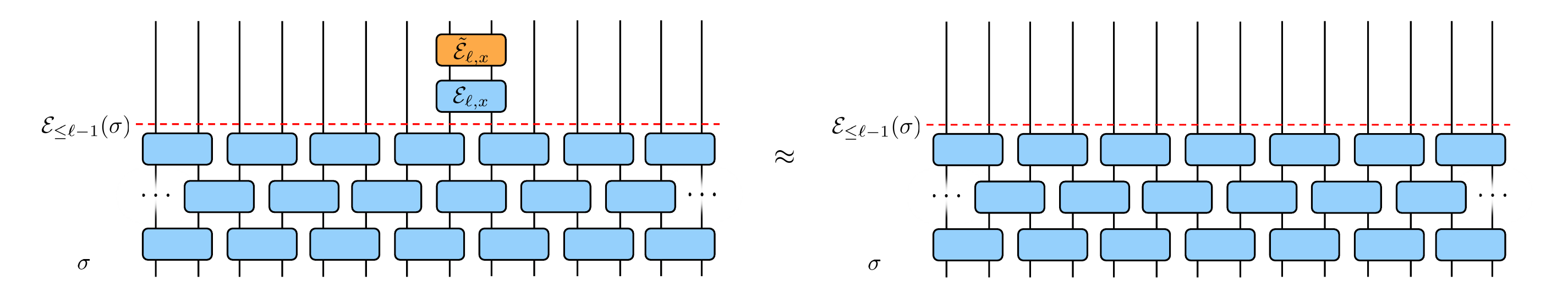}
    \caption{
    Schematic for local reversibility: any channel gate $\mathcal{E}_{\ell, x}$ can always be effectively canceled by some other channel gate $\tilde{\mathcal{E}}_{\ell, x}$.
    }
    \label{fig:LR}
\end{figure}

\subsection{Learning local systems}
\label{sec:local_tomography}

A natural starting point for learning a many-body mixed state is to reconstruct its local marginal $\rho_S$ in subsystems $S$. For local regions, this task is information-theoretically and computationally tractable: given polynomially many samples, one can estimate each local density matrix to constant 
accuracy using standard tomography procedures (for example, classical shadows \cite{huang_2020_predicting}).

However, learning all local marginals does not by itself yield a global state. The fundamental obstacle is a global consistency problem. An arbitrary collection of locally learned density matrices $\{ \rho_{S_1}, \rho_{S_2}, \cdots \}$ need not arise as the marginals of any single global mixed state.
Even if the reduced density matrices of every pair of overlapping regions are consistently the same, it is still hard to find a global state that simultaneously realizes all of them. 
This phenomenon is closely related to the quantum marginal problem, which is known to be as hard as the local Hamiltonian problem in general \cite{liu_2006_consistency}. Thus, the central challenge is not local learning, but how to sew together local pieces into a single globally consistent mixed state.

A naive approach would attempt to solve a global consistency program over all qubits simultaneously. Such a formulation quickly becomes computationally infeasible, as the global density matrix lives in a Hilbert space of dimension exponential in the system size $n$. The key difficulty is therefore to enforce global compatibility using only local information and local computations.

Our algorithm overcomes this obstacle by exploiting structural properties of trivial phase mixed states. 
These states are not arbitrary; we will see that they satisfy locality constraints that ensure global correlations are mediated through a short sequence of local correlations.
As a consequence, global consistency can be reduced to a sequence of local consistency. 
The remainder of this section explains the structural properties that make such a stitching procedure possible via two important properties of trivial phase mixed states: \textit{approximate markovianity} and \textit{local extendibility}.

\subsection{Approximate Markovianity}
\label{sec:TO-AM}

The first key structural property is the approximate Markovianity.
Informally, a tripartite state on regions $\Lambda = A B C$ is approximately Markovian if the correlations between $A B C$ are mediated almost entirely through the separator region $B$. Equivalently, conditioned on $B$, the regions $A$ and $C$ are nearly independent.

In quantum information theory, this property is quantified by the conditional mutual information (CMI) $I (A : C |  B) = S (A B) + S (B C) - S (B) - S (A B C)$. Exact vanishing of CMI $I (A : C |  B) = 0$ characterizes quantum Markov chains; small CMI implies the existence of a \textit{local recovery map} $\Psi_{B \rightarrow B C}$ that approximately recovers the global state (see Definition \ref{def:AM} formally):
\begin{equation}
    \Psi_{B \rightarrow B C} (\rho_{A B}) \approx \rho_{A B C} .
\end{equation}
A typical partition $\Lambda = A B C$ is given in Figure \ref{fig:partition}a. Under such a partition, as long as we have the knowledge of $\rho_{B C}$, even if we discard the subsystem $C$, the global state can still be recovered from the annulus-shaped $B$, without knowing any further information about $A$ or implementing any operation acting on $A$.

In Section \ref{sec:proof_AM}, we establish that trivial phase mixed states satisfy approximate Markovianity by explicitly constructing a local recovery map $\Psi_{B \rightarrow B C}$ from $B$ to $B C$ with controlled error (see Theorem \ref{thm:AM}).
Importantly, constructing an approximate local recovery map is algorithmically efficient. The definition of approximate Markovianity only involves the  $\rho_{BC}$ and $\rho_C$, which act on constant-size regions under our locality assumptions. 
Once the reduced density matrices $\rho_{B}$ and $\rho_{BC}$ are estimated to sufficient accuracy, one can compute a recovery map $\Psi_{B \to BC}$ by solving a local optimization problem \cite{berta_2016_fidelity}.
Moreover, the \textit{twirled Petz map} \cite{fawzi_2015_quantum, junge_2018_universal, mark_2016_quantum} admits a closed-form expression in terms of $\rho_{BC}$ and $\rho_C$, and is guaranteed to achieve near-optimal recovery, even without solving an optimization problem explicitly (see Appendix \ref{sec:cmi_decay}).

However, approximate Markovianity alone is not sufficient to guarantee that the locally learned marginals can be assembled into a single global state. It ensures that correlations are locally mediated, but it does not by itself provide a constructive mechanism for extending a partially reconstructed region. 
The scenario in Figure \ref{fig:partition}a, where one extends from an annulus-shaped region to fill its hole, is usually not the case when we want to extend two separate local subsystems into a larger, simply-connected component. However, such a connecting extension acting on two separate local subsystems is inevitable in learning the global states. For this reason, we require an additional structural ingredient, local extendibility.

\begin{figure}
    \centering
    \includegraphics[width=0.9\linewidth]{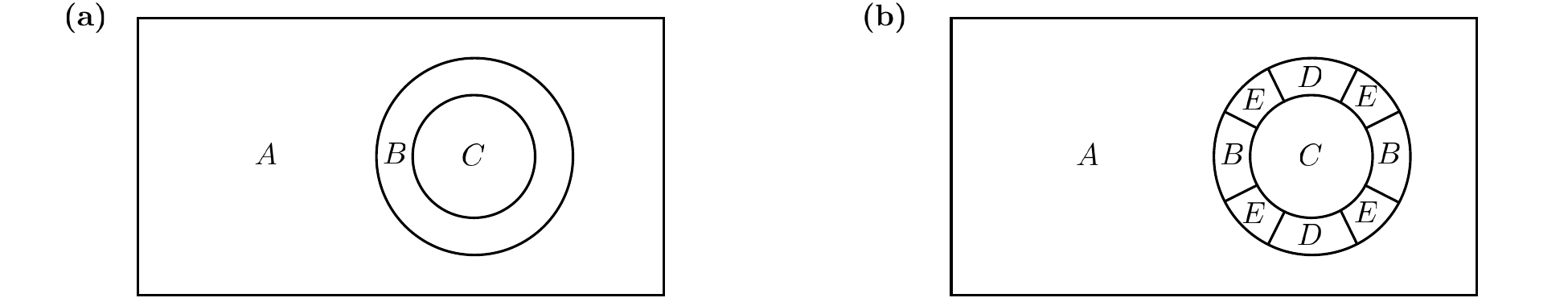}
    \caption{
    (a) Approximate Markovianity under tripartition $\Lambda = A B C$. There exists a channel $\Psi_{B \rightarrow B C}$ such that $\Psi_{B \rightarrow B C} (\rho_{A B}) \approx \rho_{A B C}$.
    (b) Local extendibility under pentapartition $\Lambda = A B C D E$. For $B' = B E$, there exists a channel $\Phi_{B E \rightarrow B C}$ such that $\Phi_{B E \rightarrow B C} (\rho_{A B E}) \approx \rho_{A B C}$. }
    \label{fig:partition}
\end{figure}

\subsection{Local extendibility}
\label{sec:TO-LE}

Local extendibility, a relaxed variant of approximate Markovianity, is first proposed in Ref.\,\cite{kim_2026_topological}, and leveraged to efficiently learn trivial phase pure states \cite{kim_2024_learning}.
The concept of local extendibility is initially introduced to complement the theory of \textit{entanglement bootstrap} \cite{shi_2020_fusion, shi_2021_entanglement, yang_2025_topological}.
Informally speaking, as we introduced, a quantum Markov recovery map is a strictly extensive map $B \rightarrow B C$, locally acting on $B C$; in contrast, a \textit{local extension map} $\Phi_{B' \rightarrow B C}$ is defined as a quantum channel where $B$ is a proper subset of $B'$ (namely $B \subsetneq B'$), locally acting on $B' \to BC$ (see Definition \ref{def:LE} formally):
\begin{equation}
    \Phi_{B' \rightarrow B C} (\rho_{A B'}) \approx \rho_{A B C} .
\end{equation}

The key feature of this definition is that the redundancy $B' \backslash B$ is discardable for extending from $B'$ toward $C$. This ``retreat-to-advance'' mechanism is crucial: it prevents the reconstructed region from accumulating incorrect entanglement as it grows. We will see that this property is particularly powerful when $B'$ is disconnected while $B C$ forms a connected component.
A typical example of local extension is depicted in Figure \ref{fig:partition}b, where $B' = B E$. 
Such configurations arise naturally from the lattice coverings in our reconstruction procedure.
% after learning several separated patches, the current reconstructed region may consist of multiple disconnected components, and the local extendibility property precisely guarantees the connection among these disconnected components.
We will see this in Section \ref{sec:learning_scheme} (Figure \ref{fig:covering}c).

In Section \ref{sec:proof_LE}, we prove that trivial phase mixed states always satisfy local extendibility, by explicitly constructing a local channel $\Phi_{B' \rightarrow B C}$ from $B'$ to $B C$ (with $B \subsetneq B'$) with controlled error (see Theorem \ref{thm:LE}).
Unlike local recovery maps, it is still unknown whether local extension maps have closed-form expressions like twirled Petz maps that can be learned efficiently in generic cases.
Nevertheless, we will show in Section \ref{sec:learning_locally} that when assisted with approximate Markovianity, any local extension map can be learned by efficiently solving a semi-definite programming algorithm.

\subsection{Learning scheme}
\label{sec:learning_scheme}

Combining approximate Markovianity and local extendibility, we obtain a global state reconstruction procedure that proceeds iteratively while maintaining consistency across overlapping regions. Crucially, this construction avoids the exponential blowup that would arise from naively stitching together local marginals.

The core idea is captured by a geometric \textit{covering scheme}.
Informally, consider a $k$-dimensional lattice. We reconstruct the target state $\rho$ starting from the product state $\ket{0}\!\bra{0}^{\otimes n}$ via $k+1$ sequential layers of local channels, gradually extending the support of the state from smaller regions to the entire lattice:

\vspace{0mm}
\begin{itemize}[itemsep=0mm, topsep=0mm, partopsep=0mm]
    \item \textit{Layer 1} (\textit{Initialization}). We generate disjoint local patches whose reduced states match the corresponding marginals of $\rho$.

    \item \textit{Layers 2 through $k$} (\textit{Local Extension}). At each subsequent layer, we apply local extension maps that connect the previously disconnected regions that have already been learned.

    \item \textit{Layer $k+1$} (\textit{Local Recovery}). In the final layer, we apply local recovery maps to reconstruct the remaining regions that are individually simply-connected.
\end{itemize}

After $k+1$ layers, the resulting state approximates $\rho$ in trace distance.
We give a visualization for the special case $k=2$ in Figure \ref{fig:covering}, and leave the formalized description of the covering scheme in Section \ref{sec:cover_scheme}.
Moreover, the sequence of local channels applied across all layers defines a shallow channel circuit $\mathcal{W}$ that generates $\rho$ from $\ket{0}\!\bra{0}^{\otimes n}$. Because each operation acts only on local neighborhoods and each layer consists of parallel local maps, the overall depth of $\mathcal{W}$ depends only on the lattice dimension $k$ (and not on system size $n$). Combined with the efficient learnability of the local recovery and local extension maps established earlier, this yields a reconstruction algorithm with (quasi)-polynomial sample complexity and runtime.

% \begin{figure}
%     \centering
%     \includegraphics[width=0.9\linewidth]{Figure-Covering-Intro.pdf}
%     \caption{
%     The covering scheme for the special case $k=2$. 
%     (a) Layer 1: we learn and generate all local states with spatial support in purple.
%     (b-e) Schematic of layer 2. All channels in layer 2 are local extension maps $\Phi_{BE \to BC}$, exemplified in (b), where we learn new regions $C$ from $BE$ and discard $E$, leading to the state in (c) (see pentapartition $\Lambda = ABCDE$ in Figure \ref{fig:partition}b).
%     All local extension maps in layer 2 act in parallel. Examples of these maps are visualized in (d), encircled by dashed lines.
%     After acting all local extension maps in layer 2, the learned state has spatial support on (e).
%     (c) In layer 3, we simply extend to the remaining hole-shaped region $C$ which is not learned yet by using a local recovery map $\Psi_{B \to BC}$ supported on local $BC$. 
%     All local extension maps and local recovery maps in this Figure act locally and are classical learnable from local classical shadows (Section \ref{sec:learning_locally}).
%     }
%     \label{fig:covering}
% \end{figure}

\subsection{Proof roadmap}
\label{sec:proof_roadmap}

In Section \ref{sec:setup}, we introduce the formal framework used throughout the paper. This includes notation and conventions, the definition of local reversibility, the formal definition of mixed phases and trivial phase mixed states, and the precise formulation of approximate Markovianity and local extendibility. We will also present an important technique called \textit{local inversion} -- arising from the local extendibility property -- at the end of this section. The subsequent proofs of approximate Markovianity and local extendibility will largely rely on this local inversion technique.

In Section \ref{sec:proof_AM} and \ref{sec:proof_LE},
we prove that trivial phase mixed states satisfy approximate Markovianity and local extendibility by showing the \textit{existence} of local recovery and local extension maps, respectively.
% Together, these two structural properties form the backbone of the learning algorithm. 
In the subsequent Section \ref{sec:learning_alg}, we combine these ingredients to analyze the full generation procedure.
We prove the main theorem by providing error and complexity analysis in Section \ref{sec:analysis}, and give the implications of our main result to quantum generative models in Section \ref{sec:application}.

\section{Preliminaries}
\label{sec:setup}

\subsection{Notions and concepts}
\label{sec:notions}

Throughout this paper, we always use $\mathcal{E}$ to denote the channel gate circuit that transform from $\ket{0} \! \bra{0}^{\otimes n}$ to $\rho$. We further assume we know the circuit structure of $\mathcal{E}=\mathcal{E}_d \circ \cdots \circ \mathcal{E}_2 \circ \mathcal{E}_1$ and $\mathcal{E}_{\ell} = \prod_x \mathcal{E}_{\ell, x}$. The support of each local $\mathcal{E}_{\ell, x}$ is at most $c$, and $x$ is the position label of $\mathcal{E}_{\ell, x}$ (see Figure \ref{fig:circuit}).
Here, \textit{quantum channel} (or simply \textit{channel}) is a completely positive trace-preserving (CPTP) map, and we use ``$\circ$'' to denote the map composition.
We always abbreviate $\mathcal{I}_{\Lambda\backslash S} \otimes \mathcal{C}_S$ as $\mathcal{C}_S$ for any channel $\mathcal{C}_S$ acting on $S$, where $\mathcal{I}_{\Lambda \backslash S}$ is the identity channel on $\Lambda\backslash S$.
Terminologically, we always refer to the unknown $\mathcal{E}$ as \textit{preparation channel} for $\rho$, and call the learned channel $\mathcal{W}$ (not necessarily equal to $\mathcal{E}$) in Theorem \ref{thm:main_informal} the \textit{generation channel}.

Some other notions and concepts are listed as follows:

\begin{itemize}
    \item \textbf{Subsystems}. We use $| S |$ as the size of subsystem $S \subseteq \Lambda$, namely the number of qubits in $S$. We abbreviate $\bar{S} = \Lambda \backslash S$. We denote $S (s)$ as the set of qubits that are either in $S$ or are distance at most $s$ away from $S$. 
    We denote $\mathrm{Supp} (\mathcal{C})$ as the support of channel circuit $\mathcal{C}$, namely the collection of qubits that $\mathcal{C}$ exactly acts on.

    \item \textbf{Sub-circuits}. For a channel circuit $\mathcal{C}$, we say it is a sub-circuit of $\mathcal{E}$, if its channel gates are a subset of $\mathcal{E}$, and $\mathcal{C}$ inherits the connection relation of gates between adjacent layers of $\mathcal{E}$. We use $| \mathcal{C} |$ as the size of sub-circuit $\mathcal{C} \subseteq \mathcal{E}$, namely the number of channel gates in $\mathcal{E}$. For any channel gate $\mathcal{E}_{\ell, x}$ of $\mathcal{E}$ and any sub-circuit $\mathcal{C}$, we denote $\mathcal{E}_{\ell, x} \in \mathcal{C}$ if $\mathcal{E}_{\ell, x}$ is in $\mathcal{C}$ and $\mathcal{E}_{\ell, x} \notin \mathcal{C}$ if $\mathcal{E}_{\ell, x}$ is not in $\mathcal{C}$.

    For any two sub-circuit $\mathcal{C}_1$ and $\mathcal{C}_2$, we denote $\mathcal{C}_1 \backslash \mathcal{C}_2$ be the remaining channel gates by removing all the gates that are simultaneously in $\mathcal{C}_1$ and $\mathcal{C}_2$, from $\mathcal{C}_1$: for layer decomposition $\mathcal{C}_1 =\mathcal{C}_{1, d} \circ \cdots \circ \mathcal{C}_{1, 1}$, we define
    \begin{equation}
        \mathcal{C}_1 \backslash \mathcal{C}_2 := \mathcal{Q}_d \circ \cdots \circ \mathcal{Q}_1, \text{\quad with } \mathcal{Q}_{\ell} = \prod_{\mathcal{E}_{\ell, x} \in \mathcal{C}_{1, \ell}, \mathcal{E}_{\ell, x} \notin \mathcal{C}_{2, \ell}} \mathcal{E}_{\ell, x} .
    \end{equation}
    We will also use $\mathcal{E}_{\leq \ell} :=\mathcal{E}_{\ell} \circ \cdots \circ \mathcal{E}_2 \circ \mathcal{E}_1$ to represent the first $\ell$ layers of $\mathcal{E}$.
    
    \item \textbf{Backward lightcones}. We denote $\mathcal{B}_S$ as the \textit{backward lightcone} of $S$, we will give a formal definition of $\mathcal{B}_S$ immediately in Section \ref{sec:backward_lightcone}.

    \item \textbf{Reset channels}. For convenience, we introduce the \textit{reset channel} $\mathcal{R}_S$ on subsystem $S$ such that for any operator $X$,
    \begin{equation}
        \mathcal{R}_S (X) = \mathrm{Tr}_S (X) \otimes \ket{0} \! \bra{0}_S .
    \end{equation}
    An useful property of reset channel $\mathcal{R}_S$ is that, for any channel $\mathcal{G}$ acting on some $S' \subseteq S$, we always have $\mathcal{R}_S \circ \mathcal{G}=\mathcal{R}_S$, because $\mathcal{R}_S \circ \mathcal{G} (X) = \mathrm{Tr}_S (\mathcal{G} (X)) \otimes \ket{0} \! \bra{0}_S = \mathrm{Tr}_S (X) \otimes \ket{0} \! \bra{0}_S =\mathcal{R}_S (X)$ for any operator $X$.

    \item \textbf{Distance measures for states}. Throughout the paper, we measure the distance between mixed states using the Schatten \textit{$1$-norm distance}. For density operators $\rho$ and $\sigma$ on the same Hilbert space, the $1$-norm distance is defined as
    \begin{equation}
        \| \rho - \sigma\|_1 = \mathrm{Tr} \! \left(\sqrt{(\rho - \sigma)^\dagger(\rho - \sigma)}\right).
    \end{equation}
    In quantum information theory, the trace distance is defined as $D_{\mathrm{tr}}(\rho, \sigma) = \frac{1}{2} \| \rho - \sigma\|_1$. 
    We also relate the 1-norm distance to the quantum fidelity, $F(\rho, \sigma) = \| \sqrt{\rho} \sqrt{\sigma} \|_1$.
    The Fuchs-van de Graaf inequalities imply $1-F(\rho, \sigma) \leq \frac{1}{2} \| \rho - \sigma\|_1 \leq \sqrt{1-F(\rho, \sigma)^2}$.

    Throughout the paper, we will repeatedly using the \textit{triangle inequality} $\| \rho - \sigma\|_1 \leq \| \rho - \tau\|_1 + \| \tau - \sigma\|_1$ and the \textit{contractivity} $\| \mathcal{C}(\rho) - \mathcal{C}(\sigma)\|_1 \leq \| \rho - \sigma\|_1$ for any state $\rho, \sigma, \tau$ and channel $\mathcal{C}$.

    \item \textbf{Errors}. There are multiple notions of errors (under 1-norm distance) in this paper. We list them here: $\varepsilon_{\mathrm{LR}}$ is the error of local reversibility for each channel gate,
    $\varepsilon_{\mathrm{LI}}$ is the error of the local inversion for each subsystem, 
    $\varepsilon_{\mathrm{LE}}$ is the error of local extension map,
    $\varepsilon_{\mathrm{LT}}$ is the error of local tomography for each subsystem, 
    and $\varepsilon_{\mathrm{SDP}}$ is the target error of running the semidefinite programming (SDP) algorithm stated in Section \ref{sec:SDP}.

    \item \textbf{Integer Set}. We use $[N]$ to represent the set $\{1, 2, \cdots, N\}$.
\end{itemize}

\subsection{Backward lightcone}
\label{sec:backward_lightcone}

\begin{definition}[Backward lightcone]
    \label{def:backward_lightcone}
    Given the circuit $\mathcal{E}$ and a region $S \subseteq \Lambda$. We say a sub-circuit is a \textnormal{backward lightcone} of $S$ for $\mathcal{E}$, denoted as $\mathcal{B}_S$, if: let $\mathcal{B}_S =\mathcal{B}_{S, d} \circ \cdots \circ \mathcal{B}_{S, 1}$,
    \begin{enumerate}%[label=(\arabic*)]
        \item For any $\mathcal{E}_{d, x} \in \mathcal{B}_{S, d}$, we have $\mathrm{Supp} (\mathcal{E}_{d, x}) \cap S \neq \varnothing$; and and for any $\mathcal{E}_{d, x} \notin \mathcal{C}_d$, we have $\mathrm{Supp} (\mathcal{E}_{d, x}) \cap S = \varnothing$ (\textit{i.e.}, $\mathcal{C}_d$ is the maximal gate set in layer $d$ such that no gate is entirely outside $S$);
    
        \item For any $\mathcal{E}_{\ell, x} \in \mathcal{B}_{S, \ell}$ with $\ell < d$, we have $\mathrm{Supp} (\mathcal{E}_{\ell, x}) \cap (\mathrm{Supp} (\mathcal{B}_{S, d} \circ \cdots \circ \mathcal{B}_{S, \ell + 1}) \cup S) \neq \varnothing$; and for any $\mathcal{E}_{\ell, x} \notin \mathcal{B}_{S,\ell}$, we have $\mathrm{Supp} (\mathcal{E}_{\ell, x}) \cap (\mathrm{Supp} (\mathcal{B}_{S, d} \circ \cdots \circ \mathcal{B}_{S, \ell + 1}) \cup S) = \varnothing$ (\textit{i.e.}, $\mathcal{B}_{S, \ell}$ is the maximal gate set in layer $\ell$ such that no gate is entirely outside $\mathrm{Supp} (\mathcal{B}_{S, d} \circ \cdots \circ \mathcal{B}_{S, \ell + 1}) \cup S$).
    \end{enumerate}
\end{definition}

The definition of backwark lightcone immediately implies $\mathrm{Supp} (\mathcal{B}_S) \subseteq S (s)$ with the \textit{lightcone extension width}, defined as
\begin{equation}
    s = c d. \label{eq:lightcone_radius}
\end{equation}
We have $s=O(1)$ when $c,d = O(1)$ and $s=\mathrm{polylog}(n)$ when $c,d = \mathrm{polylog}(n)$.
It means that for any shallow and local channel circuits, $\mathrm{Supp}(\mathcal{B}_S) \backslash S$ -- the region extended from $S$ for getting its backward lightcone -- is a narrow region whose thickness is at most $s$. 
For this reason, for any initial state $\sigma$ and region $S$, we call 
\begin{equation}
    \mathrm{Tr}_{S \cup \overline{\mathrm{Supp}(\mathcal{B}_S)}} (\mathcal{B}_S (\sigma)) \label{eq:edge_LC}
\end{equation}
the \textit{edge of the backward lightcone}.

We give two simple facts that will be useful in the following constructive proof of approximate Markovianity and local extendibility in Section \ref{sec:proof_AM} and Section \ref{sec:proof_LE}, respectively. The proofs are quite routine, so we leave them in Appendix \ref{sec:backward_lightcone_decomposition}. We give a visualization of these three facts in Figure \ref{fig:lightcone}.

\begin{fact}[Lightcone decomposition]
    \label{fac:LCD-1}
    Given any region $S = S_1 S_2$, then there is a channel $\mathcal{Q}=\mathcal{B}_{S_2} \backslash \mathcal{B}_{S_1}$ such that we can reorganize $\mathcal{B}_S$ into the following lightcone decomposition
    \begin{equation}
        \mathcal{B}_S =\mathcal{Q} \circ \mathcal{B}_{S_1} .
    \end{equation}
    Furthermore, $\mathrm{Supp} (\mathcal{Q}) \subseteq S_2 (s) \backslash S_1$.
\end{fact}

\begin{fact}
    \label{fac:LCD-2}
    Given any region $S$, there is a channel $\mathcal{Q}=\mathcal{B}_S \backslash \mathcal{B}_{\bar{S}}$ such that we can reorganize $\mathcal{E}$ into the following lightcone decomposition
    \begin{equation}
        \mathcal{E}=\mathcal{Q} \circ \mathcal{B}_{\bar{S}} .
    \end{equation}
    Furthermore, $\mathrm{Supp} (\mathcal{Q}) \subseteq S$. 
\end{fact}

\begin{fact}[Zero long-range correlation]
    \label{fac:LCD-3}
    Given any two regions $S_1$ and $S_2$ with $\mathrm{dist}(S_1, S_2) \geq 2s$, if the initial state is the product state $\rho_0 = \ket{0}\!\bra{0}^{\otimes n}$, then the reduced density matrix $\rho_{S_1 S_2}$ has the following product form
    \begin{equation}
        \rho_{S_1 S_2} = \rho_{S_1} \otimes \rho_{S_2} .
    \end{equation}
\end{fact}

\begin{figure}
    \centering
    \includegraphics[width=0.95\linewidth]{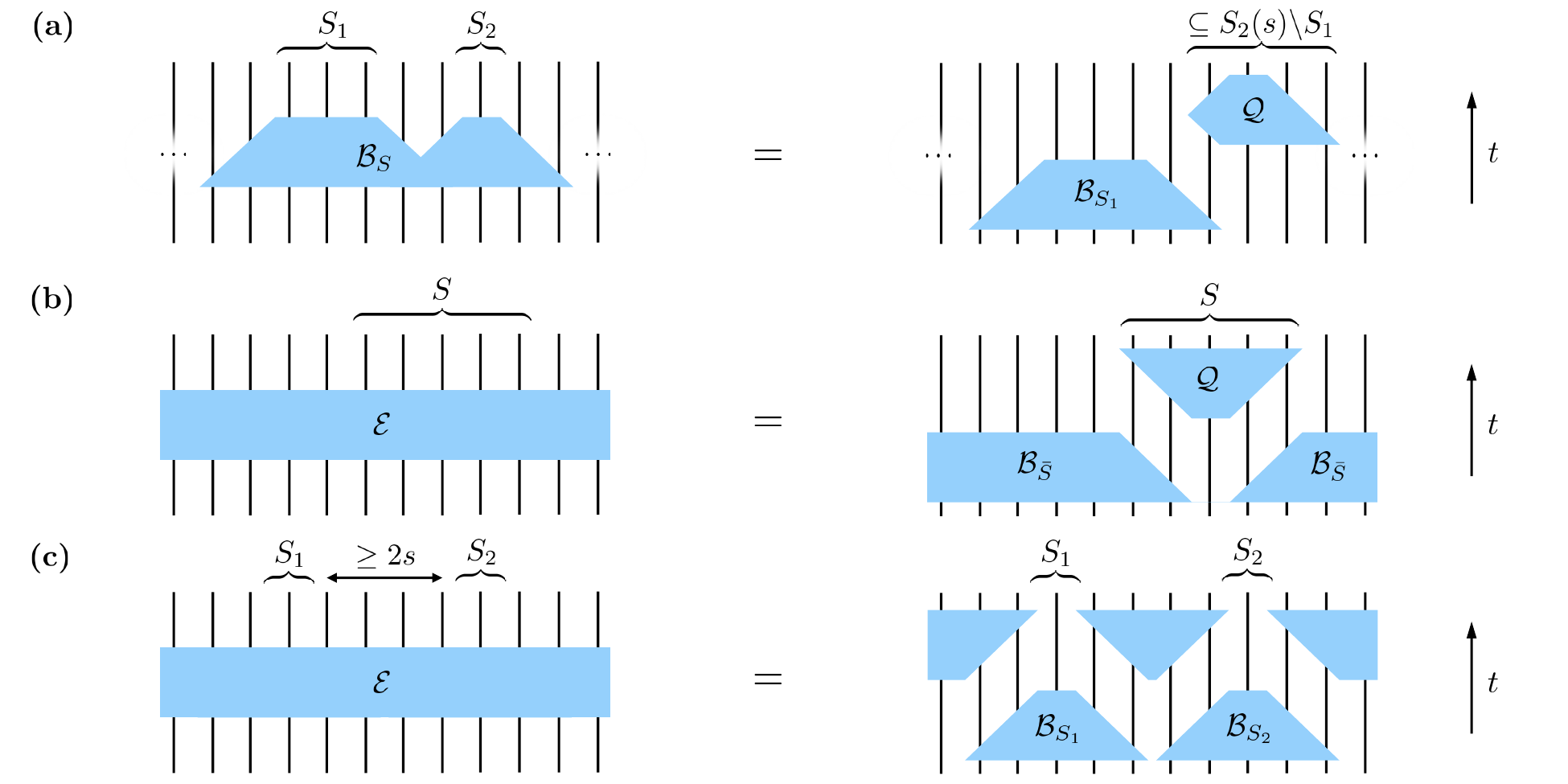}
    \caption{
    (a) Schematic for lightcone decomposition (Fact \ref{fac:LCD-1}). For any $S = S_1 S_2$, the backward lightcone has a decomposition $ \mathcal{B}_S = \mathcal{Q} \circ \mathcal{B}_{S_1} $ where $\mathrm{Supp}(\mathcal{Q}) \subseteq S_2(s) \backslash S_1$ with $s = c d$. We recall that $S(s)$ represents the region that extends $S$ by a distance $s$.
    (b) Schematic for a special case lightcone decomposition (Fact \ref{fac:LCD-2}). For any $S \subseteq \Lambda$, the overall circuit has a decomposition $ \mathcal{B}_S = \mathcal{Q} \circ \mathcal{B}_{\bar{S}} $ where $\mathrm{Supp}(\mathcal{Q}) \subseteq S$.
    (c) A trivial phase mixed state has no long-range correlation $\rho_{S_1 S_2} = \rho_{S_1} \otimes \rho_{S_2}$ if $\mathrm{dist}(S_1, S_2) \geq 2s$ (Fact \ref{fac:LCD-3}), because their backward lightcones $\mathcal{B}_{S_1}$ and $\mathcal{B}_{S_2}$ have no overlap. }
    \label{fig:lightcone}
\end{figure}

\subsection{Approximate Markovianity and local extendibility}
\label{sec:def_AM_LE}

In this subsection, we give the formal definition of the Approximate Markovianity and the local extendibility. 

\begin{definition}[Approximate Markovianity]
    \label{def:AM}
    Suppose a mixed state $\rho$ with a partition $\Lambda = A B C$, we say it is \textnormal{$\varepsilon_{\mathrm{AM}}$-Markovian}, if and only if there always exists a quantum channel $\Psi_{B \rightarrow B C}$ acting on $B \rightarrow B C$, called a \textnormal{local recovery map}, such that:
    \begin{equation}
        \| \rho_{A B C} - \Psi_{B \rightarrow B C} (\rho_{A B}) \|_1 \leq \varepsilon_{\mathrm{AM}} . \label{eq:BtoBC}
    \end{equation}
\end{definition}

For local extendibility, recall that it allows us to discard a small subsystem when we perform an extension map, which generalizes the traditional concept of a recovery map that does not discard any subsystems.

\begin{definition}[Local extendibility]
    \label{def:LE}
    Suppose a mixed state $\rho$ with a partition $\Lambda = A B C D E$, we say its local marginal state $\rho_{B C D E}$ is \textnormal{$\varepsilon_{\mathrm{LE}}$-locally extendible} with respect to $\rho$, if and only if there always exists a quantum channel $\Phi_{B E \rightarrow B C}$ acting on $B E \rightarrow B C$, called a \textnormal{local extension map}, such that:
    \begin{equation}
        \| \rho_{A B C} - \Phi_{B E \rightarrow B C} (\rho_{A B E}) \|_1 \leq \varepsilon_{\mathrm{LE}} . \label{eq:BEtoBC}
    \end{equation}
\end{definition}
We remark that Definition \ref{def:LE} is slightly different from Definition 1 of local extendibility presented in Ref.\,\cite{kim_2024_learning}. In Ref.\,\cite{kim_2024_learning}, the local extension map satisfies Eq\,.(\ref{eq:BEtoBC}) for any purification $\rho_{R B C D E}$ for $\rho_{B C D E}$, which is a stronger assumption than our Definition \ref{def:LE}.
% The other observation one may make is that $D$ does not appear in Eq.\,(\ref{eq:BEtoBC}) at all. It is natural to ask the question what the role of $D$ is in Definition \ref{def:LE}. We will comment on this question in Section \ref{sec:cover_scheme}.
% In this paper, we will show that the trivial phase mixed states satisfy the local extendibility in the sense of our Definition \ref{def:LE} (not the one defined in Ref.\,\cite{kim_2024_learning}). We will show in Section \ref{sec:learning_alg} that there also exists an efficient algorithm to learn the corresponding local extension $\Phi_{B E \rightarrow B C}$ for trivial phase mixed states.

\subsection{Mixed state phases via local reversibility}
\label{sec:MSP_via_LR}

We recall that in Section \ref{sec:problem_setting}, as an analog of the quantum pure state phase definition, we informally define two mixed states $\sigma$ and $\rho =\mathcal{E} (\sigma)$ to be in the same phase if and only if one can approximately recovery $\sigma$ by applying a sequence of channel gates $\{ \tilde{\mathcal{E}}_{\ell, x} \}$ that cancel the effect of channel gates $\{ \mathcal{E}_{\ell, x} \}$ gate-by-gate from $\ell = d$ to $\ell = 1$. Now, we formalize the idea of ``gate-by-gate reversibility'' by introducing the concept \textit{local reversibility}, initially introduced in Ref.\,\cite{sang_2025_statbility}. For technical simplicity, we adopt the version of local reversibility in Definition 3 of Ref.\,\cite{ma_2025_circuit}:

\begin{definition}[Local reversibility]
    \label{def:LR}
    Suppose any two mixed states $\sigma$ and $\rho$.
    We say $\sigma$ is \textnormal{locally reversible} from $\rho$, if and only if for any $\eta > 0$ and $\varepsilon_{\mathrm{LR}} > 0$, there exists $\mathcal{E} = \mathcal{E}_d \circ \cdots \circ \mathcal{E}_2 \circ \mathcal{E}_1$ where $\mathcal{E}_{\ell} = \prod_x \mathcal{E}_{\ell, x}$ and the support size of each local gate $\mathcal{E}_{\ell, x}$ is at most $c$, such that:
    \begin{enumerate}
        \item $\| \rho - \mathcal{E} (\sigma) \|_1 \leq \eta$,
        \item $\sigma$ is \textnormal{$\varepsilon_{\mathrm{LR}}$-locally reversible} from $\mathcal{E}(\sigma)$ along $\mathcal{E}$: for each $\mathcal{E}_{\ell, x}$, there exists another channel gate $\tilde{\mathcal{E}}_{\ell, x}$ with the same spatial support as $\mathcal{E}_{\ell, x}$ such that
        \begin{equation}
            \| \tilde{\mathcal{E}}_{\ell, x} \circ \mathcal{E}_{\ell, x} \circ \mathcal{E}_{\leq \ell - 1} (\sigma) - \mathcal{E}_{\leq \ell - 1} (\sigma) \|_1 \leq \varepsilon_{\mathrm{LR}},
        \end{equation}
        where $\mathcal{E}_{\leq \ell - 1} :=\mathcal{E}_{\ell - 1} \circ \cdots \circ \mathcal{E}_2 \circ \mathcal{E}_1$ represents the first $\ell - 1$ layers of $\mathcal{E}$.
    \end{enumerate}  
\end{definition}

We gave a visualization for the second condition of Definition \ref{def:LR} in Figure \ref{fig:LR}.
There are three comments that we can make on this definition.
First, the local reversibility in Definition \ref{def:LR} satisfies \textit{reflexivity}, \textit{symmetry}, and \textit{transitivity}. Namely, local reversibility is an equivalence relation of the quantum mixed states. The proofs of these three properties are routine, so we leave them in Appendix \ref{sec:equivalence_relation}.

Second, in condensed matter physics, the condition (2) in Definition \ref{def:LR} can be relaxed to the existence of a local channel gate $\tilde{\mathcal{E}}_{\ell, x}$ with an enlarged spatial support with respect to $\mathcal{E}_{\ell, x}$. However, using the reorganization techniques proposed in Ref.\,\cite{sang_2025_statbility}, one can conclude that whether the local spatial support of $\tilde{\mathcal{E}}_{\ell, x}$ is the same or larger than the support of  $\mathcal{E}_{\ell, x}$ does not change the definition of local reversibility. We give a more detailed explanation about this point in Appendix \ref{sec:cmi_decay} and Figure \ref{fig:reorganization}.

Third, we also emphasize here that the existence of such an $\tilde{\mathcal{E}}_{\ell, x}$ highly depends on the initial state $\sigma$ and the sub-circuit $\mathcal{E}_{\leq \ell - 1}$. 
In general, the channel gate $\mathcal{E}_{\ell, x}$ is not invertible. Thus $\tilde{\mathcal{E}}_{\ell, x} \circ \mathcal{E}_{\ell, x} \neq \mathcal{I}$, as a channel, is not the identity channel for any CPTP map $\tilde{\mathcal{E}}_{\ell, x}$. 
In Definition \ref{def:LR}, the channel $\tilde{\mathcal{E}}_{\ell, x} \circ \mathcal{E}_{\ell, x}$ is only an approximate stabilizer channel of the state $\mathcal{E}_{\leq \ell - 1} (\sigma)$.

Now, given this well-defined equivalence relation -- the local reversibility -- we can formally state its equivalence classes, namely the \textit{mixed state phases}.

\begin{definition}[Mixed state phases via local reversibility]
    \label{def:MSP_via_LR}
    For any two mixed states $\sigma$ and $\rho$, we say they are in the same mixed phase if and only if for any $\eta>0$ and $\varepsilon_{\mathrm{LR}} > 0$, there always exists a shallow circuit $\mathcal{E} = \mathcal{E}_d \circ \cdots \circ \mathcal{E}_2 \circ \mathcal{E}_1$ where $\mathcal{E}_{\ell} = \prod_x \mathcal{E}_{\ell, x}$ and each $\mathcal{E}_{\ell, x}$ has a local support, such that $\| \rho - \mathcal{E} (\sigma) \|_1 \leq \eta$ and $\sigma$ is $\varepsilon_{\mathrm{LR}}$-locally reversible from $\mathcal{E}(\sigma)$ along the circuit $\mathcal{E}$.
  
    Especially, we say $\rho$ is in the trivial phase if and only if $\rho$ and $\rho_0 = \ket{0}\!\bra{0}^{\otimes n}$ are in the same mixed state phase.
\end{definition}

\textbf{Remark}. Without loss of generality, we restrict our learning target by adopting the following convention in Definition 3 of Ref.\,\cite{ma_2025_circuit} to simplify our derivation: for the trivial phase mixed state $\rho$ to be learned, we can always assume it is perfectly prepared by some shallow circuit (namely we assume the preparation error $\eta = 0$)
\begin{equation}
    \rho = \mathcal{E}(\sigma),
\end{equation}
and $\sigma$ is $\varepsilon_{\mathrm{LR}}$-locally reversible from $\rho$.
This does not change the essence of the problem we are solving. In fact, for any given $\varepsilon>0$, we can take $\eta = \varepsilon/2$ such that $\| \rho - \mathcal{E} (\sigma) \|_1 \leq \varepsilon/2$ for the corresponding $\mathcal{E}$, and $\sigma$ is $\varepsilon_{\mathrm{LR}}$-locally reversible from $\rho$, then efficiently learning and generating a shallow circuit $\mathcal{W}$ with $\| \mathcal{E} (\sigma) - \mathcal{W} (\sigma) \|_1 \leq \varepsilon/2$ immediately implies $\| \rho - \mathcal{W} (\sigma) \|_1 \leq \varepsilon$ by the triangle inequality.
This means that setting $\rho = \mathcal{E}(\sigma)$ does not change the scaling of time and sample complexity.

Moreover, in quantum many-body physics, trivial phase mixed states are usually defined in a much stronger way: $\rho$ is in the trivial phase if and only if $\rho$ can be prepared from $\ket{0} \! \bra{0}^{\otimes n}$ using a local shallow channel circuits, and during the preparation the spatial CMI of the states for the partition in Figure \ref{fig:partition}a must always decay exponentially as $I (A : C | B) \sim e^{- r / \xi}$ for some constant characteristic length $\xi$.
There are many examples of trivial phase mixed states in quantum many-body physics, for example, high-temperature quantum Gibbs states \cite{chen_2025_quantum} and quantum error-correcting codes (like toric codes) influenced by noise over the noise threshold \cite{sang_2025_statbility}.
In Appendix \ref{sec:cmi_decay}, we will show that the mixed state with CMI decay always imply local reversibility.
As a result, these examples are also classified as trivial phase mixed states based on our Definition \ref{def:MSP_via_LR}, since our definition is broader than that in quantum many-body physics.

\subsection{Local inversion within a phase}
\label{sec:LI}

We are going to end this section with an important concept called \textit{local inversion}, arising from the local reversibility property. 
This local inversion technique will be quite useful when we prove approximate Markovianity by constructing a recovery map (Section \ref{sec:proof_AM}), and prove local extendibility by constructing a local extension map (Section \ref{sec:proof_LE}).

In Definition \ref{def:MSP_via_LR}, the condition $\tilde{\mathcal{E}}_{\ell, x} \circ \mathcal{E}_{\ell, x} \circ \mathcal{E}_{\leq \ell - 1} (\sigma) \approx \mathcal{E}_{\leq \ell - 1} (\sigma)$ can be further relaxed if there are more gates in $\{ \mathcal{E}_{\ell', x'} \}_{\ell' \geq \ell}$ being acted, as long as $\mathcal{E}_{\ell', x'}$ acts exactly outside $\mathrm{Supp} (\mathcal{E}_{\ell, x})$. We formalize this observation by the following useful lemma that we will use later:

\begin{figure}
    \centering
    \includegraphics[width=\linewidth]{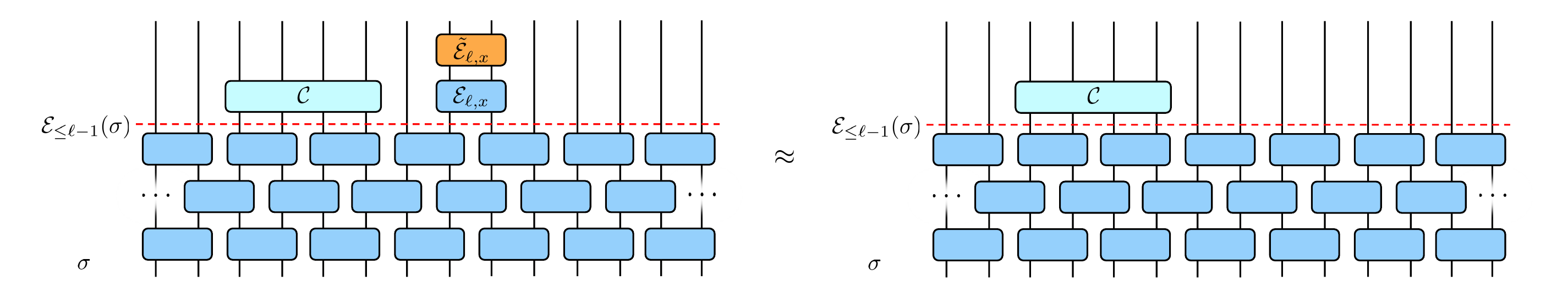}
    \caption{
    Schematic for Lemma \ref{lem:LR-modify}. Local reversibility still holds for any extra channel $\mathcal{C}$ acting on $\mathcal{E}_{\leq \ell-1} (\sigma)$, as long as with $\mathrm{Supp}(\mathcal{C}) \cap \mathrm{Supp}(\mathcal{E}_{\ell, x}) = \varnothing$. }
    \label{fig:LR-relaxed}
\end{figure}

\begin{lemma}[Relaxed local reversibility]
    \label{lem:LR-modify}
    Suppose $\| \tilde{\mathcal{E}}_{\ell, x} \circ \mathcal{E}_{\ell, x} \circ \mathcal{E}_{\leq \ell - 1} (\sigma) - \mathcal{E}_{\leq \ell - 1} (\sigma) \|_1 \leq \varepsilon_{\mathrm{LR}}$ for any $\varepsilon_{\mathrm{LR}} > 0$, where $\mathcal{E}_{\leq \ell - 1} := \mathcal{E}_{\ell - 1} \circ \cdots \circ \mathcal{E}_2 \circ \mathcal{E}_1$ represents the first $\ell - 1$ layers of $\mathcal{E}$. Consider any channel $\mathcal{C}$ acting on $\Lambda \backslash \mathrm{Supp} (\mathcal{E}_{\ell, x})$ % = \Lambda \backslash \mathrm{Supp} (\tilde{\mathcal{E}}_{\ell, x})$
    \begin{equation}
        \| \tilde{\mathcal{E}}_{\ell, x} \circ \mathcal{E}_{\ell, x} \circ \mathcal{E}' (\sigma) -\mathcal{E}' (\sigma) \|_1 \leq \varepsilon_{\mathrm{LR}},
    \end{equation}
    where $\mathcal{E}' = \mathcal{C} \circ \mathcal{E}_{\leq \ell - 1}$ (see Figure \ref{fig:LR-relaxed}).
\end{lemma}

\begin{proof}
    We prove directly
    \begin{eqnarray}
        \| \tilde{\mathcal{E}}_{\ell, x} \circ \mathcal{E}_{\ell, x} \circ \mathcal{E}' (\sigma) -\mathcal{E}' (\sigma) \|_1 & = & \| \tilde{\mathcal{E}}_{\ell, x} \circ \mathcal{E}_{\ell, x} \circ \mathcal{C} \circ \mathcal{E}_{\leq \ell - 1} (\sigma) -\mathcal{C} \circ \mathcal{E}_{\leq \ell - 1} (\sigma) \|_1 \nonumber\\
        & \overset{\text{(i)}}{=} & \| \mathcal{C} (\tilde{\mathcal{E}}_{\ell, x} \circ \mathcal{E}_{\ell, x} \circ \mathcal{E}_{\leq \ell - 1} (\sigma)) -\mathcal{C} (\mathcal{E}_{\leq \ell - 1} (\sigma)) \|_1 \nonumber\\
        & \overset{\text{(ii)}}{\leq} & \| \tilde{\mathcal{E}}_{\ell, x} \circ \mathcal{E}_{\ell, x} \circ \mathcal{E}_{\leq \ell - 1} (\sigma) -\mathcal{E}_{\leq \ell - 1} (\sigma) \|_1 \leq \varepsilon_{\mathrm{LR}}, 
    \end{eqnarray}
    where equality (i) follows from the fact that $\tilde{\mathcal{E}}_{\ell, x} \circ \mathcal{E}_{\ell, x}$ and $\mathcal{C}$ commute, and inequality (ii) follows from the contractivity of the CPTP map $\mathcal{C}$.
\end{proof}

Now, we define and construct the local inversion: for any region $S$, there always exists a channel acting on $S$ which approximately cancels part of the gates such that the overall channel circuit is effectively the backward lightcone of $\bar{S} = \Lambda \backslash S$ (see a visualization in Figure \ref{fig:LI}). From now on, we will always denote
\begin{equation}
    \varepsilon_{\mathrm{LI}} = n d \cdot \varepsilon_{\mathrm{LR}}.
\end{equation}

\begin{lemma}[Local inversion]
    \label{lem:LI}
    Suppose mixed state $\sigma$ and $\rho =\mathcal{E} (\sigma)$, and $\sigma$ is $\varepsilon_{\mathrm{LR}}$-locally reversible from $\rho$ along the shallow channel circuit $\mathcal{E}$. Given a local region $S$. Then, there exists a quantum channel $\mathcal{P}_S$ acting on $S$ such that
    \begin{equation}
        \| \mathcal{P}_S (\rho) -\mathcal{B}_{\bar{S}} (\sigma) \|_1 = \| \mathcal{P}_S \circ \mathcal{E} (\sigma) -\mathcal{B}_{\bar{S}} (\sigma) \|_1 \leq \varepsilon_{\mathrm{LI}},
    \end{equation}
    where $\mathcal{B}_{\bar{S}}$ is the backward lightcone of $\bar{S} = \Lambda \backslash S$.

    Especially for $S = \Lambda$, there is a channel $\mathcal{P}$ acting on the whole $\Lambda$, such that $\| \mathcal{P} (\rho) - \sigma \|_1 \leq \varepsilon_{\mathrm{LI}}$.
\end{lemma}

\begin{figure}
    \centering
    \includegraphics[width=\linewidth]{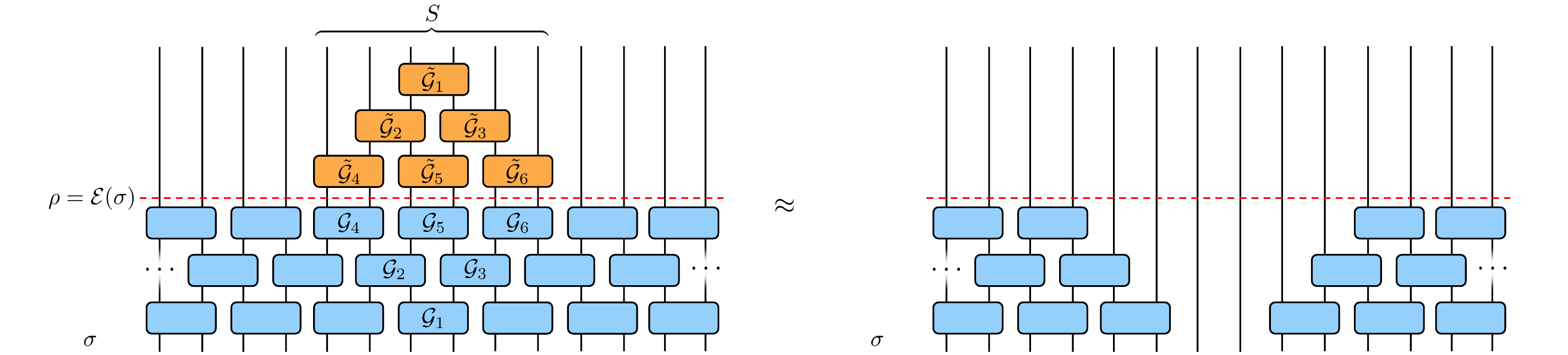}
    \caption{
    Schematic for local inversion (Lemma \ref{lem:LI}). For any $S \subseteq \lambda$, there exists a channel $\mathcal{P}_S$ acting on $S$, such that $\mathcal{P}_S (\rho) \approx \mathcal{B}_{\bar{S}} (\sigma)$.
    To construct $\mathcal{P}_S$, we first assign an ordering of indices to gates in $\mathcal{B}_S \backslash \mathcal{B}_{\bar{S}}$ -- denoted as $\{\tilde{\mathcal{G}}_m\}_{m \in [N]}$ (where $N=6$ in this schematic) -- based on the ascending order of the layer indices $\ell$, and then we act the local recovery gates $\{\tilde{\mathcal{G}}_m\}_{m \in [N]}$ (the orange channel gates) based on the descending order of the layer indices $\ell$.
    }
    \label{fig:LI}
\end{figure}

\begin{proof}
    According to Fact \ref{fac:LCD-2}, let the backward lightcone $\mathcal{B}_{\bar{S}} =\mathcal{B}_{\bar{S}, d} \circ \cdots \circ \mathcal{B}_{\bar{S}, 1}$. Then we can reorganize $\mathcal{E}$ into the following lightcone decomposition $\mathcal{E}=\mathcal{Q} \circ \mathcal{B}_{\bar{S}}$, where $\mathcal{Q} = \mathcal{B}_S \backslash \mathcal{B}_{\bar{S}}= \mathcal{Q}_d \circ \cdots \circ \mathcal{Q}_1$ with $\mathcal{Q}_{\ell} = \prod_{\mathcal{E}_{\ell, x} \notin \mathcal{B}_{\bar{S}, \ell}} \mathcal{E}_{\ell, x}$ and $\mathrm{Supp} (\mathcal{Q}) \subseteq S$.
  
    Let $N = | \mathcal{Q} |$. We can upper bound it by $N \leq | \mathcal{E} | \leq (n/c) \cdot d \leq n d$. Now, we assign an ordering for the channel gates inside $\mathcal{G}_1, \mathcal{G}_2, \cdots, \mathcal{G}_N \in \mathcal{Q}$, such that: for any $\mathcal{G}_{m_1} =\mathcal{E}_{\ell_1, x_1}$ and $\mathcal{G}_{m_2} =\mathcal{E}_{\ell_2, x_2}$, we always have $m_1 \leq m_2$ if and only if $\ell_1 \leq \ell_2$. In other words, the circuit $\mathcal{Q}$ is re-organized into $\mathcal{Q}=\mathcal{G}_N \circ \cdots \circ \mathcal{G}_2 \circ \mathcal{G}_1$. Furthermore, we can rewrite $\mathcal{E}$ into $\mathcal{E}=\mathcal{G}_N \circ \cdots \circ \mathcal{G}_2 \circ \mathcal{G}_1 \circ \mathcal{B}_{\bar{S}}$ by the lightcone decomposition.
  
    According to the local reversibility property, there exists a set of channel gates $\{ \tilde{\mathcal{G}}_m \}_{m \in [N]}$ such that for any $\mathcal{G}_m =\mathcal{E}_{\ell, x}$, we have $\mathrm{Supp} (\mathcal{G}_m) = \mathrm{Supp} (\tilde{\mathcal{G}}_m)$ and
    \begin{equation}
        \| \tilde{\mathcal{G}}_m \circ \mathcal{G}_m \circ \mathcal{E}_{\leq \ell - 1} (\sigma) -\mathcal{E}_{\leq \ell - 1} (\sigma) \|_1 \leq \varepsilon_{\mathrm{LR}},
    \end{equation}
    where $\mathcal{E}_{\leq \ell - 1} =\mathcal{E}_{\ell - 1} \circ \cdots \circ \mathcal{E}_2 \circ \mathcal{E}_1$ represents the first $\ell - 1$ layers of $\mathcal{E}$ and $\ell$ is the layer index of $\mathcal{G}_m =\mathcal{E}_{\ell, x}$.
  
    We let $\mathcal{P}_S = \tilde{\mathcal{G}}_1 \circ \tilde{\mathcal{G}}_2 \circ \cdots \circ \tilde{\mathcal{G}}_N$ with $\mathrm{Supp} (\mathcal{P}_S) = \cup_{m = 1}^N \mathrm{Supp} (\mathcal{G}_m) = \mathrm{Supp} (\mathcal{Q}) \subseteq S$. Then, by substituting $\rho =\mathcal{E} (\sigma) =\mathcal{G}_N \circ \cdots \circ \mathcal{G}_2 \circ \mathcal{G}_1 \circ \mathcal{B}_{\bar{S}} (\sigma)$, we can bound
    \begin{eqnarray}
        & &\| \mathcal{P}_S (\rho) -\mathcal{B}_{\bar{S}} (\sigma) \|_1 \nonumber\\
        & = & \| \tilde{\mathcal{G}}_1 \circ \cdots \circ \tilde{\mathcal{G}}_N \circ \mathcal{G}_N \circ \cdots \circ \mathcal{G}_1 \circ \mathcal{B}_{\bar{S}} (\sigma) -\mathcal{B}_{\bar{S}} (\sigma) \|_1 \nonumber\\
        & \overset{\text{(i)}}{\leq} & \sum_{m = 1}^N \| \tilde{\mathcal{G}}_1 \circ \cdots \circ \tilde{\mathcal{G}}_m \circ \mathcal{G}_m \circ \cdots \circ \mathcal{G}_1 \circ \mathcal{B}_{\bar{S}} (\sigma) - \tilde{\mathcal{G}}_1 \circ \cdots \circ \tilde{\mathcal{G}}_{m - 1} \circ \mathcal{G}_{m - 1} \circ \cdots \circ \mathcal{G}_1 \circ \mathcal{B}_{\bar{S}} (\sigma) \|_1 \nonumber\\
        & \overset{\text{(ii)}}{\leq} & \sum_{m = 1}^N \| \tilde{\mathcal{G}}_m \circ \mathcal{G}_m (\mathcal{G}_{m - 1} \circ \cdots \circ \mathcal{G}_1 \circ \mathcal{B}_{\bar{S}} (\sigma)) -\mathcal{G}_{m - 1} \circ \cdots \circ \mathcal{G}_1 \circ \mathcal{B}_{\bar{S}} (\sigma) \|_1 \nonumber\\
        & \overset{\text{(iii)}}{\leq} & \sum_{m \in [N], \mathcal{G}_m =\mathcal{E}_{\ell, x}} \| \tilde{\mathcal{G}}_m \circ \mathcal{G}_m (\mathcal{E}_{\leq \ell - 1} (\rho_0)) -\mathcal{E}_{\leq \ell - 1} (\rho_0) \|_1 \leq N \cdot \varepsilon_{\mathrm{LR}} \leq n d \cdot \varepsilon_{\mathrm{LR}} = \varepsilon_{\mathrm{LI}} . 
    \end{eqnarray}
    The first inequality (i) follows from the triangle inequality of trace norm, and the second inequality (ii) follows from the contractivity of the CPTP map $\tilde{\mathcal{G}}_1 \circ \cdots \circ \tilde{\mathcal{G}}_{m - 1}$. In the third inequality (iii), $\ell$ is the layer index of $\mathcal{G}_m =\mathcal{E}_{\ell, x}$, and we can always decompose $\mathcal{G}_{m - 1} \circ \cdots \circ \mathcal{G}_1 \circ \mathcal{B}_{\bar{S}} (\sigma) =\mathcal{C} \circ \mathcal{E}_{\leq \ell - 1}$ for some $\mathcal{C}$ acting outside $\mathrm{Supp} (\mathcal{G}_m)$. Then the third inequality (iii) is from Lemma \ref{lem:LR-modify} immediately. 
\end{proof}

% We refer to Figure \ref{fig:LI} as a schematic for the construction of $\mathcal{P}_S = \tilde{\mathcal{G}}_1 \circ \tilde{\mathcal{G}}_2 \circ \cdots \circ \tilde{\mathcal{G}}_N$ in the proof.

\section{Existence of Local Recovery Maps}
\label{sec:proof_AM}

\begin{theorem}[Approximate Markovianity]
    \label{thm:AM}
    Suppose a trivial phase mixed state $\rho =\mathcal{E} (\rho_0)$, and $\rho_0$ is $\varepsilon_{\mathrm{LR}}$-locally reversible from $\rho$ along the shallow channel circuit $\mathcal{E}$. Now, consider a partition $\Lambda = A B C$ with $\mathrm{dist} (A, C) \geq 2 s$.
    Then $\rho$ is \textnormal{$\varepsilon_{\mathrm{LI}}$-approximately Markovian}, namely there exists a recovery map $\Psi$ acting on $B C$ such that $\left\| \rho - \Psi \left( \rho_{A B} \otimes \ket{0} \! \bra{0}_C \right) \right\|_1 \leq \varepsilon_{\mathrm{LI}}$. This naturally induces a CPTP map $\Psi_{B \rightarrow B C} (X_B) = \Psi \left( X_B \otimes \ket{0} \! \bra{0}_C \right)$ such that
    \begin{equation}
        \| \rho - \Psi_{B \rightarrow B C} (\rho_{A B}) \|_1 \leq \varepsilon_{\mathrm{LI}} .
    \end{equation}
    Symmetrically, there also exists a recovery map $\Psi_{B \rightarrow B A}$ such that $\| \rho - \Psi_{B \rightarrow B A} (\rho_{B C}) \|_1 \leq \varepsilon_{\mathrm{LI}}$.
\end{theorem}

\begin{figure}
    \centering
    \includegraphics[width=\linewidth]{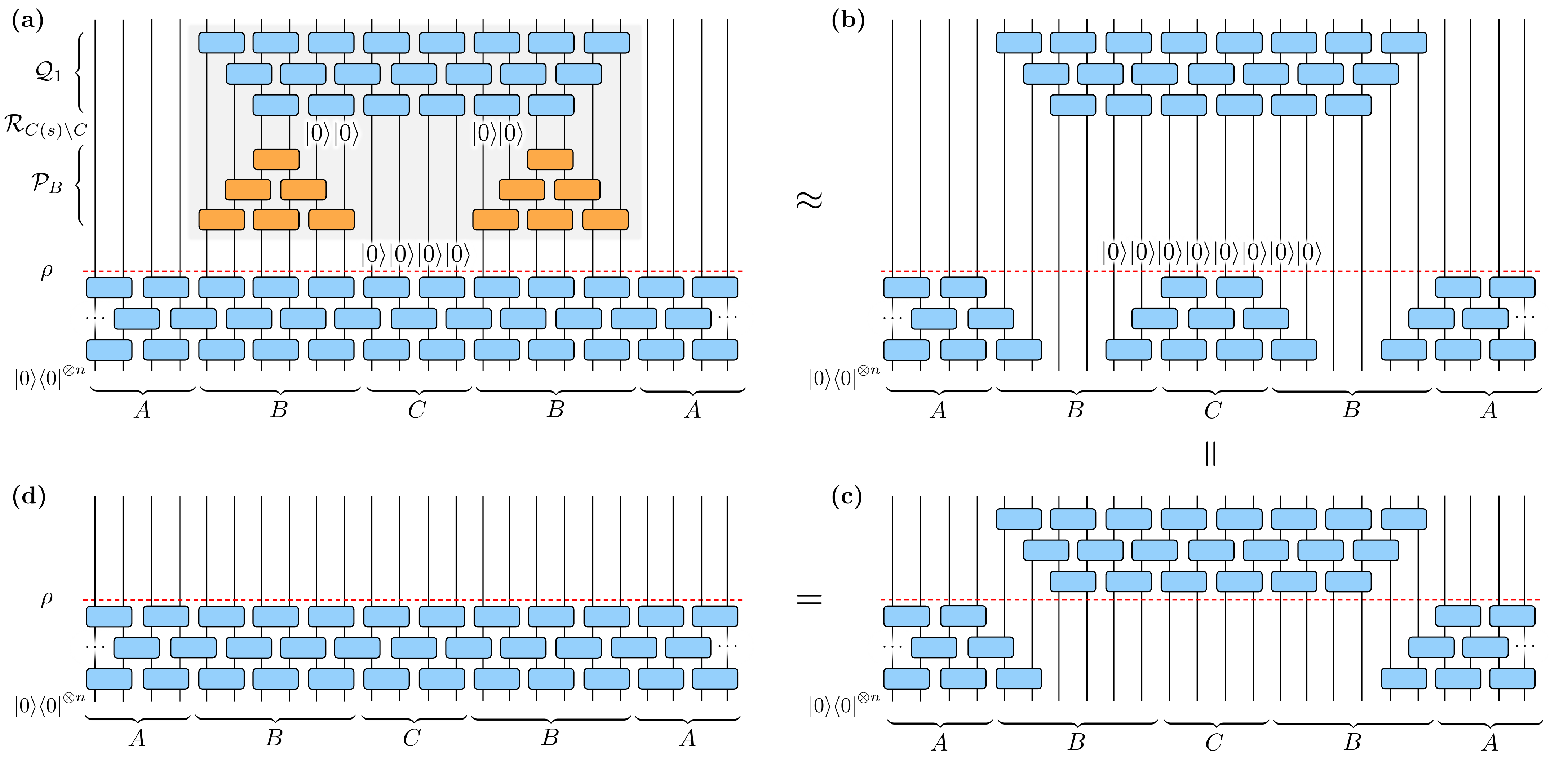}
    \caption{
    Schematic for the proof of the \textit{existence} of local recovery maps (Theorem \ref{thm:AM}).
    We show the approximate Markovianity by the explicit construction of the local recovery $\Psi =\mathcal{Q}_1 \circ \mathcal{R}_{C (s) \backslash C} \circ \mathcal{P}_B$ (light gray region).
    The step from (a) to (b) is from the local inversion $\mathcal{P}_B (\rho) \approx \mathcal{B}_{A C} (\rho_0)$.
    The step from (b) to (c) is from the reset channel property $\mathcal{R}_{C (s)} \circ \mathcal{B}_C = \mathcal{R}_{C (s)}$ for $\mathrm{Supp}(\mathcal{B}_{C}) \subseteq C(s)$.
    The step from (c) to (d) is from the backward lightcone decomposition $ \mathcal{E} = \mathcal{Q}_1 \circ \mathcal{B}_{A} $.
    }
    \label{fig:AM}
\end{figure}

\begin{proof}
    Consider the backward lightcones $\mathcal{B}_A$ and $\mathcal{B}_C$. Because $\mathrm{Supp} (\mathcal{B}_A) \subseteq A (s)$, $\mathrm{Supp} (\mathcal{B}_C) \subseteq C (s)$ and $\mathrm{dist} (A, C) \geq 2 s$, it means $\mathrm{Supp} (\mathcal{B}_A) \cap \mathrm{Supp} (\mathcal{B}_C) = \varnothing$ and $\mathcal{B}_{A C} =\mathcal{B}_A \circ \mathcal{B}_C =\mathcal{B}_C \circ \mathcal{B}_A$.
  
    By Fact \ref{fac:LCD-2}, we have lightcone decompositions $\mathcal{E}=\mathcal{Q}_1 \circ \mathcal{B}_A$ with $\mathrm{Supp} (\mathcal{Q}_1) \subseteq B C$ and $\mathcal{E}=\mathcal{Q}_2 \circ \mathcal{B}_{A C}$ with $\mathrm{Supp} (\mathcal{Q}_2) \subseteq B$. According to Lemma \ref{lem:LI}, let $\mathcal{P}_B$ acting on $B$ such that
    \begin{equation}
        \| \mathcal{P}_B (\rho) -\mathcal{B}_{A C} (\rho_0) \|_1 \leq \varepsilon_{\mathrm{LI}} . \label{eq:P_B}
    \end{equation}
    $\mathrm{Supp} (\mathcal{R}_{C (s) \backslash C}) \subseteq C (s)$ and $\mathrm{Supp} (\mathcal{B}_A) \subseteq A (s)$ implies that $\mathrm{Supp} (\mathcal{R}_{C (s) \backslash C}) \cap \mathrm{Supp} (\mathcal{B}_A) = \varnothing$, due to the fact $\mathrm{dist} (A, C) \geq 2 s$.
  
    We construct (see Figure \ref{fig:AM}):
    \begin{equation}
        \Psi =\mathcal{Q}_1 \circ \mathcal{R}_{C (s) \backslash C} \circ \mathcal{P}_B , 
    \end{equation}
    where $\Psi$ is a channel circuit consisting of $(2d + 1)$ layers of $c$-local gates.
    
    We also notice that $\rho_{A B} \otimes \ket{0} \! \bra{0}_C =\mathcal{R}_C (\rho)$. Then, we can plug $\Psi =\mathcal{Q}_1 \circ \mathcal{R}_{C (s) \backslash C} \circ \mathcal{P}_B$ and bound the recovery error
    \begin{eqnarray}
        &  & \left\| \rho - \Psi \left( \rho_{A B} \otimes \ket{0} \! \bra{0}_C \right) \right\|_1 \nonumber\\
        & = & \| \rho -\mathcal{Q}_1 \circ \mathcal{R}_{C (s) \backslash C} \circ \mathcal{P}_B \circ \mathcal{R}_C (\rho) \|_1 \overset{\text{(i)}}{=} \| \rho -\mathcal{Q}_1 \circ \mathcal{R}_{C (s)} \circ \mathcal{P}_B (\rho) \|_1 \nonumber\\
        & \overset{\text{(ii)}}{\leq} & \| \rho -\mathcal{Q}_1 \circ \mathcal{R}_{C (s)} \circ \mathcal{B}_{A C} (\rho_0) \|_1 + \| \mathcal{Q}_1 \circ \mathcal{R}_{C (s)} \circ \mathcal{B}_{A C} (\rho_0) -\mathcal{Q}_1 \circ \mathcal{R}_{C (s)} \circ \mathcal{P}_B (\rho) \|_1 \nonumber\\
        & \overset{\text{(iii)}}{\leq} & \| \rho -\mathcal{Q}_1 \circ \mathcal{R}_{C (s)} \circ \mathcal{B}_A \circ \mathcal{B}_C (\rho_0) \|_1 + \| \mathcal{B}_{A C} (\rho_0) -\mathcal{P}_B (\rho) \|_1 \nonumber\\
        & \overset{\text{(iv)}}{\leq} & \| \rho -\mathcal{Q}_1 \circ \mathcal{B}_A \circ \mathcal{R}_{C (s)} \circ \mathcal{B}_C (\rho_0) \|_1 + \varepsilon_{\mathrm{LI}} \nonumber\\
        & \overset{\text{(v)}}{=} & \| \rho -\mathcal{E} (\rho_0) \|_1 + \varepsilon_{\mathrm{LI}} = \varepsilon_{\mathrm{LI}} . 
    \end{eqnarray}
    Here, equality (i) is from the commutability between $\mathcal{P}_B$ and $\mathcal{R}_C$, and the reset channel property $\mathcal{R}_{C (s) \backslash C} \circ \mathcal{R}_C =\mathcal{R}_{C (s)}$; inequality (ii) is from the triangle inequality; inequality (iii) is from the fact that $\mathcal{B}_{A C} =\mathcal{B}_A \circ \mathcal{B}_C$ and the contractivity of CPTP map $\mathcal{Q}_1 \circ \mathcal{R}_{C (s)}$; inequality (iv) is from the commutability between $\mathcal{R}_{C (s)}$ and $\mathcal{B}_A$, and the Eq.\,(\ref{eq:P_B}) $\| \mathcal{P}_B (\rho) -\mathcal{B}_{A C} (\rho_0) \|_1 \leq \varepsilon_{\mathrm{LI}}$; equality (v) is from the lightcone decomposition $\mathcal{E}=\mathcal{Q}_1 \circ \mathcal{B}_A$, the relation $\mathcal{R}_{C (s)} \circ \mathcal{B}_C =\mathcal{R}_{C (s)}$ for $C (s) \supseteq C$ and the fact $\mathcal{R}_{C (s)} (\rho_0) = \rho_0$ for all-zero state. This completes the proof. 
\end{proof}

\begin{corollary}
    \label{cor:AM}
    Suppose a trivial phase mixed state $\rho =\mathcal{E} (\rho_0)$, and $\rho_0$ is $\varepsilon_{\mathrm{LR}}$-locally reversible from $\rho$ along the shallow channel circuit $\mathcal{E}$. Now, consider a partition $\Lambda = A B C$ with $\mathrm{dist} (A, C) \geq 2 s$. For any partition $A = A_1 A_2$, there always exists a channel acting on $B \rightarrow B C$ such that
    \begin{equation}
        \| \rho_{A_1 B C} - \Psi_{B \rightarrow B C} (\rho_{A_1 B}) \|_1 \leq \varepsilon_{\mathrm{LI}} .
    \end{equation}
\end{corollary}

\begin{proof}
    We take the partial trace $\mathrm{Tr}_{A_2}$, and the corollary is implied by the contractivity of $\mathrm{Tr}_{A_2}$.
\end{proof}

\section{Existence of Local Extension Maps}
\label{sec:proof_LE}

\begin{theorem}[Local extendibility]
    \label{thm:LE}
    Suppose a trivial phase mixed state $\rho =\mathcal{E} (\rho_0)$, and $\rho_0$ is $\varepsilon_{\mathrm{LR}}$-locally reversible from $\rho$ along the shallow channel circuit $\mathcal{E}$. Now, consider a partition $\Lambda = A B C D E$ with $\mathrm{dist} (A, C) \geq 2 s$ and $\mathrm{dist} (B, D) \geq 2 s$. Then $\rho$ is $\varepsilon_{\mathrm{LI}}$-locally extendible, namely there exists a recovery map $\Phi$ acting on $B C D E$, such that $\left\| \rho_{A B C} - \mathrm{Tr}_{D E}(\Phi ( \rho_{A B E} \otimes \ket{0} \! \bra{0}_{C D}) ) \right\|_1 \leq \varepsilon_{\mathrm{LI}}$. This naturally induce a CPTP map $\Phi_{B E \rightarrow B C} (X_{B E}) = \mathrm{Tr}_{D E}( \Phi ( X_{B E} \otimes \ket{0} \! \bra{0}_{C D} ))$ such that
    \begin{equation}
        \| \rho_{A B C} - \Phi_{B E \rightarrow B C} (\rho_{A B E}) \|_1 \leq \varepsilon_{\mathrm{LI}} . \label{eq:def_LE}
    \end{equation}
\end{theorem}

\begin{figure}
    \centering
    \includegraphics[width=0.95\linewidth]{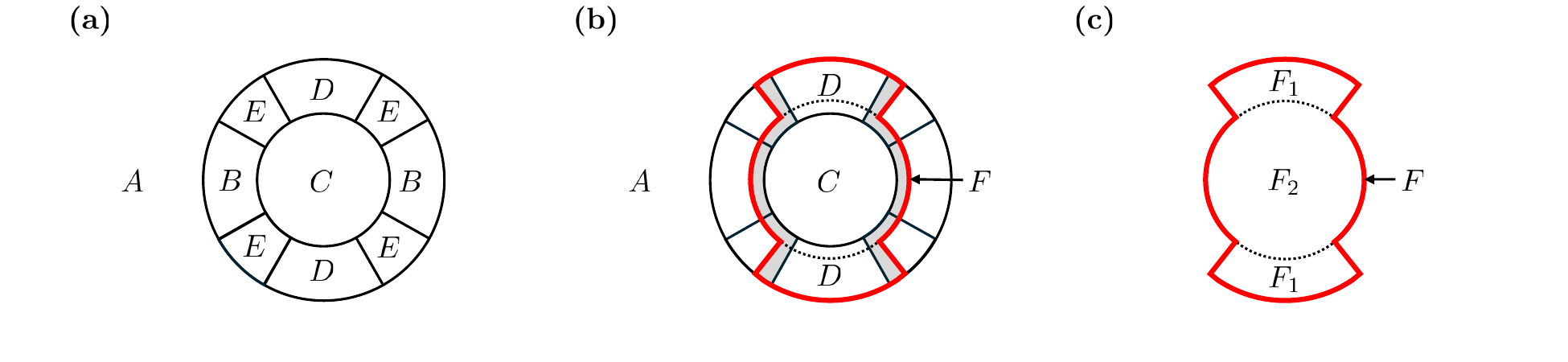}
    \caption{
    Schematic of the partition used in the proof of the existence of the local extension map (Theorem \ref{thm:LE}).
    (a) The pentapartition $\Lambda = A B C D E$ with $\mathrm{dist} (A, C) \geq 2 s$ and $\mathrm{dist} (B, D) \geq 2 s$.
    (b) The region $F$ is defined as $F = C D (s) \backslash A$ (surrounded by the red boundary). It can be obtained by extending $CD$ outwards by $CD(s) \backslash AC$ (gray region).
    (c) The partition $F = F_1 F_2$ where $F_1 = D (s) \backslash (A \cup C (s))$ and $F_2 = C (s)$. The boundary between $F_1$ and $F_2$ is depicted by dashed lines.
    }
    \label{fig:F}
\end{figure}
    
\begin{proof}
    Set $S = A B C$, $S_1 = A$ and $S_2 = B C$ in Fact \ref{fac:LCD-1}, we have
    \begin{equation}
        \mathcal{B}_{A B C} =\mathcal{Q}_1 \circ \mathcal{B}_A .
    \end{equation}
    where $\mathcal{Q}_1 =\mathcal{B}_{B C} \backslash \mathcal{B}_A$ with $\mathrm{Supp} (\mathcal{Q}_1) \subseteq B C (s) \backslash A$. Set $S = A C D$, \ $S_1 = A$ and $S_2 = C D$ in Fact \ref{fac:LCD-1}. We have
    \begin{equation}
        \mathcal{B}_{A C D} =\mathcal{Q}_2 \circ \mathcal{B}_A,
    \end{equation}
    where $\mathcal{Q}_2 =\mathcal{B}_{C D} \backslash \mathcal{B}_A$ with $\mathrm{Supp} (\mathcal{Q}_2) \subseteq C D (s) \backslash A$.
  
    According to Lemma \ref{lem:LI}, let $\mathcal{P}_{B E}$ acting on $B E$ such that
    \begin{equation}
        \| \mathcal{P}_{B E} (\rho) -\mathcal{B}_{A C D} (\rho_0) \|_1 \leq \varepsilon_{\mathrm{LI}} .
    \end{equation}
    
    Define $F = C D (s) \backslash A$. We also make partition $F = F_1 F_2$ where $F_1 = D (s) \backslash (A \cup C (s))$ and $F_2 = C (s)$. Since $\mathrm{dist} (A, C) \geq 2 s$ and $\mathrm{dist} (B, D) \geq 2 s$, we have the following set relations (see schematics in Figure \ref{fig:F})
    \begin{eqnarray}
        \mathrm{Supp} (\mathcal{Q}_2) \subseteq F, & \Rightarrow & \mathcal{R}_F \circ \mathcal{Q}_2 =\mathcal{R}_F, \\
        \mathrm{Supp} (\mathcal{Q}_1) \cap F_1 = \varnothing, & \Rightarrow & \mathcal{R}_{F_1} \circ \mathcal{Q}_1 =\mathcal{Q}_1 \circ \mathcal{R}_{F_1}, \\
        \mathrm{Supp} (\mathcal{B}_A) \cap F_2 = \varnothing, & \Rightarrow & \mathcal{R}_{F_2} \circ \mathcal{B}_A =\mathcal{B}_A \circ \mathcal{R}_{F_2}, \\
        F_1 \subseteq D E & \Rightarrow & \mathrm{Tr}_{D E} \circ \mathcal{R}_{F_1} = \mathrm{Tr}_{D E} . 
    \end{eqnarray}

    We construct (see Figure \ref{fig:LE}):
    \begin{equation}
        \Phi = \mathcal{Q}_1 \circ \mathcal{R}_{F \backslash C D} \circ \mathcal{P}_{B E} .
    \end{equation}
    where $\Phi$ is a channel circuit consisting of $(2d + 1)$ layers of $c$-local gates.

    \begin{figure}
        \centering
        \includegraphics[width=\linewidth]{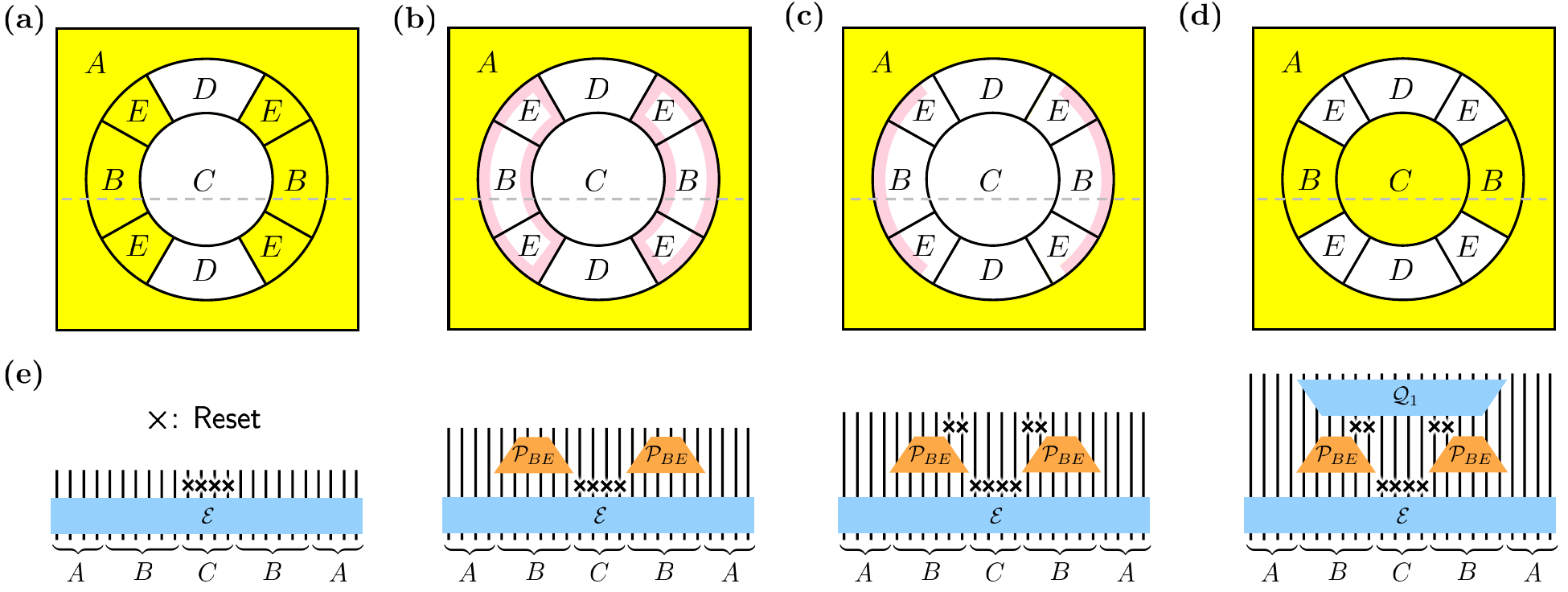}
        \caption{
        Schematic for the proof of the \textit{existence} of local extension maps (Theorem \ref{thm:LE}).
        We show the local extendibility by the explicit construction of the local extension $\mathrm{Tr}_{D E} \circ \Phi$ where $\Phi = \mathcal{Q}_1 \circ \mathcal{R}_{F \backslash C D} \circ \mathcal{P}_{B E}$.
        Yellow regions represent the regions consistent with $\rho$; white regions represent the regions being in $\ket{0}\!\bra{0}$; pink regions represent the edge of backward lightcone for $ACD$ (see Eq.\,(\ref{eq:edge_LC})).
        The step from (a) to (b) is from the local inversion $\mathcal{P}_{B E}$.
        The step from (b) to (c) is from the reset channel $\mathcal{R}_{F \backslash C D}$ acting on $F \backslash CD$ (\textit{i.e.}, the gray region in Figure \ref{fig:F}b).
        The step from (c) to (d) is from the map $\mathrm{Tr}_{D E} \circ \mathcal{Q}_1$ where $\mathcal{Q}_1 =\mathcal{B}_{B C} \backslash \mathcal{B}_A$ acting on $B C (s) \backslash A$. 
        (e) The evolution for the cross-section of $A B C$ (gray dashed lines in a-d), where the circuits are similar to Figure \ref{fig:AM}.
        }
        \label{fig:LE}
    \end{figure}
    
    We also notice that $\rho_{A B E} \otimes \ket{0} \! \bra{0}_{C D} =\mathcal{R}_{C D} (\rho)$. Then, we can plug $\Phi = \mathrm{Tr}_{D E} \circ \mathcal{Q}_1 \circ \mathcal{R}_{F \backslash C D} \circ \mathcal{P}_{B E}$ and bound the recovery error
    \begin{eqnarray}
        & & \left\| \rho_{A B C} - \mathrm{Tr}_{D E} \circ \Phi \left( \rho_{A B E} \otimes \ket{0} \bra{0}_{C D} \right) \right\|_1 \nonumber\\
        & = & \| \rho_{A B C} - \mathrm{Tr}_{D E} \circ \mathcal{Q}_1 \circ \mathcal{R}_{F \backslash C D} \circ \mathcal{P}_{B E} \circ \mathcal{R}_{C D} (\rho) \|_1 \nonumber\\
        & \overset{\text{(i)}}{=} & \| \rho_{A B C} - \mathrm{Tr}_{D E} \circ \mathcal{Q}_1 \circ \mathcal{R}_F \circ \mathcal{P}_{B E} (\rho) \|_1 \nonumber\\
        & \overset{\text{(ii)}}{\leq} & \| \rho_{A B C} - \mathrm{Tr}_{D E} \circ \mathcal{Q}_1 \circ \mathcal{R}_F \circ \mathcal{B}_{A C D} (\rho_0) \|_1 \nonumber\\
        & & + \| \mathrm{Tr}_{D E} \circ \mathcal{Q}_1 \circ \mathcal{R}_F \circ \mathcal{B}_{A C D} (\rho_0) - \mathrm{Tr}_{D E} \circ \mathcal{Q}_1 \circ \mathcal{R}_F \circ \mathcal{P}_{B E} (\rho) \|_1 . \label{eq:tri}
    \end{eqnarray}
    Here, equality (i) is from the commutability between $\mathcal{P}_{B E}$ and $\mathcal{R}_{C D}$, and reset channel property $\mathcal{R}_{F \backslash C D} \circ \mathcal{R}_{C D} =\mathcal{R}_F$; inequality (ii) is from the triangle inequality. The first term in Eq.\,(\ref{eq:tri}) is
    \begin{eqnarray}
        \| \rho_{A B C} - \mathrm{Tr}_{D E} \circ \mathcal{Q}_1 \circ \mathcal{R}_F \circ \mathcal{B}_{A C D} (\rho_0) \|_1 & \overset{\text{(iii)}}{=} & \| \rho_{A B C} - \mathrm{Tr}_{D E} \circ \mathcal{Q}_1 \circ \mathcal{R}_F \circ \mathcal{Q}_2 \circ \mathcal{B}_A (\rho_0) \|_1 \nonumber\\
        & \overset{\text{(iv)}}{=} & \| \rho_{A B C} - \mathrm{Tr}_{D E} \circ \mathcal{Q}_1 \circ \mathcal{R}_F \circ \mathcal{B}_A (\rho_0) \|_1 \nonumber\\
        & \overset{\text{(v)}}{=} & \| \rho_{A B C} - \mathrm{Tr}_{D E} \circ  \mathcal{R}_{F_1} \circ \mathcal{Q}_1 \circ \mathcal{B}_A \circ \mathcal{R}_{F_2} (\rho_0) \|_1 \nonumber\\
        & \overset{\text{(vi)}}{=} & \| \rho_{A B C} - \mathrm{Tr}_{D E} \circ \mathcal{B}_{A B C} (\rho_0) \|_1 = 0. 
    \end{eqnarray}
    Here, equality (iii) is from the decomposition $\mathcal{B}_{A C D} =\mathcal{Q}_2 \circ \mathcal{B}_A$; equality (iv) is from the fact $\mathcal{R}_F \circ \mathcal{Q}_2 =\mathcal{R}_F$; for equality (v) we first use the decomposition $\mathcal{R}_F =\mathcal{R}_{F_1} \circ \mathcal{R}_{F_2}$, then we use the fact that $\mathcal{R}_{F_1}$ commutes with $\mathcal{Q}_1$, and $\mathcal{R}_{F_2}$ commutes with $\mathcal{B}_A$; equality (vi) is from the fact $\mathrm{Tr}_{D E} \circ \mathcal{R}_{F_1} = \mathrm{Tr}_{D E}$, the decomposition $\mathcal{B}_{A B C} =\mathcal{Q}_1 \circ \mathcal{B}_A$ and the fact $\mathcal{R}_{F_2} (\rho_0) = \rho_0$ for all-zero state.
  
    The second term in Eq.\,(\ref{eq:tri}) is bounded using the contractivity of the CPTP map $\mathrm{Tr}_{D E} \circ \mathcal{Q}_1 \circ \mathcal{R}_F$
    \begin{equation}
        \| \mathrm{Tr}_{D E} \circ \mathcal{Q}_1 \circ \mathcal{R}_F \circ \mathcal{B}_{A C D} (\rho_0) - \mathrm{Tr}_{D E} \circ \mathcal{Q}_1 \circ \mathcal{R}_F \circ \mathcal{P}_{B E} (\rho) \|_1 \leq \| \mathcal{B}_{A C D} (\rho_0) -\mathcal{P}_{B E} (\rho) \|_1 \leq \varepsilon_{\mathrm{LI}} .
    \end{equation}
    Put the results together:
    \begin{equation}
        \left\| \rho_{A B C} - \mathrm{Tr}_{D E}\!\left( \Phi\!\left( \rho_{A B E} \otimes \ket{0} \! \bra{0}_{C D} \right) \right) \right\|_1 \leq \varepsilon_{\mathrm{LI}},
    \end{equation}
    which completes the proof. 
\end{proof}

The proof of Theorem \ref{thm:LE} explains why area $E$ is needed. Intuitively, after applying $\mathcal{Q}_1 \circ \mathcal{R}_F \circ \mathcal{P}_{B E}$, the area $E$ contains wrong information of $\rho_E$, and some incorrect entanglement between $E$ and $ABC$ accumulates. Refreshing $E$ into $\ket{0}\!\bra{0}_E$ eliminates these errors and prevents them from propagating within the state generation.

\begin{corollary}
    \label{cor:LE}
    Suppose a trivial phase mixed state $\rho =\mathcal{E} (\rho_0)$, and $\rho_0$ is $\varepsilon_{\mathrm{LR}}$-locally reversible from $\rho$ along the shallow channel circuit $\mathcal{E}$. Now, consider a partition $\Lambda = A B C D E$ with $\mathrm{dist} (A, C) \geq 2 s$ and $\mathrm{dist} (B, D) \geq 2 s$. For any partition $A = A_1 A_2$, there always exists a channel acting on $B E \rightarrow B C$ such that
    \begin{equation}
        \| \rho_{A_1 B C} - \Phi_{B E \rightarrow B C} (\rho_{A_1 B E}) \|_1 \leq \varepsilon_{\mathrm{LI}} .
    \end{equation}
\end{corollary}

\begin{proof}
    We take the partial trace $\mathrm{Tr}_{A_2}$, and the corollary is implied by the contractivity of $\mathrm{Tr}_{A_2}$.
\end{proof}

\section{Learning Algorithm}
\label{sec:learning_alg}

To implement the learning and generation procedure, we first introduce a suitable covering scheme of the lattice in Section \ref{sec:cover_scheme}. Based on this covering, we construct a sequence of local extension steps that progressively reconstruct the target state $\rho$ patch by patch. At a high level, each step applies a learned local extension map so that, after all patches are processed, the resulting state approximates $\rho$.

However, realizing this strategy efficiently is nontrivial. Several technical obstacles arise in learning and certifying the required local extension maps. In particular, (i) requirement of global information for certification; (ii) approximation error in the optimization learning procedures; and (iii) local tomography error. We address these three issues in turn in Section \ref{sec:learning_locally}, Section \ref{sec:SDP}, and Section \ref{sec:imperfect_local_tomography}, and summarize the overall learning error in Section \ref{sec:overall_error_single}.

\subsection{Covering scheme}
\label{sec:cover_scheme}

\begin{figure}[t]
    \centering
    \includegraphics[width=0.9\linewidth]{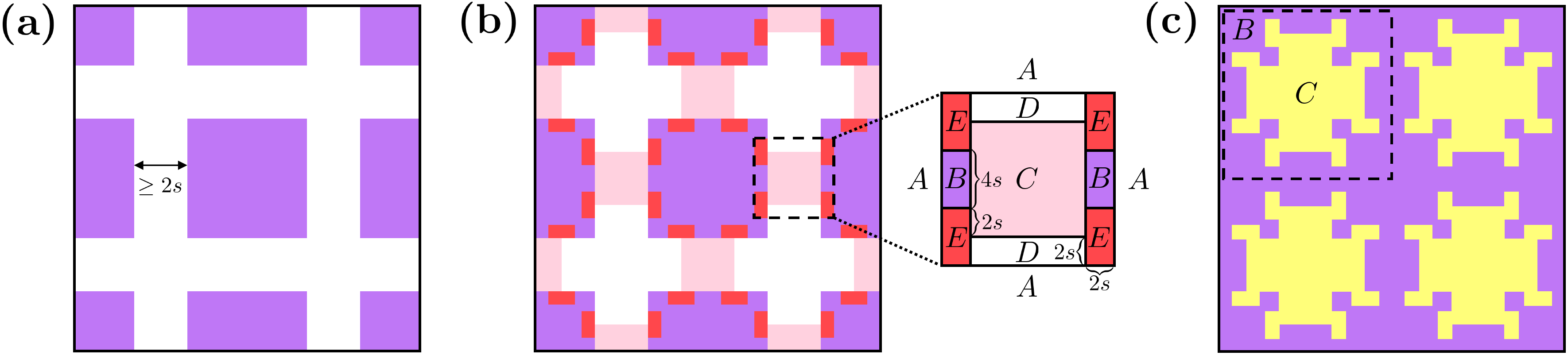}
    \caption{
    Details of the covering scheme for the case $k=2$, simplified from Figure \ref{fig:covering}.
    (a) Layer 1: learn and generate all local states in purple, and all patches are separated by a distance at least $2s$.
    (b) Layer 2: local extension maps, which extend to new regions $C$ (pink) from $BE$ (purple and red) and discard $E$ (red).
    (c) Layer 3: local recovery maps, which extend to regions $C$ (yellow) from $B$ (purple).
    }
    \label{fig:covering-dim}
\end{figure}

In this section, we present a covering scheme that enables us to generate $\rho$ from $\ket{0}\!\bra{0}^{\otimes n}$ step by step. This scheme was first proposed in Ref.\,\cite{kim_2024_learning} to learn any trivial phase pure state.

\begin{theorem}[Covering scheme, Section 4.1 of \cite{kim_2024_learning}]
    \label{thm:cover_scheme}
    There exists a sequence of regions $S_0, S_1, \cdots, S_K$ and pentapartitions $\{ \Lambda = A_i B_i C_i D_i E_i \}_{i \in [K]}$ where $K \leq n$, satisfying the following five conditions:
    \begin{enumerate}
        \item $S_0 = \varnothing$ and $S_K = \Lambda$.
        
        \item For any $i \in [K]$, we have $B_i \subseteq S_{i - 1} \cap S_i$, $C_i = S_i \backslash S_{i - 1}$, $D_i \subseteq \overline{S_i \cup S_{i - 1}}$, $E_i = S_{i - 1} \backslash S_i$, and $\mathrm{dist} (A_i, C_i) \geq 2 s$, $\mathrm{dist} (B_i, D_i) \geq 2 s$.
        
        \item $| B_i C_i D_i E_i | = O (c d)^k$, and both $C_i$ and $B_i C_i D_i E_i$ are simply-connected.
        
        \item The partitions can be arranged into $k + 1$ layers and each layer contains $K_i$ partitions, such that the support $\mathrm{Supp} (B_i C_i D_i E_i) \cap \mathrm{Supp} (B_j C_j D_j E_j) = \varnothing$ if partitions $i$ and $j$ are in the same layer. We adopt the convention that the partition index $i$ increases with the layer index.
        
        \item For the partition $i$ in layer $1$, $B_i D_i E_i = \varnothing$ and all $C_i$ in layer 1 are separated by a distance at least $2s$; for the partition $i$ in layer $k + 1$, $D_i E_i = \varnothing$. 
    \end{enumerate}
\end{theorem}

According to the condition 3 and 4 in Theorem \ref{thm:cover_scheme}, we know that there exists a sequence of quantum channels $\{ \mathcal{W}_i \}_{i \in [K]}$ such that
\begin{equation}
    \| \mathcal{W}_i (\rho_{S_{i - 1}}) - \rho_{S_i} \|_1 \leq \varepsilon_{\mathrm{LI}}. \label{eq:1eLI}
\end{equation}
These $K$ channels $\{ \mathcal{W}_i\}_{i \in [K]}$ can be arranged into $k+1$ layers.
In each layer, the channels $\mathcal{W}_i$ in this layer are non-overlapping local channels so that they can act in parallel. The condition 5 in Theorem \ref{thm:cover_scheme} says that: in layer 1, $\mathcal{W}_i$ is a trivial state generation channel for the local state $\rho_{C_i}$, and $\rho_{S_{K_1}} = \rho_{C_1} \otimes \cdots \otimes \rho_{C_{K_1}}$ according to Fact \ref{fac:LCD-3}; in layer $k + 1$, $\mathcal{W}_i = \Psi_{B_i \rightarrow B_i C_i}$ is a local recovery map; in the rest layers, $\mathcal{W}_i = \Phi_{B_i E_i \rightarrow B_i C_i}$ is a local extension map.
Finally, condition 1 in Theorem \ref{thm:cover_scheme} roughly says that $\rho_{S_0} = \rho_0$ and $\rho_{S_K} = \rho$, namely $\mathcal{W}_K \circ \cdots \circ \mathcal{W}_1$ approximately generates $\rho$ from $\rho_0$ (see quantitative error analysis in Section \ref{sec:analysis}).
We provided a visualization of the $k = 2$ case in Figure \ref{fig:covering-dim}, which generates a 2D trivial phase mixed state using 3 layers of local channels.

Theorem \ref{thm:cover_scheme} also answers why $D$ appears in the pentapartion $\Lambda = ABCDE$ in Definition \ref{def:LE}, but $D$ does not appear explicitly in Eq.\,(\ref{eq:def_LE}). 
If we absorb $D$ into $A$, then the proof of Theorem \ref{thm:LE} does not work, because $A$ contacts $C$ and a map $\Phi_{BE \to BC}(\rho_{A B E}) \approx \rho_{A B C}$ may not exist.
On the other hand, if we absorb $D$ into $E$, the proof of Theorem \ref{thm:LE} still works. However, it leads to the region $B E$ forming a connected region, which does not always occur in the covering scheme described in Theorem \ref{thm:cover_scheme}.

So far, we have already shown that there \textit{exists} a sequence of local extension maps that can approximately generate $\rho$ from $\rho_0$. 
However, in our constructive proofs in Section \ref{sec:proof_AM} and Section \ref{sec:proof_LE}, we require the knowledge of each $\mathcal{E}_{\ell, x}$ and $\tilde{\mathcal{E}}_{\ell, x}$. This information is not accessible because we only have access to the unknown state, not the channel circuit that prepares it. 
Therefore, it is not clear how to operably learn these $\{ \mathcal{W}_i \}_{i \in [K]}$ in an efficient way.
There are three difficulties for learning these local extension maps efficiently and accurately:
\begin{enumerate}
    \item Although a local extension map $\Phi_{B E \rightarrow B C}$ acts only on the local region $B E$ (and produces output on $B C$), its correctness is defined through a \textit{global} partial recovery condition:
    $\rho_{A B C} \approx \Phi_{B E \rightarrow B C} (\rho_{A B E})$ which involves the full region $A$ (see Eq.\,(\ref{eq:def_LE})). 
    Thus, certifying or learning such a channel appears to require access to global information, even though the map itself is local.
    \item The learning of each local extension map $\mathcal{W}_i$ is implemented via an optimization procedure (e.g., a semidefinite programming). However, any numerical or algorithmic optimization inevitably introduces approximation error.
    \item Learning $\Phi_{B E \rightarrow B C}$ requires accurate knowledge of the local reduced state $\rho_{BCDE}$, which must itself be obtained via quantum tomography. In practice, the number of copies and measurements is finite, leading to statistical estimation error. Such tomography error directly affects the quality of the learned extension map and must be carefully incorporated into the overall sample and time complexity analysis.
\end{enumerate}
We will tackle these three problems in the remaining part of this section. 

\subsection{Learning with spatial cutoff}
\label{sec:learning_locally}

We first deal with the first difficulty mentioned in Section \ref{sec:cover_scheme}. The solution to avoiding certifying the global partial recovery condition $\rho_{A B C} \approx \Phi_{B E \rightarrow B C} (\rho_{A B E})$ is manipulating the learning procedure only on a much smaller local system. In general, learning a local extension map for a smaller system does not ensure that the learned local extension map also works for the global system.
Thanks to the approximate Markovianity of the trivial phase mixed states, we now show that the learning of $\Phi_{B E \rightarrow B C}$ can be done only using a local system without causing an unbounded error.

To make our argument be more precise, we take the partition $A = A_{\mathrm{in}} A_{\mathrm{out}}$, where $A_{\mathrm{in}}$ is the buffer between $A_{\mathrm{out}}$ and $B C D E$ such that $\mathrm{dist} (B C D E, A_{\mathrm{out}}) \geq 2 s$ (see Figure \ref{fig:Ain}).

\begin{lemma}[Learning extension maps locally]
    \label{lem:learning_locally}
    Suppose a trivial phase mixed state $\rho =\mathcal{E} (\rho_0)$, and $\rho_0$ is $\varepsilon_{\mathrm{LR}}$-locally reversible from $\rho$ along the shallow channel circuit $\mathcal{E}$. Suppose a channel $\Phi_{B E \rightarrow B C}$ satisfying
    \begin{equation}
        \| \rho_{A_{\mathrm{in}} B C} - \Phi_{B E \rightarrow B C} (\rho_{A_{\mathrm{in}} B E}) \|_1 \leq \varepsilon_{\mathrm{LE}} .
    \end{equation}
    Then, we have
    \begin{equation}
    \| \rho_{A B C} - \Phi_{B E \rightarrow B C} (\rho_{A B E}) \|_1 \leq \varepsilon_{\mathrm{LE}} + 2 \varepsilon_{\mathrm{LI}} .
    \end{equation}
\end{lemma}

\begin{figure}
    \centering
    \includegraphics[width=0.6\linewidth]{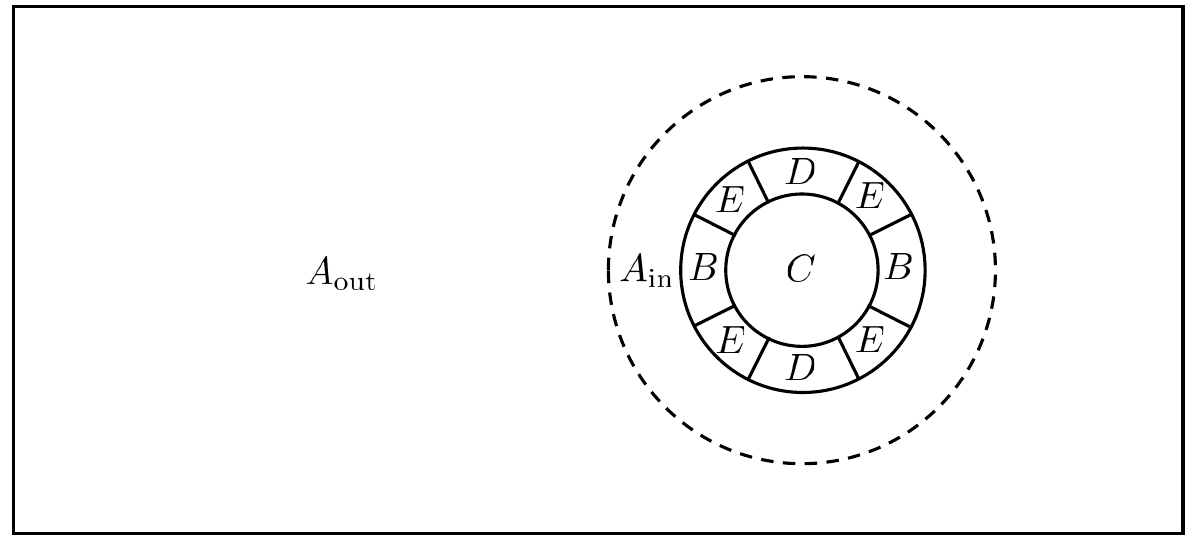}
    \caption{
    For a trivial phase mixed state $\rho$, if there is a local extension map $\Phi_{B E \to B C}$ learned on the local state $\rho_{A_{\mathrm{in}} B C D E}$ with high accuracy, then $\Phi_{B E \to B C}$ is also a local extension map for global state $\rho_{A B C D E}$ with high accuracy.
    }
    \label{fig:Ain}
\end{figure}

\begin{proof}
    Because $\rho$ trivial phase mixed state, then according to Theorem \ref{thm:AM}, there exists a local recovery map $\Psi_{A_{\mathrm{in}} \rightarrow A}$ such that
    \begin{equation}
        \| \Psi_{A_{\mathrm{in}} \rightarrow A} (\rho_{A_{\mathrm{in}} B C D E}) - \rho_{A B C D E} \| \leq \varepsilon_{\mathrm{LI}} .
    \end{equation}
    We remark here that $\Phi_{B E \rightarrow B C}$ and $\Psi_{A_{\mathrm{in}} \rightarrow A}$ commute. According to the contractivity of the CPTP map, this also guarantees
    \begin{align}
        \| \Psi_{A_{\mathrm{in}} \rightarrow A} (\rho_{A_{\mathrm{in}} B E}) - \rho_{A B E} \| & \leq \varepsilon_{\mathrm{LI}}, \\
        \| \Psi_{A_{\mathrm{in}} \rightarrow A} (\rho_{A_{\mathrm{in}} B C}) - \rho_{A B C} \| & \leq \varepsilon_{\mathrm{LI}} . 
    \end{align}
    Now, using the triangle inequality
    \begin{align}
        & \| \rho_{A B C} -\mathcal{I}_A \otimes \Phi_{B E \rightarrow B C} (\rho_{A B E}) \|_1 \nonumber\\
        \leq~& \| \rho_{A B C} - \Psi_{A_{\mathrm{in}} \rightarrow A} (\rho_{A_{\mathrm{in}} B C}) \|_1 + \| \Psi_{A_{\mathrm{in}} \rightarrow A} (\rho_{A_{\mathrm{in}} B C}) - \Psi_{A_{\mathrm{in}} \rightarrow A} (\Phi_{B E \rightarrow B C} (\rho_{A_{\mathrm{in}} B E})) \|_1 \nonumber\\
        & + \| \Psi_{A_{\mathrm{in}} \rightarrow A} (\Phi_{B E \rightarrow B C} (\rho_{A_{\mathrm{in}} B E})) - \Phi_{B E \rightarrow B C} (\rho_{A B E}) \|_1 . \label{eq:learning_locally_triangle}
    \end{align}
    The first term in Eq.\,(\ref{eq:learning_locally_triangle}) is bounded by $\| \rho_{A B C} - \Psi_{A_{\mathrm{in}} \rightarrow A} (\rho_{A_{\mathrm{in}} B C}) \|_1 \leq \varepsilon_{\mathrm{LI}}$. The second term in Eq.\,(\ref{eq:learning_locally_triangle}) is bounded by the contractivity of the CPTP map $\Psi_{A_{\mathrm{in}} \rightarrow A}$
    \begin{equation}
        \| \Psi_{A_{\mathrm{in}} \rightarrow A} (\rho_{A_{\mathrm{in}} B C}) - \Psi_{A_{\mathrm{in}} \rightarrow A} (\Phi_{B E \rightarrow B C} (\rho_{A_{\mathrm{in}} B E})) \|_1 \leq \| \rho_{A_{\mathrm{in}} B C} - \Phi_{B E \rightarrow B C} (\rho_{A_{\mathrm{in}} B E}) \| \leq \varepsilon_{\mathrm{LE}} .
    \end{equation}
    The third term in Eq.\,(\ref{eq:learning_locally_triangle}) is bounded by
    \begin{eqnarray}
        &  & \| \Psi_{A_{\mathrm{in}} \rightarrow A} (\Phi_{B E \rightarrow B C} (\rho_{A_{\mathrm{in}} B E})) - \Phi_{B E \rightarrow B C} (\rho_{A B E}) \|_1 \nonumber\\
        & \overset{\text{(i)}}{=} & \| \Phi_{B E \rightarrow B C} (\Psi_{A_{\mathrm{in}} \rightarrow A} (\rho_{A_{\mathrm{in}} B E})) - \Phi_{B E \rightarrow B C} (\rho_{A B E}) \|_1 \nonumber\\
        & \overset{\text{(ii)}}{\leq} & \| \Psi_{A_{\mathrm{in}} \rightarrow A} (\rho_{A_{\mathrm{in}} B E}) - \rho_{A B E} \|_1 \leq \varepsilon_{\mathrm{LI}} . 
    \end{eqnarray}
    Here, equality (i) is from the commutability between $\Psi_{A_{\mathrm{in}} \rightarrow A}$ and $\Phi_{B E \rightarrow B C}$; inequality (ii) is from the contractivity of the CPTP map $\Phi_{B E \rightarrow B C}$. Combined the results together, we have
    \begin{equation}
        \| \rho_{A B C} - \Phi_{B E \rightarrow B C} (\rho_{A B E}) \|_1 \leq \varepsilon_{\mathrm{LE}} + 2 \varepsilon_{\mathrm{LI}} ,
    \end{equation}
    which completes the proof.
\end{proof}

We remark that Ref.\,\cite{sang_2026} studies the closeness between two distances $\| \rho_{A B C} - \Phi_{B E \rightarrow B C} (\rho_{A B E}) \|_1$ and $\| \rho_{A_{\mathrm{in}} B C} - \Phi_{B E \rightarrow B C} (\rho_{A_{\mathrm{in}} B E}) \|_1$, even when both of them are not close to zero for a $\rho$ in a nontrivial topological phase.
The scenario of our Lemma \ref{lem:learning_locally} is a special case of this generic setting in Ref.\,\cite{sang_2026}.

\subsection{SDP algorithm for learning local extension maps}
\label{sec:SDP}

Now, we have already reduced the learning problem of the local extension map into an optimization problem $\min_{\Phi_{BE \to BC}} \| \rho_{A_{\mathrm{in}}B C} - \Phi_{B E \rightarrow B C} (\rho_{A_{\mathrm{in}} B E}) \|_1$. This minimization problem can be solved by the related optimization problem proposed in Ref.\,\cite{berta_2016_fidelity}.
Alternatively, by Fuchs-van de Graaf inequalities, we can focus on the maximization problem of the recovery fidelity
\begin{equation}
    \max_{\Phi_{BE \to BC}} F( \rho_{A_{\mathrm{in}} B C} , \Phi_{B E \rightarrow B C} (\rho_{A_{\mathrm{in}} B E}) ), \label{eq:optF}
\end{equation}
where the tomography of $\rho_{A_{\mathrm{in}} B C}$ and $\rho_{A_{\mathrm{in}} B E}$ can be computed classically from the classical shadow snapshots. We leave the discussion about the error of tomography in Section \ref{sec:imperfect_local_tomography}.

We now state how we can leverage the algorithm given in Ref.\,\cite{berta_2016_fidelity} to find such an optimal local extension map $\Phi_{BE \to BC}$.

This subroutine is given as follows.
Let $\rho_{A_{\mathrm{in}} B E R}$ be any purification of $\rho_{A_{\mathrm{in}} B E}$, and $\Gamma_{A_{\mathrm{in}} R : B E} = \ket{\Gamma}\!\bra{\Gamma}_{A_{\mathrm{in}} R : B E}$, where $\ket{\Gamma}_{A_{\mathrm{in}} R : B E}=\sum_{i}\ket{i}_{A_{\mathrm{in}} R}\ket{i}_{B E}$ is the unnormalized maximally entangled state between $A_{\mathrm{in}} R$ and $B E$. 
The optimal recovery fidelity Eq.\,(\ref{eq:optF}) is obtained by solving the following SDP algorithm:
\begin{eqnarray}
    & \min & \frac{1}{2} [Z_{A_{\mathrm{in}} B C} + Z_{A_{\mathrm{in}} B C}^{\dagger}] \label{eq:sdp_Z}\\
    & \text{subject to} & Z_{A_{\mathrm{in}} B C} \in \mathcal{L} (A_{\mathrm{in}} B C), \tau_{A_{\mathrm{in}} B C R} \succeq 0, \mathrm{Tr}_{B C} [\tau_{A_{\mathrm{in}} B C R}] =\mathbb{I}_{A_{\mathrm{in}} R}, \nonumber\\
    &  & \left(\begin{array}{cc} 
        \rho_{A_{\mathrm{in}} B C} & Z_{A_{\mathrm{in}} B C}\\ 
        Z_{A_{\mathrm{in}} B C}^{\dagger} & \mathrm{Tr}_R \left( \sqrt{\rho_{A_{\mathrm{in}} R}} \tau_{A_{\mathrm{in}} B C R} \sqrt{\rho_{A_{\mathrm{in}} R}} \right)
    \end{array}\right) 
    \succeq 0. \nonumber
\end{eqnarray}
Here, $\mathcal{L}$ denotes the set of linear operators. 
By solving the SDP algorithm Eq.\,(\ref{eq:sdp_Z}), we will get a mixed state $\tau_{A_{\mathrm{in}} B C R}$.
This is state is exactly the unnormalized \textit{Choi-Jamiolkowski state} of the optimal local extension map $\Phi_{B E \rightarrow B C}$. Therefore, $\Phi_{B E \rightarrow B C}$ can be obtained by reshuffling:
\begin{equation}
    \tau_{A_{\mathrm{in}} B C R} = \Phi_{B E \rightarrow B C} (\Gamma_{A_{\mathrm{in}} R : B E}), \quad \mathrm{Tr}_{B C} [\tau_{A_{\mathrm{in}} B C R}] =\mathbb{I}_{A_{\mathrm{in}} R} .
\end{equation}
Given any recovery 1-norm distance error $\varepsilon_{\mathrm{SDP}}$ for $\max_{\Phi_{BE \to BC}} F( \rho_{A_{\mathrm{in}} B C} , \Phi_{B E \rightarrow B C} (\rho_{A_{\mathrm{in}} B E}) )$, the running time of this SDP problem is:
\begin{equation}
    2^{O (c d)^k} \log \! \left( \frac{1}{\varepsilon_{\mathrm{SDP}}} \right). \label{eq:e_sdp}
\end{equation}
We remark that the Fuchs-van de Graaf inequalities also ensures the scaling of solving the minimization problem $\min_{\Phi_{BE \to BC}} \| \rho_{A B C} - \Phi_{B E \rightarrow B C} (\rho_{A B E}) \|_1$ is the same as Eq.\,(\ref{eq:e_sdp}).

\subsection{Imperfect local tomography}
\label{sec:imperfect_local_tomography}

We turn to the third difficulty mentioned in Section \ref{sec:cover_scheme} -- the imperfect local tomography.

For a perfect locally tomographed state $\rho_{A_{\mathrm{in}} B C D E}$, Corollary \ref{cor:AM} and Corollary \ref{cor:LE} guarantee that $\rho_{A_{\mathrm{in}} B C D E}$ is $\varepsilon_{\mathrm{LI}}$-locally extendible, \textit{i.e.}, there always exists a channel $\Phi_{B E \rightarrow B C}$ satisfying the condition $\| \rho_{A_{\mathrm{in}} B C} - \Phi_{B E \rightarrow B C} (\rho_{A_{\mathrm{in}} B E}) \|_1 \leq \varepsilon_{\mathrm{LI}}$. 
Now suppose we have the local tomography $\rho'_{A_{\mathrm{in}} B C D E}$, which is imperfect due to measurement fluctuation, which causes a $1$-norm tomography error up to some $\varepsilon_{\mathrm{LT}} > 0$. 
Our following lemma indicates that the imperfect state $\rho'_{A_{\mathrm{in}} B C D E}$ is still a $(\varepsilon_{\mathrm{LI}} + 2 \varepsilon_{\mathrm{LT}})$-locally extendible state.
% In Section XXX, we will give a bound on the sample and time complexity for achieving such a local tomography accuracy with a bounded failure probability.

\begin{lemma}[Robustness to imperfect local tomography]
    \label{lem:robustness_LT}
    Suppose a trivial phase mixed state $\rho =\mathcal{E} (\rho_0)$, and $\rho_0$ is $\varepsilon_{\mathrm{LR}}$-locally reversible from $\rho$ along the shallow channel circuit $\mathcal{E}$. Let a channel $\Phi_{B E \rightarrow B C}$ satisfying $\| \rho_{A_{\mathrm{in}} B C} - \Phi_{B E \rightarrow B C} (\rho_{A_{\mathrm{in}} B E}) \|_1 \leq \varepsilon_{\mathrm{LI}}$. Now, consider an imperfect local tomography $\rho'_{A_{\mathrm{in}} B C D E}$ such that
    \begin{equation}
        \| \rho_{A_{\mathrm{in}} B C D E} - \rho'_{A_{\mathrm{in}} B C D E} \|_1 \leq \varepsilon_{\mathrm{LT}}
    \end{equation}
    Then we have
    \begin{equation}
        \| \rho'_{A_{\mathrm{in}} B C} - \Phi_{B E \rightarrow B C} (\rho'_{A_{\mathrm{in}} B E}) \|_1 \leq \varepsilon_{\mathrm{LI}} + 2 \varepsilon_{\mathrm{LT}} .
    \end{equation}
\end{lemma}

\begin{proof}
    According to the contractivity of the CPTP map, we have
    \begin{align}
        \| \rho_{A_{\mathrm{in}} B C} - \rho'_{A_{\mathrm{in}} B C} \| & \leq \varepsilon_{\mathrm{LT}}, \\
        \| \rho_{A_{\mathrm{in}} B E} - \rho'_{A_{\mathrm{in}} B E} \| & \leq \varepsilon_{\mathrm{LT}} . 
    \end{align}
    Now, using the triangle inequality
    \begin{align}
        \| \rho'_{A_{\mathrm{in}} B C} - \Phi_{B E \rightarrow B C} (\rho'_{A_{\mathrm{in}} B E}) \|_1 \leq~& \| \rho'_{A_{\mathrm{in}} B C} - \rho_{A_{\mathrm{in}} B C} \|_1 + \| \rho_{A_{\mathrm{in}} B C} - \Phi_{B E \rightarrow B C} (\rho_{A_{\mathrm{in}} B E}) \|_1 \nonumber\\ 
        & + \| \Phi_{B E \rightarrow B C} (\rho_{A_{\mathrm{in}} B E}) - \Phi_{B E \rightarrow B C} (\rho'_{A_{\mathrm{in}} B E}) \|_1 . \label{eq:ILT_triangle}
    \end{align}
    The first term in Eq.\,(\ref{eq:ILT_triangle}) is bounded by $\| \rho'_{A_{\mathrm{in}} B C} - \rho_{A_{\mathrm{in}} B C} \|_1 \leq \varepsilon_{\mathrm{LT}}$. 
    The second term in Eq.\,(\ref{eq:ILT_triangle}) is bounded by the fact that $\| \rho_{A_{\mathrm{in}} B C} - \Phi_{B E \rightarrow B C} (\rho_{A_{\mathrm{in}} B E}) \|_1 \leq \varepsilon_{\mathrm{LI}}$. 
    The third term in Eq.\,(\ref{eq:ILT_triangle}) is bounded by the contractivity of the CPTP map $\Phi_{B E \rightarrow B C}$
    \begin{equation}
        \| \Phi_{B E \rightarrow B C} (\rho_{A_{\mathrm{in}} B E}) - \Phi_{B E \rightarrow B C} (\rho'_{A_{\mathrm{in}} B E}) \|_1 \leq \| \rho_{A_{\mathrm{in}} B E} - \rho'_{A_{\mathrm{in}} B E} \| \leq \varepsilon_{\mathrm{LT}} .
    \end{equation}
    Combined the results together, we have
    \begin{equation}
        \| \rho'_{A_{\mathrm{in}} B C} - \Phi_{B E \rightarrow B C} (\rho'_{A_{\mathrm{in}} B E}) \|_1 \leq \varepsilon_{\mathrm{LI}} + 2 \varepsilon_{\mathrm{LT}} ,
    \end{equation}
    which completes the proof.
\end{proof}

\subsection{Overall error for learning a local extension map}
\label{sec:overall_error_single}

Suppose we have an imperfect locally tomographed state $\rho'_{A_{\mathrm{in}} B C D E}$ with $\| \rho_{A_{\mathrm{in}} B C D E} - \rho'_{A_{\mathrm{in}} B C D E} \|_1 \leq \varepsilon_{\mathrm{LT}}$.
According to Lemma \ref{lem:robustness_LT}, this imperfect state $\rho'_{A_{\mathrm{in}} B C D E}$ is still $(\varepsilon_{\mathrm{LI}} + 2 \varepsilon_{\mathrm{LT}})$-locally extendible.
Now we show that if we run the SDP algorithm Eq.\,(\ref{eq:sdp_Z}) for $\rho'_{A_{\mathrm{in}} B C D E}$ instead of for $\rho_{A_{\mathrm{in}} B C D E}$, the channel we obtain is also a local extension map for the global state $\rho_{A B C D E}$, even in the presence of positive $\varepsilon_{\mathrm{LT}}$ and $\varepsilon_{\mathrm{SDP}}$.

\begin{corollary}[Overall error for the learned $\Phi_{BE \to BC}$]
    \label{cor:overall_error_single}
    Suppose a trivial phase mixed state $\rho =\mathcal{E} (\rho_0)$, and $\rho_0$ is $\varepsilon_{\mathrm{LR}}$-locally reversible from $\rho$ along the shallow channel circuit $\mathcal{E}$. Consider an imperfect local tomography $\rho'_{A_{\mathrm{in}} B C D E}$ such that $\| \rho_{A_{\mathrm{in}} B C D E} - \rho'_{A_{\mathrm{in}} B C D E} \|_1 \leq \varepsilon_{\mathrm{LT}}$. 
    If we run the SDP algorithm for $\rho'_{A_{\mathrm{in}} B C D E}$ and output a channel $\Phi'_{B E \rightarrow B C}$ with at most $\varepsilon_{\mathrm{SDP}}$ error, i.e.,
    \begin{equation}
        \| \rho'_{A_{\mathrm{in}} B C} - \Phi'_{B E \rightarrow B C} (\rho'_{A_{\mathrm{in}} B E}) \|_1 \leq \varepsilon_{\mathrm{LI}} + 2 \varepsilon_{\mathrm{LT}} + \varepsilon_{\mathrm{SDP}} .
    \end{equation}
    Then, we have
    \begin{equation}
        \| \rho_{A B C} - \Phi'_{B E \rightarrow B C} (\rho_{A B E}) \|_1 \leq 3 \varepsilon_{\mathrm{LI}} + 4 \varepsilon_{\mathrm{LT}} + \varepsilon_{\mathrm{SDP}}. 
    \end{equation}
\end{corollary}

\begin{proof}
    We first show that $\| \rho_{A_{\mathrm{in}} B C} - \Phi'_{B E \rightarrow B C} (\rho_{A_{\mathrm{in}} B E}) \|_1 \leq \varepsilon_{\mathrm{LI}} + 4 \varepsilon_{\mathrm{LT}} + \varepsilon_{\mathrm{SDP}}$. The proof is similar to the proof of Lemma \ref{lem:robustness_LT}. 
    In fact, using the triangle inequality
    \begin{align}
        \| \rho_{A_{\mathrm{in}} B C} - \Phi'_{B E \rightarrow B C} (\rho_{A_{\mathrm{in}} B E}) \|_1 \leq~& \| \rho_{A_{\mathrm{in}} B C} - \rho'_{A_{\mathrm{in}} B C} \|_1 + \| \rho'_{A_{\mathrm{in}} B C} - \Phi'_{B E \rightarrow B C} (\rho'_{A_{\mathrm{in}} B E}) \|_1 \nonumber\\
        & + \| \Phi'_{B E \rightarrow B C} (\rho'_{A_{\mathrm{in}} B E}) - \Phi'_{B E \rightarrow B C} (\rho_{A_{\mathrm{in}} B E}) \|_1 . \label{eq:OES_triangle}
    \end{align}
    The first term in Eq.\,(\ref{eq:OES_triangle}) is bounded by $\| \rho_{A_{\mathrm{in}} B C} -  \rho'_{A_{\mathrm{in}} B C} \|_1 \leq \varepsilon_{\mathrm{LT}}$.
    The second term in Eq.\,(\ref{eq:OES_triangle}) is bounded by the fact that $\| \rho'_{A_{\mathrm{in}} B C} - \Phi'_{B E \rightarrow B C} (\rho'_{A_{\mathrm{in}} B E}) \|_1 \leq \varepsilon_{\mathrm{LI}} + 2 \varepsilon_{\mathrm{LT}} + \varepsilon_{\mathrm{SDP}}$.
    The third term in Eq.\,(\ref{eq:OES_triangle}) is bounded by the contractivity of the CPTP map $\Phi'_{B E \rightarrow B C}$
    \begin{equation}
        \| \Phi'_{B E \rightarrow B C} (\rho'_{A_{\mathrm{in}} B E}) - \Phi'_{B E \rightarrow B C} (\rho_{A_{\mathrm{in}} B E}) \|_1 \leq \| \rho'_{A_{\mathrm{in}} B E} - \rho_{A_{\mathrm{in}} B E} \| \leq \varepsilon_{\mathrm{LT}} .
    \end{equation}
    Combining the result together, we have $\| \rho_{A_{\mathrm{in}} B C} - \Phi'_{B E \rightarrow B C} (\rho_{A_{\mathrm{in}} B E}) \|_1 \leq \varepsilon_{\mathrm{LI}} + 4 \varepsilon_{\mathrm{LT}} + \varepsilon_{\mathrm{SDP}}$. 
    Now, we set $\varepsilon_{\mathrm{LE}} = \varepsilon_{\mathrm{LI}} + 4 \varepsilon_{\mathrm{LT}} + \varepsilon_{\mathrm{SDP}}$ in Lemma \ref{lem:learning_locally}. We immediately obtain
    \begin{equation}
        \| \rho_{A B C} - \Phi'_{B E \rightarrow B C} (\rho_{A B E}) \|_1 \leq \varepsilon_{\mathrm{LE}} + 2
    \varepsilon_{\mathrm{LI}} \leq 3 \varepsilon_{\mathrm{LI}} + 4 \varepsilon_{\mathrm{LT}} + \varepsilon_{\mathrm{SDP}}  ,
  \end{equation}
  which completes the proof.
\end{proof}

From now on, we will always simply take the local tomography error as $\varepsilon_{\mathrm{LT}} = \varepsilon_{\mathrm{LI}}$ and the SDP algorithm error as $\varepsilon_{\mathrm{SDP}} = \varepsilon_{\mathrm{LI}}$. This implies
\begin{equation}
    \| \rho_{A B C} - \Phi'_{B E \rightarrow B C} (\rho_{A B E}) \|_1 \leq 8 \varepsilon_{\mathrm{LI}} . \label{eq:8eLI}
\end{equation}
To summarize, by Theorem \ref{thm:LE} we know that the trivial phase mixed state $\rho$ is $\varepsilon_{\mathrm{LI}}$-locally extendible for some channel $\Phi_{B E \rightarrow B C}$ with $\| \rho_{A B C} - \Phi_{B E \rightarrow B C} (\rho_{A B E}) \|_1 \leq \varepsilon_{\mathrm{LI}}$. When we take the learning efficiency, SDP error, and tomography error into account, we can still learn a sub-optimal local extension $\Phi'_{B E \rightarrow B C}$ with $\| \rho_{A B C} - \Phi'_{B E \rightarrow B C} (\rho_{A B E}) \|_1 \leq 8 \varepsilon_{\mathrm{LI}}$.
We give the pseudo-algorithm for learning $\rho$ in Algorithm \ref{alg:learning}.

\RestyleAlgo{ruled}
\SetKwComment{Comment}{/* }{ */}
\SetKwInOut{input}{Input}
\SetKwInOut{output}{Output}
\SetKwFor{For}{For}{:}{EndFor}
% \SetKwIf{If}{If}{then:}{EndIf}
\SetKwIF{If}{ElseIf}{Else}{If}{then:}{Else if}{Else}{}
\SetKw{Continue}{Continue}
\SetKw{Break}{Break}

\begin{algorithm}[t]
    \caption{Learning Trivial Phase Mixed State $\rho$}\label{alg:learning}
    \input{$M$ copies of $n$-qubit $\rho$ in a $k$-dimensional lattice $\Lambda$, lightcone radius $s$, error $\varepsilon$ and failure success $\delta$}
    \output{Circuit $\mathcal{W}$, and state $\rho_{\mathrm{gen}}=\mathcal{W}(\ket{0}\!\bra{0}^{\otimes n})$, \textit{s.t.} $\|\rho_{\mathrm{gen}} - \rho\|_1 \leq \varepsilon$ with probability $1-\delta$}
    Compute $\varepsilon_{\mathrm{LI}} = \varepsilon/(8 n)$ by Eq.\,(\ref{eq:eLE_e8n})\;
    Perform classical shadows for $M$ copies of $\rho$\;
    % Set $\mathtt{Flag} \gets \mathtt{Fail}$\;
    Perform covering scheme $\langle S_i, \Lambda = A_i B_i C_i D_i E_i \rangle$ for $i \in [K]$ described in Theorem \ref{thm:cover_scheme}\;
    Organize $\{S_i\}$ into $k+1$ layers\Comment*[r]{$S_0 = \varnothing$ and $S_{K} = \Lambda$}
    Initialize $\rho_{S_0} = \ket{0}\!\bra{0}^{\otimes n} $\;
    \For{$i \gets 1$ to $K$}{
        \If{partition $i$ is in layer $1$}{
            Generate local state $\rho_{C_i} \gets \mathcal{W}'_i(\ket{0}\!\bra{0}^{\otimes n}) $ for some $\mathcal{W}'_i$\Comment*[r]{$B_i D_i E_i = \varnothing$}
        }
        \If{partition $i$ is in layer $2$ to $k$}{
            Partite $A_i = A_{\mathrm{in},i} A_{\mathrm{out},i}$ such that $\mathrm{dist}(B_i C_i D_i E_i, A_{\mathrm{out},i}) \geq 2s$\;
            Simulate $\rho_{A_{\mathrm{in},i} B_i C_i D_i E_i}$ using classical shadow snapshots\;
            Run SDP algorithm Eq.\,(\ref{eq:sdp_Z}) to get local extension map $\Phi'_{B_i E_i \to B_i C_i}$\;
            Set $\mathcal{W}'_i \gets \Phi'_{B_i E_i \to B_i C_i}$ and generate $\rho_{S_{i}} \gets \Phi'_{B_i E_i \to B_i C_i} (\rho_{S_{i - 1}})$\;
        }
        \If{partition $i$ is in layer $k+1$}{
            Partite $A_i = A_{\mathrm{in},i} A_{\mathrm{out},i}$ such that $\mathrm{dist}(B_i C_i, A_{\mathrm{out},i}) \geq 2s$\Comment*[r]{$D_i E_i = \varnothing$}
            Simulate $\rho_{A_{\mathrm{in},i} B_i C_i}$ using classical shadow snapshots\;
            Run SDP algorithm Eq.\,(\ref{eq:sdp_Z}) to get local recovery map $\Psi'_{B_i \to B_i C_i}$\;
            Set $\mathcal{W}'_i \gets \Psi'_{B_i \to B_i C_i}$ and generate $\rho_{S_{i}} \gets \Psi'_{B_i \to B_i C_i} (\rho_{S_{i - 1}})$\;
        }
    }
    \textbf{Return} $\mathcal{W} \gets \mathcal{W}'_K \circ \mathcal{W}'_{K - 1} \circ \cdots \circ \mathcal{W}'_1$ and $\rho_{\mathrm{gen}} \gets \rho_{S_{K}}$ %\Comment*[r]{All maps of $\Phi'_{B_i E_i \to B_i C_i}$ and $\Psi'_{B_i \to B_i C_i}$ exist}
\end{algorithm}

\section{Complexity Analysis}
\label{sec:analysis}

\begin{theorem}[Main Theorem]
    \label{thm:main}
    Suppose $\rho$ is an $n$-qubit trivial phase mixed state in a $k$-dimensional lattice. 
    That is, for any $\varepsilon_{\mathrm{LR}} > 0$, there always exists a shallow channel circuit $\mathcal{E}=\mathcal{E}_d \circ \cdots \circ \mathcal{E}_2 \circ \mathcal{E}_1$ with $\mathcal{E}_{\ell} = \prod_x \mathcal{E}_{\ell, x}$ where each gate $\mathcal{E}_{\ell, x}$ has a local support size bounded by $c \leq \mathrm{polylog} (n)$ and the overall circuit depth $d \leq \mathrm{polylog} (n)$, such that $\rho =\mathcal{E} ( \ket{0}\!\bra{0}^{\otimes n} )$ and $\ket{0}\!\bra{0}^{\otimes n}$ is $\varepsilon_{\mathrm{LR}}$-locally reversible from $\rho$ along the shallow channel circuit $\mathcal{E}$. 
    Then, given any $\varepsilon>0$ and $\delta > 0$, there exists an algorithm that runs time $T$, learns from $M$ many copies of $\rho$, and generates $\rho$ using a $(k + 1)$-layer local channel circuit $\mathcal{W}$, where
    \begin{equation}
        T = n \cdot 2^{O (c d)^k} \log \left( \frac{n}{\varepsilon} \right) + O(M), \quad M = \frac{n^2 \cdot 2^{O (c d)^k}}{\varepsilon^2} \log \left( \frac{n}{\delta} \right),
    \end{equation}
    such that $\| \mathcal{W}(\ket{0}\!\bra{0}^{\otimes n}) - \rho \|_1 \leq \varepsilon$ with success probability at least $1 - \delta$.
\end{theorem}

According to Eq.(\ref{eq:1eLI}), we know that there exists a sequence of quantum channels $\{ \mathcal{W}_i \}_{i \in [K]}$ such that
\begin{equation}
    \| \mathcal{W}_i (\rho_{S_{i - 1}}) - \rho_{S_i} \|_1 \leq \varepsilon_{\mathrm{LI}} .
\end{equation}
When we take the local tomography and SDP algorithm error into account, Eq.(\ref{eq:8eLI}) says that we can \textit{efficiently} obtain another sequence of quantum channels $\{ \mathcal{W}'_i \}_{i \in [K]}$ such that
\begin{equation}
    \| \mathcal{W}'_i (\rho_{S_{i - 1}}) - \rho_{S_i} \|_1 \leq 8 \varepsilon_{\mathrm{LI}} .
\end{equation}
We construct the \textit{learned generation channel}
\begin{equation}
    \mathcal{W}=\mathcal{W}'_K \circ \mathcal{W}'_{K - 1} \circ \cdots \circ \mathcal{W}'_1 .
\end{equation}
To prove our main Theorem \ref{thm:main}, we need the following lemma:

\begin{lemma}[Overall error of $\mathcal{W}$]
    For the learned generation channel $\mathcal{W}=\mathcal{W}'_K \circ \mathcal{W}'_{K - 1} \circ \cdots \circ \mathcal{W}'_1$, we have
    \begin{equation}
        \| \mathcal{W} (\rho_0) - \rho \|_1 \leq 8 n^2 d \cdot \varepsilon_{\mathrm{LR}}.
    \end{equation}
\end{lemma}

\begin{proof}
    We notice that $\rho_{S_0} = \rho_0$ and $\rho_{S_K} = \rho$, and
    \begin{eqnarray}
        \| \mathcal{W} (\rho_0) - \rho \|_1 & = & \| \mathcal{W}'_K \circ \mathcal{W}'_{K - 1} \circ \cdots \circ \mathcal{W}'_1 (\rho_{S_0}) - \rho_{S_K} \|_1 \nonumber\\
        & \overset{\text{(i)}}{\leq} & \| \mathcal{W}'_K \circ \mathcal{W}'_{K - 1} \circ \cdots \circ \mathcal{W}'_1 (\rho_0) -\mathcal{W}'_K (\rho_{S_{K - 1}}) \|_1 + \| \mathcal{W}'_K (\rho_{S_{K - 1}}) - \rho_{S_K} \|_1 \nonumber\\
        & \overset{\text{(ii)}}{\leq} & \| \mathcal{W}'_{K - 1} \circ \cdots \circ \mathcal{W}'_1 (\rho_0) - \rho_{S_{K - 1}} \|_1 + \| \mathcal{W}'_K (\rho_{S_{K - 1}}) - \rho_{S_K} \|_1 \nonumber\\
        & \overset{\text{(iii)}}{\leq} & \| \mathcal{W}'_{K - 1} \circ \cdots \circ \mathcal{W}'_1 (\rho_0) - \rho_{S_{K - 1}} \|_1 + 8 \varepsilon_{\mathrm{LI}} . 
    \end{eqnarray}
    Here, inequality (i) follows from the triangle inequality; inequality (ii) follows from the contractivity of the CPTP map $\mathcal{W}'_K$; inequality (iii) is from the property of learned channel $\| \mathcal{W}'_K (\rho_{S_{K - 1}}) - \rho_{S_K} \|_1 \leq 8 \varepsilon_{\mathrm{LI}}$. Then, by induction, we have
    \begin{equation}
        \| \mathcal{W} (\rho_0) - \rho \|_1 \leq K \cdot 8 \varepsilon_{\mathrm{LI}} \leq 8 n^2 d \cdot \varepsilon_{\mathrm{LR}},
    \end{equation}
    which completes the proof.
\end{proof}

% Since $K \leq n$, then we can bound the total error as
% \begin{equation}
%     \| \mathcal{E}_{\mathrm{learned}} (\rho_0) - \rho \|_1 \leq 8 n \cdot \varepsilon_{\mathrm{LI}} = 8 n^2 d \cdot \varepsilon_{\mathrm{LR}} .
% \end{equation}

Therefore, for any given target error $\varepsilon$, we can simply choose:
\begin{equation}
    \varepsilon_{\mathrm{LR}} = \frac{\varepsilon}{8 n^2 d}, \quad \varepsilon_{\mathrm{LI}} = \frac{\varepsilon}{8 n} . \label{eq:eLE_e8n}
\end{equation}

\textbf{Sample complexity}. Using the standard classical shadow methods \cite{huang_2020_predicting}, the copies of $\rho$ we need is
\begin{equation}
    M = \frac{n^2 \cdot 2^{O (cd)^k}}{\varepsilon^2} \log \left( \frac{n}{\delta} \right) .
\end{equation}
The running time of the classical shadow has the same scaling as the sample number.

\textbf{Time complexity}. 
For learning all $K \leq n$ local channels, the total running time is at most:
\begin{equation}
    T = n \cdot  2^{O (cd)^k} \log \left( \frac{n}{\varepsilon} \right) + \frac{n^2 \cdot 2^{O (cd)^k}}{\varepsilon^2} \log \left( \frac{n}{\delta} \right),
\end{equation}
where the first term comes from running the SDP algorithm, and the second term is from performing the classical shadow.

\textbf{Remark}. As mentioned in Section \ref{sec:main_result}, Theorems \ref{thm:AM} and \ref{thm:LE} show that each local recovery and local extension map can be decomposed into at most $(2d + 2)$ layers of gates, each of size at most $c$.
If the circuit structure of $\mathcal{E}$ is known and $c$ is a constant, the we can apply an epsilon-net searching to learn each $\mathcal{W}'_i$.
A standard epsilon-net covering argument yields the time complexity
\begin{equation}
    \left( \frac{n^2 d \cdot (c d)^k}{\varepsilon} \right)^{O\left((c d)^k \cdot 16^c \right)}
\end{equation}
to learn $\mathcal{W}'_i$ up to an error of $O(\varepsilon/(n^2 d))$. This runtime remains polynomial in $n$ when $c,d$ are constant, but becomes inefficient when $c=\mathrm{polylog}(n)$.

\section{Application to Quantum Generative Models}
\label{sec:application}

In quantum machine learning, quantum generative models aim to generate samples from a target probability distribution using quantum systems. A common approach encodes the desired classical distribution into a quantum state $\rho$, such that measuring $\rho$ in a fixed basis produces samples from the target distribution. The learning task is therefore to reconstruct the unknown state $\rho$ from measurement data and to prepare a quantum circuit that generates $\rho$, enabling sampling via repeated measurements.
Recently, quantum diffusion models, a quantum analogue of classical diffusion models, have been proposed as a framework for such generative tasks.

\begin{definition}[Quantum diffusion model]
    \label{def:QDM}
    For two mixed states $\rho$ and $\sigma$ in a $k$-dimensional lattice $\Lambda$, suppose there is a process (called \textnormal{diffusion}) generating $\rho$ from $\sigma$ through a continuous time evolution, described by the Lindblad equation
    \begin{equation}
            \frac{\partial \tilde{\varrho}_t}{\partial t} = \tilde{\mathcal{L}}_t (\tilde{\varrho}_t), \quad \tilde{\varrho}_{t = 0} = \rho, \quad \tilde{\varrho}_{t = 1} = \sigma, \quad \forall t \in [0, 1] .
    \end{equation}
    where $\tilde{\mathcal{L}}_t$ is the Lindbladian super-operator. Then, there is another process (called \textnormal{denoising}) $\mathcal{L}_t$ such that (notice that the initial and final distributions are inverted now)
    \begin{equation}
        \frac{\partial \varrho_t}{\partial t} = \mathcal{L}_t (\varrho_t), \quad \varrho_t = \tilde{\varrho}_{1-t}, \quad \forall t \in [0, 1] .
    \end{equation}
    The pair $(\tilde{\mathcal{L}}_t, \mathcal{L}_t)$ is a \textnormal{quantum diffusion model} that generates $\rho$ from $\sigma$ along the noise path $\varrho_t$. Especially, $\sigma$ is usually selected as the trivial distribution $\sigma = \ket{0}\!\bra{0}^{\otimes n}$.
\end{definition}

In practice, the dynamics are always time-discrete, which means we can give a time-discrete version of quantum diffusion models, which is reduced to the setting in our mixed state learning scenario: there exist channel circuits $\mathcal{E}$ and $\tilde{\mathcal{E}}$ with $d$ layers such that $\sigma$ is layer-wise reversible from $\rho$ along $\mathcal{E}$:
\begin{equation}
    \sigma = \tilde{\mathcal{E}} (\rho) =: \tilde{\mathcal{E}}_1 \circ \tilde{\mathcal{E}}_2 \circ \cdots \circ \tilde{\mathcal{E}}_d (\rho), \quad \rho = \mathcal{E} (Q) = \mathcal{E}_d \circ \cdots \circ \mathcal{E}_2 \circ \mathcal{E}_1 (Q).
\end{equation}
If we further take the locality into account, and suppose that each $\mathcal{E}_{\ell}$ and $\tilde{\mathcal{E}}_{\ell}$ has finer structures:
\begin{equation}
    \mathcal{E}_{\ell} = \prod_x \mathcal{E}_{\ell, x}, \quad \tilde{\mathcal{E}}_{\ell} = \prod_x \tilde{\mathcal{E}}_{\ell, x}, \quad \text{$ \mathcal{E}_{\ell, x} $ and $ \tilde{\mathcal{E}}_{\ell, x} $ are local} . 
\end{equation}
where $| \mathrm{Supp} (\mathcal{E}_{\ell, x}) |, | \mathrm{Supp} (\tilde{\mathcal{E}}_{\ell, x}) | \leq c$. 
Our Main Theorem \ref{thm:main} gives a novel, efficient quantum generative model:

\begin{corollary}[Efficient quantum generative model]
    \label{cor:QGM}
    Suppose the desired distribution is encoded in a trivial phase mixed state $\rho$. 
    Given arbitrary $\varepsilon > 0$ and $\delta > 0$, there exists an algorithm that runs time $T$, learns from $M$ many samples of $P$, and generates $P$ using a $(k + 1)$-layer local noisy channel circuit $W$, where
    \begin{equation}
        T = n \cdot 2^{O (c d)^k} \log \left( \frac{n}{\varepsilon} \right) + O (M), \quad M = \frac{n^2 \cdot 2^{O (c d)^k}}{\varepsilon^2} \log \left( \frac{n}{\delta} \right) .
    \end{equation}
    such that $\| W(\ket{0}\!\bra{0}^{\otimes n}) - \rho \|_1 \leq \varepsilon$ with success probability at least $1 - \delta$.
\end{corollary}

We also remark that Corollary \ref{cor:QGM} also works for the well-known \textit{classical diffusion models} \cite{sohl_2015_deep, anderson_1982_reverse}, because classical diffusion models are simply the decoherence limits of the quantum diffusion models \cite{hu_2025_local}.
To be more specific, suppose each site of the lattice $\Lambda$ supports a classical random variable from a sample space $\mathcal{X}$ (for simplicity, we assume $\mathcal{X}$ is discrete).
A probability distribution can be treated as being encoded in a diagonal mixed state $\rho = \mathrm{diag} (P) := \sum_i P_i \ket{i}\!\bra{i}$.
A \textit{classical noisy channel} (or simply \textit{noisy channel}) is given by the probability transition matrix $T$ such that $T (P)_i = \sum_j P_j T_{j i}$.
We can also embed the classical noisy channel into the quantum channel by defining
\begin{equation}
    \mathcal{E} (X) = \sum_{i j} K_{i j} X K^{\dagger}_{i j}, \quad K_{i j} = T_{j i}^{\frac{1}{2}} \ket{i}\!\bra{j},
\end{equation}
where $K_{i j} \in \mathbb{R}^{| \mathcal{X} |^n \times | \mathcal{X} |^n}$ are the Kraus operators satisfying the normalization relation $\sum_{i j} K^{\dagger}_{i j} K_{i j} =\mathbb{I}_{| \mathcal{X} |^n}$. We can verify in Appendix \ref{sec:diffusion} that $\mathcal{E}$ acting on $\rho$ indeed realize the classical noisy channel $\mathcal{E} (\mathrm{diag} (P)) = \mathrm{diag} (T(P))$.
Under this classical limit, our Main Theorem \ref{thm:main} also immediately gives a novel efficient classical diffusion model (for a formal statement, we refer to Appendix \ref{sec:diffusion}):

\begin{corollary}[Efficient classical diffusion model, informal]
    \label{cor:learning_TPD_informal}
    For any distribution $P$ that lies in the trivial phase, it satisfies the approximate Markovianity and local extendibility.
    Furthermore, given arbitrary $\varepsilon > 0$ and $\delta > 0$, there exists an algorithm that runs time $T$, learns from $M$ many samples of $P$, and generates $P$ using a $(k + 1)$-layer local noisy channel circuit $W$, where (recall $\mathcal{X}$ is a discrete sample space)
    \begin{equation}
        T = n \cdot | \mathcal{X} |^{O (c d)^k} \log \left( \frac{n}{\varepsilon} \right) + O (M), \quad M = \frac{n^2 \cdot | \mathcal{X} |^{O (c d)^k}}{\varepsilon^2} \log \left( \frac{n}{\delta} \right) .
    \end{equation}
    such that the total variation distance $\frac{1}{2}| W(Q) - P |_1 \leq \varepsilon$ with success probability at least $1 - \delta$.
\end{corollary}

\section*{Acknowledgement}

F.\,H.\,would like to thank Shengqi Sang for fruitful technical discussions.
The authors also acknowledge Chi-Fang Chen, Sitan Chen, Timothy H.\,Hsieh, Liang Jiang, Su-un Lee, Guangkuo Liu, Yunchao Liu, Jonathan Wurtz, and Yifan F.\,Zhang for stimulating discussions about the work.
F.\,H.\,acknowledges support from the QuEra Quantum Innovation Postdoctoral Fellowship.
W.\,G.\,acknowledges support from NSF Grant CCF-2430375 and the Von Neumann Award from Harvard Computer Science.
X.\,G.\,acknowledges support from NSF-CCF-2534807 and NSF PFC grant No.\,PHYS 2317149.

% \bibliographystyle{unsrt}
% \bibliography{ref}
\printbibliography

% \newpage

\appendix

\section{Proofs of Backward Lightcone Decompositions}
\label{sec:backward_lightcone_decomposition}

In this appendix, we provide the proof of backward lightcone decompositions (Fact \ref{fac:LCD-1} and Fact \ref{fac:LCD-2}) given in Section \ref{sec:backward_lightcone} of the main text.

\begin{proof}[Proof of Fact \ref{fac:LCD-1}]
    Let $\mathcal{B}_S =\mathcal{B}_{S, d} \circ \cdots \circ \mathcal{B}_{S, 1}$, $\mathcal{B}_{S_1} =\mathcal{B}_{S_1, d} \circ \cdots \circ \mathcal{B}_{S_1, 1}$ and $\mathcal{B}_{S_2} =\mathcal{B}_{S_2, d} \circ \cdots \circ \mathcal{B}_{S_2, 1}$. We first prove a claim: for any $\mathcal{E}_{\ell, x} \in \mathcal{B}_{S, \ell}$, if $\mathcal{E}_{\ell, x} \notin \mathcal{B}_{S_1, \ell}$ then we must have $\mathcal{E}_{\ell, x} \in \mathcal{B}_{S_2, \ell}$. We prove this claim by induction.
    \begin{itemize}
        \item For $\ell = d$, suppose a gate $\mathcal{E}_{d, x} \in \mathcal{B}_{S, d}$ and $\mathcal{E}_{d, x} \notin \mathcal{B}_{S_1, d}$. According to the definition of backward lightcone, we have $\mathrm{Supp} (\mathcal{E}_{d, x}) \cap (S_1 \cup S_2) \neq \varnothing$ but $\mathrm{Supp} (\mathcal{E}_{d, x}) \cap S_1 = \varnothing$. This implies that $\mathrm{Supp} (\mathcal{E}_{d, x}) \cap S_2 \neq \varnothing$, namely, $\mathcal{E}_{d, x} \in \mathcal{B}_{S_2, d}$.
    
        \item Now suppose for any $\ell' \geq \ell + 1$ (where $1 \leq \ell \leq d - 1$), suppose we have: for any $\mathcal{E}_{\ell', x} \in \mathcal{B}_{S, \ell'}$, if $\mathcal{E}_{\ell', x} \notin \mathcal{B}_{S_1, \ell'}$ then we must have $\mathcal{E}_{\ell', x} \in \mathcal{B}_{S_2, \ell'}$. Now, consider a gate $\mathcal{E}_{\ell, x} \in \mathcal{B}_{S, \ell}$ and $\mathcal{E}_{\ell, x} \notin \mathcal{B}_{S_1, \ell}$. According to the definition of backward lightcone, we have $\mathrm{Supp} (\mathcal{E}_{\ell, x}) \cap (\mathrm{Supp} (\mathcal{B}_{S, d} \circ \cdots \circ \mathcal{B}_{S, \ell + 1}) \cup S) \neq \varnothing$. There are two cases, and in both cases, we always have $\mathcal{E}_{\ell, x} \in \mathcal{B}_{S_2, \ell}$:
        \begin{itemize}
            \item There exists an $\ell' \geq \ell + 1$ such that $\mathrm{Supp} (\mathcal{E}_{\ell, x}) \cap \mathrm{Supp} (\mathcal{B}_{S, \ell'}) \neq \varnothing$. Let $\mathrm{Supp} (\mathcal{E}_{\ell, x}) \cap \mathrm{Supp} (\mathcal{E}_{\ell', x'}) \neq \varnothing$ for some position index $x'$. However, since $\mathcal{E}_{\ell, x} \notin \mathcal{B}_{S_1, \ell}$ then $\mathrm{Supp} (\mathcal{E}_{\ell, x}) \cap \mathrm{Supp} (\mathcal{B}_{S_1, \ell'}) = \varnothing$ and hence $\mathcal{E}_{\ell', x'} \notin \mathcal{B}_{S_1, \ell'}$. By the induction hypothesis, we have $\mathcal{E}_{\ell, x} \in \mathcal{B}_{S_2, \ell'}$, namely $\mathrm{Supp} (\mathcal{E}_{\ell, x}) \cap \mathrm{Supp} (\mathcal{B}_{S_2, \ell'}) \neq \varnothing$ for this $\ell' \geq \ell + 1$. That is, $\mathcal{E}_{\ell, x} \in \mathcal{B}_{S_2, \ell}$.
      
            \item $\mathrm{Supp} (\mathcal{E}_{\ell, x}) \cap S \neq \varnothing$. However, since $\mathcal{E}_{\ell, x} \notin \mathcal{B}_{S_1, \ell}$, then $\mathrm{Supp} (\mathcal{E}_{\ell, x}) \cap S_1 = \varnothing$ and hence $\mathrm{Supp} (\mathcal{E}_{\ell, x}) \cap S_2 \neq \varnothing$. That is, $\mathcal{E}_{\ell, x} \in \mathcal{B}_{S_2, \ell}$.
        \end{itemize}
    Our claim follows immediately by induction.
    \end{itemize}
    On the other hand, we also claim that for any $\mathcal{E}_{\ell, x} \in \mathcal{B}_{S_2, \ell}$, we must have $\mathcal{E}_{\ell, x} \in \mathcal{B}_{S, \ell}$. This claim is simply from the monotonicity of the backward lightcone for $S_2 \subseteq S$. According to the definition of $\mathcal{B}_{S_2} \backslash \mathcal{B}_{S_1}$, we have $\mathcal{Q}=\mathcal{B}_{S_2} \backslash \mathcal{B}_{S_1} =\mathcal{Q}_d \circ \cdots \circ \mathcal{Q}_1$. According to our two claims above, we have
    \begin{equation}
        \mathcal{Q}_{\ell} = \prod_{\mathcal{E}_{\ell, x} \notin \mathcal{B}_{S_1, \ell}, \mathcal{E}_{\ell, x} \in \mathcal{B}_{S_2, \ell}} \mathcal{E}_{\ell, x} = \prod_{\mathcal{E}_{\ell, x} \notin \mathcal{B}_{S_1, \ell}, \mathcal{E}_{\ell, x} \in \mathcal{B}_{S, \ell}} \mathcal{E}_{\ell, x} .
    \end{equation}
    Notice that for any $\ell \in [d]$, by the definition of backward lightcone, for any $\mathcal{E}_{\ell, x} \notin \mathcal{B}_{S_1, \ell}$, we have $\mathrm{Supp} (\mathcal{E}_{\ell, x}) \cap \mathrm{Supp} (\mathcal{B}_{S_1, d} \circ \cdots \circ \mathcal{B}_{S_1, \ell + 1}) = \varnothing$. Since $\mathrm{Supp} (\mathcal{Q}_{\ell}) = \cup_{\mathcal{E}_{\ell, x} \notin \mathcal{B}_{S_1, \ell}, \mathcal{E}_{\ell, x} \in \mathcal{B}_{S, \ell}} \mathrm{Supp} (\mathcal{E}_{\ell, x})$, we have $\mathrm{Supp} (\mathcal{Q}_{\ell}) \cap \mathrm{Supp} (\mathcal{B}_{S_1, d} \circ \cdots \circ \mathcal{B}_{S_1, \ell + 1}) = \varnothing$. It means $\mathcal{Q}_{\ell}$ and $\mathcal{B}_{S_1, d} \circ \cdots \circ\mathcal{B}_{S_1, \ell + 1}$ commute. Moreover, by the definition of $\mathcal{Q}_{\ell}$, we have $\mathcal{Q}_{\ell} \circ \mathcal{B}_{S_1, \ell} = \left( \prod_{\mathcal{E}_{\ell, x} \notin \mathcal{B}_{S_1, \ell}, \mathcal{E}_{\ell, x} \in \mathcal{B}_{S, \ell}} \mathcal{E}_{\ell, x} \right) \left( \prod_{\mathcal{E}_{\ell, x} \in \mathcal{B}_{S_1, \ell}} \mathcal{E}_{\ell, x} \right) = \prod_{\mathcal{E}_{\ell, x} \in \mathcal{B}_{S, \ell}} \mathcal{E}_{\ell, x} =\mathcal{B}_{S, \ell}$. This yields:
    \begin{eqnarray}
        (\mathcal{B}_{S_2} \backslash \mathcal{B}_{S_1}) \circ \mathcal{B}_{S_1} & = & \mathcal{Q}_d \circ \cdots \circ \mathcal{Q}_2 \circ \mathcal{Q}_1 \circ \mathcal{B}_{S_1, d} \circ \cdots \circ \mathcal{B}_{S_2, 1} \circ \mathcal{B}_{S_1, 1} \nonumber\\
        & = & \mathcal{Q}_d \circ \cdots \circ \mathcal{Q}_2 \circ \mathcal{B}_{S_1, d} \circ \cdots \circ \mathcal{B}_{S_2, 1} \circ (\mathcal{Q}_1 \circ \mathcal{B}_{S_1, 1}) . 
    \end{eqnarray}
    By induction, we have
    \begin{equation}
        \mathcal{Q} \circ \mathcal{B}_{S_1} = (\mathcal{Q}_d \circ \mathcal{B}_{S_1, d}) \circ \cdots \circ (\mathcal{Q}_1 \circ \mathcal{B}_{S_1, 1}) =\mathcal{B}_{S, d} \circ \cdots \circ \mathcal{B}_{S, 1} =\mathcal{B}_S .
    \end{equation}
    Furthermore, for any $\mathcal{E}_{\ell, x} \in \mathcal{Q}$, we have proved that $\mathrm{Supp} (\mathcal{E}_{\ell, x}) \subseteq \mathrm{Supp} (\mathcal{B}_{S_2})$ and $\mathrm{Supp} (\mathcal{E}_{\ell, x}) \cap S_1$, by the two claims we made above. Therefore, $\mathrm{Supp} (\mathcal{E}_{\ell, x}) \subseteq \mathrm{Supp} (\mathcal{B}_{S_2}) \backslash S_1 \subseteq S_2 (s) \backslash S_1$.
\end{proof}

\begin{proof}[Proof of Fact \ref{fac:LCD-2}]
    We set $S \rightarrow \Lambda$, $S_1 \rightarrow \bar{S}$ and $S_2 \rightarrow S$, we immediately have $\mathcal{E}=\mathcal{B}_{\Lambda} =\mathcal{Q}_S \circ \mathcal{B}_{\bar{S}}$ where $\mathcal{Q}_S =\mathcal{B}_S \backslash \mathcal{B}_{\bar{S}}$. Furthermore, $\mathrm{Supp} (\mathcal{Q}_S) \subseteq S (s) \backslash \bar{S} = S$.
\end{proof}

\begin{proof}[Proof of Fact \ref{fac:LCD-3}]
    The reduced density matrix of $\rho_{S_1 S_2}$ is entirely determined by the backward lightcone $\mathcal{B}_{S_1 S_2}$. Since $\mathrm{dist}(S_1, S_2) \geq 2s$, then $\mathrm{Supp}(\mathcal{B}_{S_1}) \cap \mathrm{Supp}(\mathcal{B}_{S_2}) = \varnothing$. We can directly compute $\rho_{S_1 S_2}$:
    \begin{align}
        \rho_{S_1 S_2} & = \mathrm{Tr}_{\overline{S_1 S_2}}(\mathcal{B}_{S_1 S_2}(\rho_0)) = \mathrm{Tr}_{\overline{S_1 S_2}}(\mathcal{B}_{S_1} \mathcal{B}_{S_2} (\rho_0)) \nonumber\\
        & = \mathrm{Tr}_{\bar{S_1}}(\mathcal{B}_{S_1} (\rho_0)) \otimes \mathrm{Tr}_{\bar{S_2}}(\mathcal{B}_{S_2} (\rho_0)) = \rho_{S_1} \otimes \rho_{S_2}, 
    \end{align}
    which completes the proof,.
\end{proof}

\section{Local Reversibility is an Equivalence Relation}
\label{sec:equivalence_relation}

\begin{theorem}[Equivalence relation]
  Local reversibility satisfies the three properties of an equivalence relation:
  \begin{enumerate}
    \item Reflexivity: $\rho$ is locally reversible from
    $\rho$.
    
    \item Symmetry: if $\sigma$ is locally reversible from $\rho$, then $\rho$ is locally reversible
    from $\sigma$.
    
    \item Transitivity: if $\sigma$ is locally reversible from $\rho$, and $\tau$ is locally reversible from $\rho$, then $\sigma$ is locally reversible from $\tau$.
  \end{enumerate}
\end{theorem}

\begin{proof}[Proof of Reflexivity]
    We simply take $\mathcal{E}=\mathcal{I}$, the identity map, then we have $\rho =\mathcal{I} (\rho)$.
\end{proof}

\begin{proof}[Proof of Symmetry]
    Let $\sigma$ be locally reversible from $\rho$. This means:
    For any $\eta > 0$ and $\varepsilon_{\mathrm{LR}} > 0$, there exists a shallow channel $\mathcal{E}=\mathcal{E}_d \circ \cdots \circ \mathcal{E}_2 \circ \mathcal{E}_1$ where $\mathcal{E}_{\ell} = \prod_x \mathcal{E}_{\ell, x}$ and each $\mathcal{E}_{\ell, x}$ has a local support, such that $\| \rho -\mathcal{E} (\sigma) \|_1 \leq \eta$ and $\sigma$ is $\varepsilon_{\mathrm{LR}}$-locally reversible from $\mathcal{E} (\sigma)$ such that $\| \tilde{\mathcal{E}}_{\ell, x} \circ \mathcal{E}_{\ell, x} \circ \mathcal{E}_{\leq \ell - 1} (\sigma) -\mathcal{E}_{\leq \ell - 1} (\sigma) \|_1 \leq \varepsilon_{\mathrm{LR}}$ where $\mathcal{E}_{\leq \ell - 1} =\mathcal{E}_{\ell - 1} \circ \cdots \circ \mathcal{E}_2 \circ \mathcal{E}_1$.

    Now, given $\eta' > 0$ and $\varepsilon_{\mathrm{LR}}' > 0$, we define that (notice the ordering)
    \begin{equation}
        \mathcal{S}=\mathcal{S}_d \circ \cdots \circ \mathcal{S}_2 \circ \mathcal{S}_1, \quad \mathcal{S}_{\ell} = \prod_x \tilde{\mathcal{E}}_{d - \ell + 1, x} ,
    \end{equation}
    and we claim that
    \begin{enumerate}
        \item $\| \sigma -\mathcal{S} (\rho) \|_1 \leq \eta'$,
        \item $\rho$ is $\varepsilon'_{\mathrm{LR}}$-locally reversible from $\mathcal{S} (\rho)$ along the circuit $\mathcal{S}$. 
    \end{enumerate}
    Let us start with the first claim. Since $\sigma$ is $\varepsilon_{\mathrm{LR}}$-locally reversible from $\mathcal{E} (\sigma)$, according to Lemma \ref{lem:LI}, we have $\| \sigma -\mathcal{S} (\mathcal{E} (\sigma)) \|_1 \leq n d \cdot \varepsilon_{\mathrm{LR}}$. It implies that
    \begin{equation}
        \| \sigma -\mathcal{S} (\rho) \|_1 \overset{\text{(i)}}{\leq} \| \sigma -\mathcal{S} (\mathcal{E} (\sigma)) \|_1 + \| \mathcal{S} (\mathcal{E} (\sigma)) -\mathcal{S} (\rho) \|_1 \overset{\text{(ii)}}{\leq} n d \cdot \varepsilon_{\mathrm{LR}} + \| \mathcal{E} (\sigma) - \rho \|_1 \overset{\text{(iii)}}{\leq} n d \cdot \varepsilon_{\mathrm{LR}} + \eta,
    \end{equation}
    where inequality (i) is from the triangle inequality; inequality (ii) is from the contractivity of the CPTP map $\mathcal{S}$; and inequality (iii) is by $\| \rho -\mathcal{E} (\sigma) \|_1 \leq \eta$.
    
    Next, let us prove the second claim. Consider the first $\ell - 1$ layer of $\mathcal{S}$, namely $\mathcal{S}_{\leq \ell - 1} =\mathcal{S}_{\ell - 1} \circ \cdots \circ \mathcal{S}_2 \circ \mathcal{S}_1$ (recall $\mathcal{S}_{\ell - 1} = \prod_x \tilde{\mathcal{E}}_{d - \ell + 2, x}$).
    Notice that $\mathcal{E}_{\leq d - \ell + 1} (\sigma)$ is $\varepsilon_{\mathrm{LR}}$-locally reversible from $\mathcal{E} (\sigma)$ along $\mathcal{E}_d \circ \mathcal{E}_{d - 1} \circ \cdots \circ \mathcal{E}_{d - \ell + 2}$. 
    Then according to Lemma \ref{lem:LI}, we have 
    \begin{equation}
        \| \mathcal{S}_{\leq \ell - 1} \circ \mathcal{E} (\sigma) -\mathcal{E}_{\leq d - \ell + 1} (\sigma) \|_1 \leq n d \cdot \varepsilon_{\mathrm{LR}}
    \end{equation}
    This implies that (recall $\mathcal{E}_{\leq d - \ell + 1} =\mathcal{E}_{d - \ell + 1} \circ \cdots \circ \mathcal{E}_2 \circ \mathcal{E}_1$):
    \begin{eqnarray}
        \| \mathcal{E}_{\leq d - \ell + 1} (\sigma) -\mathcal{S}_{\leq \ell - 1} (\rho) \|_1 & \overset{\text{(iii)}}{\leq} & \| \mathcal{E}_{\leq d - \ell + 1} (\sigma) -\mathcal{S}_{\leq \ell - 1} \circ \mathcal{E} (\sigma) \|_1 + \| \mathcal{S}_{\leq \ell - 1} \circ \mathcal{E} (\sigma) -\mathcal{S}_{\leq \ell - 1} (\rho) \|_1 \nonumber\\
        & \overset{\text{(iv)}}{\leq} & n d \cdot \varepsilon_{\mathrm{LR}} + \| \mathcal{E} (\sigma) - \rho \|_1 \overset{\text{(v)}}{\leq} n d \cdot \varepsilon_{\mathrm{LR}} + \eta . 
    \end{eqnarray}
    where inequality (iii) is from the triangle inequality; inequality (iv) is from the contractivity of the CPTP map $\mathcal{S}_{\leq \ell - 1}$; and the inequality (v) is by $\| \rho -\mathcal{E} (\sigma) \|_1 \leq \eta$.
    
    For any given gate position index $x$ in the $\ell$-th layer, we decompose $\mathcal{E}_{d - \ell + 1} =\mathcal{E}_{d - \ell + 1, x} \circ \mathcal{C}$ where $\mathcal{C} = \prod_{x' \neq x} \mathcal{E}_{d - \ell + 1, x'}$ such that $\mathrm{Supp} (\mathcal{C}) \cap \mathrm{Supp} (\mathcal{E}_{\ell, x}) = \varnothing$. This implies that $\mathcal{E}_{\leq d - \ell + 1} =\mathcal{E}_{d - \ell + 1, x} \circ \mathcal{C} \circ \mathcal{E}_{\leq d - \ell}$.
    
    By the triangle inequality, we can bound
    \begin{align}
        & \| \mathcal{E}_{d - \ell + 1, x} \circ \tilde{\mathcal{E}}_{d - \ell + 1, x} \circ \mathcal{S}_{\leq \ell - 1} (\rho) -\mathcal{S}_{\leq \ell - 1} (\rho) \|_1 \nonumber\\
        \leq~ & \| \mathcal{E}_{d - \ell + 1, x} \circ \tilde{\mathcal{E}}_{d - \ell + 1, x} \circ \mathcal{S}_{\leq \ell - 1} (\rho) -\mathcal{E}_{d - \ell + 1, x} \circ \tilde{\mathcal{E}}_{d - \ell + 1, x} \circ \mathcal{E}_{\leq d - \ell + 1} (\sigma) \|_1 \nonumber\\
        & + \| \mathcal{E}_{d - \ell + 1, x} \circ \tilde{\mathcal{E}}_{d - \ell + 1, x} \circ \mathcal{E}_{\leq d - \ell + 1} (\sigma) -\mathcal{E}_{\leq d - \ell + 1} (\sigma) \|_1 + \| \mathcal{E}_{\leq d - \ell + 1} (\sigma) -\mathcal{S}_{\leq \ell - 1} (\rho) \|_1 . \label{eq:symmetry}
    \end{align}
    The first term in Eq.\,(\ref{eq:symmetry}) is further bounded by $\| \mathcal{S}_{\leq \ell - 1} (\rho) -\mathcal{E}_{\leq d - \ell + 1} (\sigma) \|_1 \leq n d \cdot \varepsilon_{\mathrm{LR}} + \eta$ using the contractivity of the CPTP map $\mathcal{E}_{d - \ell + 1, x} \circ \tilde{\mathcal{E}}_{d - \ell + 1, x}$. The second term in Eq.\,(\ref{eq:symmetry}) is bounded by
    \begin{eqnarray}
        &  & \| \mathcal{E}_{d - \ell + 1, x} \circ \tilde{\mathcal{E}}_{d - \ell + 1, x} \circ \mathcal{E}_{\leq d - \ell + 1} (\sigma) -\mathcal{E}_{\leq d - \ell + 1} (\sigma) \|_1 \nonumber\\
        & \overset{\text{(vi)}}{=} & \| \mathcal{E}_{d - \ell + 1, x} \circ \tilde{\mathcal{E}}_{d - \ell + 1, x} \circ \mathcal{E}_{d - \ell + 1, x} \circ \mathcal{C} \circ \mathcal{E}_{\leq d - \ell} (\sigma) -\mathcal{E}_{d - \ell + 1, x} \circ \mathcal{C} \circ \mathcal{E}_{\leq d - \ell} (\sigma) \|_1 \nonumber\\
        & \overset{\text{(vii)}}{\leq} & \| \tilde{\mathcal{E}}_{d - \ell + 1, x} \circ \mathcal{E}_{d - \ell + 1, x} \circ \mathcal{C} \circ \mathcal{E}_{\leq d - \ell} (\sigma) -\mathcal{C} \circ \mathcal{E}_{\leq d - \ell} (\sigma) \|_1 \overset{\text{(viii)}}{\leq} \varepsilon_{\mathrm{LR}} ,
    \end{eqnarray}
    where equality (vi) is from the decomposition $\mathcal{E}_{d - \ell + 1, x} \circ \mathcal{C} \circ \mathcal{E}_{\leq d - \ell}$; inequality (vii) is from the contractivity of the CPTP map $\mathcal{E}_{d - \ell + 1, x}$; and inequality
    (viii) is from Lemma \ref{lem:LR-modify}. 
    The third term in Eq.\,(\ref{eq:symmetry}) is bounded by $\| \mathcal{S}_{\leq \ell - 1} (\rho) -\mathcal{E}_{\leq d - \ell + 1} (\sigma) \|_1 \leq n d \cdot \varepsilon_{\mathrm{LR}} + \eta$. 
    Combining these results, we have
    \begin{equation}
        \| \mathcal{E}_{d - \ell + 1, x} \circ \tilde{\mathcal{E}}_{d - \ell + 1, x} \circ \mathcal{S}_{\leq \ell - 1} (\rho) -\mathcal{S}_{\leq \ell - 1} (\rho) \|_1 \leq (2 n d + 1) \cdot \varepsilon_{\mathrm{LR}} + 2 \eta .
    \end{equation}
    Finally, by setting
    \begin{equation}
        \varepsilon_{\mathrm{LR}} = \min \left( \frac{\varepsilon_{\mathrm{LR}}'}{2 (2 n d + 1)}, \frac{\eta'}{2 n d} \right), \quad \eta = \min \left( \frac{\varepsilon_{\mathrm{LR}}'}{4}, \frac{\eta'}{2} \right),
    \end{equation}
    and taking the corresponding $\mathcal{E}$ and $\mathcal{S}$, we can ensure that
    \begin{align}
        \| \sigma -\mathcal{S} (\rho) \|_1 & \leq n d \cdot \varepsilon_{\mathrm{LR}} + \eta \leq \eta', \nonumber\\
        \| \mathcal{E}_{d - \ell + 1, x} \circ \tilde{\mathcal{E}}_{d - \ell + 1, x} \circ \mathcal{S}_{\leq \ell - 1} (\rho) -\mathcal{S}_{\leq \ell - 1} (\rho) \|_1 & \leq (2 n d + 1) \cdot \varepsilon_{\mathrm{LR}} + 2 \eta \leq \varepsilon_{\mathrm{LR}}' .
    \end{align}
    This completes the proof.
\end{proof}

\begin{proof}[Proof of Transitivity]
    Let $\sigma$ is locally reversible from $\rho$, and $\tau$ is locally reversible from $\sigma$. This means:
    For any $\eta_1, \eta_2 > 0$ and $\varepsilon_{\mathrm{LR}, 1}, \varepsilon_{\mathrm{LR}, 2} > 0$, there exists shallow channels $\mathcal{E}=\mathcal{E}_{d_1} \circ \cdots \circ \mathcal{E}_2 \circ \mathcal{E}_1$ and $\mathcal{F}=\mathcal{F}_{d_2} \circ \cdots \circ \mathcal{F}_2 \circ \mathcal{F}_1$. Here, \ $\mathcal{E}_{\ell} = \prod_x \mathcal{E}_{\ell, x}$ and $\mathcal{F}_{\ell} = \prod_x \mathcal{F}_{\ell, x}$, and each $\mathcal{E}_{\ell, x}$ and $\mathcal{F}_{\ell, x}$ has a local support. We have $\| \rho -\mathcal{E} (\sigma) \|_1 \leq \eta_1$ and $\| \tau -\mathcal{F} (\rho) \|_1 \leq \eta_2$; $\sigma$ is $\varepsilon_{\mathrm{LR}, 1}$-locally reversible from $\mathcal{E} (\sigma)$ (namely $\| \tilde{\mathcal{E}}_{\ell, x} \circ \mathcal{E}_{\ell, x} \circ \mathcal{E}_{\leq \ell - 1} (\sigma) -\mathcal{E}_{\leq \ell - 1} (\sigma) \|_1 \leq \varepsilon_{\mathrm{LR}, 1}$ where $\mathcal{E}_{\leq \ell - 1} =\mathcal{E}_{\ell - 1} \circ \cdots \circ \mathcal{E}_2 \circ \mathcal{E}_1$) ; and $\rho$ is $\varepsilon_{\mathrm{LR}, 2}$-locally reversible from $\mathcal{F} (\rho)$ (namely $\| \tilde{\mathcal{F}}_{\ell, x} \circ \mathcal{F}_{\ell, x} \circ \mathcal{F}_{\leq \ell - 1} (\sigma) -\mathcal{F}_{\leq \ell - 1} (\sigma) \|_1 \leq \varepsilon_{\mathrm{LR}, 2}$ where $\mathcal{F}_{\leq \ell - 1} =\mathcal{F}_{\ell - 1} \circ \cdots \circ \mathcal{F}_2 \circ \mathcal{F}_1$).
  
    Now, given $\eta' > 0$ and $\varepsilon_{\mathrm{LR}}' > 0$, we define that
    \begin{equation}
        \mathcal{T}=\mathcal{F} \circ \mathcal{E} =: \mathcal{T}_{d_1 + d_2} \circ \cdots \circ \mathcal{T}_2 \circ \mathcal{T}_1,
    \end{equation}
    and we claim that
    \begin{enumerate}
        \item $\| \tau -\mathcal{T} (\sigma) \|_1 \leq \eta'$,
        \item $\sigma$ is $\varepsilon'_{\mathrm{LR}}$-locally reversible from $\mathcal{T} (\tau)$ along the circuit $\mathcal{T}$. 
    \end{enumerate}
    The first claim is straightforward:
    \begin{equation}
        \| \tau -\mathcal{T} (\sigma) \|_1 \overset{\text{(i)}}{\leq} \| \tau -\mathcal{F} (\rho) \|_1 + \| \mathcal{F} (\rho) -\mathcal{F} \circ \mathcal{E} (\sigma) \|_1 \overset{\text{(ii)}}{\leq} \eta_2 + \| \rho -\mathcal{E} (\sigma) \|_1 \overset{\text{(iii)}}{\leq} \eta_1 + \eta_2 .
    \end{equation}
    where inequality (i) is from the triangle inequality; inequality (ii) is from the contractivity of the CPTP map $\mathcal{F}$, and also $\| \tau -\mathcal{F} (\rho) \|_1 \leq \eta_2$; and the inequality (iii) is by $\| \rho -\mathcal{E} (\sigma) \|_1 \leq \eta_1$.
  
    Next, let us prove the second claim. Consider the first $\ell - 1$ layer of $\mathcal{T}$, namely $\mathcal{T}_{\leq \ell - 1} =\mathcal{T}_{\ell - 1} \circ \cdots \circ \mathcal{T}_2 \circ \mathcal{T}_1$. If $\ell \leq d_1$, then $\mathcal{T}_{\leq \ell - 1} =\mathcal{E}_{\leq \ell - 1}$ and we automatically have $\| \tilde{\mathcal{E}}_{\ell, x} \circ \mathcal{E}_{\ell, x} \circ \mathcal{E}_{\leq \ell - 1} (\sigma) -\mathcal{E}_{\leq \ell - 1} (\sigma) \|_1 \leq \varepsilon_{\mathrm{LR}, 1}$. Now, we only consider the case where $\ell \geq d_1 + 1$. In this case, $\mathcal{T}_{\leq \ell - 1} =\mathcal{F}_{\leq \ell - d_1 - 1} \circ \mathcal{E}$.
  
    By the triangle inequality, we can bound
    \begin{align}
        & \| \tilde{\mathcal{F}}_{\ell - d_1, x} \circ \mathcal{F}_{\ell - d_1, x} \circ \mathcal{T}_{\leq \ell - 1} (\sigma) -\mathcal{T}_{\leq \ell - 1} (\sigma) \|_1 \nonumber\\
        \leq~& \| \tilde{\mathcal{F}}_{\ell - d_1, x} \circ \mathcal{F}_{\ell - d_1, x} \circ \mathcal{F}_{\leq \ell - d_1 - 1} \circ \mathcal{E} (\sigma) - \tilde{\mathcal{F}}_{\ell - d_1, x} \circ \mathcal{F}_{\ell - d_1, x} \circ \mathcal{F}_{\leq \ell - d_1 - 1} (\rho) \|_1 \nonumber\\
        & + \| \tilde{\mathcal{F}}_{\ell - d_1, x} \circ \mathcal{F}_{\ell - d_1, x} \circ \mathcal{F}_{\leq \ell - d_1 - 1} (\rho) -\mathcal{F}_{\leq \ell - d_1 - 1} (\rho) \|_1 \nonumber\\
        & + \| \mathcal{F}_{\leq \ell - d_1 - 1} (\rho) -\mathcal{F}_{\leq \ell - d_1 - 1} \circ \mathcal{E} (\sigma) \|_1 . \label{eq:transitivity}
    \end{align}
    The first term in Eq.\,(\ref{eq:transitivity}) is further bounded by $\| \mathcal{E} (\sigma) - \rho \|_1 \leq \eta_1$ using the contractivity of CPTP map $\tilde{\mathcal{F}}_{\ell - d_1, x} \circ \mathcal{F}_{\ell - d_1, x} \circ \mathcal{F}_{\leq \ell - d_1 - 1}$. The second term in Eq.\,(\ref{eq:transitivity}) bounded by $\varepsilon_{\mathrm{LR}, 2}$ because $\rho$ is $\varepsilon_{\mathrm{LR}, 2}$-locally reversible from $\mathcal{F} (\rho)$. The third term in Eq.\,(\ref{eq:transitivity}) is bounded by $\| \rho -\mathcal{E} (\sigma) \|_1 \leq \eta_1$. 
    Combine there results together, we have
    \begin{equation}
        \| \tilde{\mathcal{F}}_{\ell - d_1, x} \circ \mathcal{F}_{\ell - d_1, x} \circ \mathcal{T}_{\leq \ell - 1} (\sigma) -\mathcal{T}_{\leq \ell - 1} (\sigma) \|_1 \leq \varepsilon_{\mathrm{LR}, 2} + 2 \eta_1 .
    \end{equation}
    Then, by setting
    \begin{equation}
        \varepsilon_{\mathrm{LR}, 1} = \varepsilon_{\mathrm{LR}}', \quad \varepsilon_{\mathrm{LR}, 2} = \frac{\varepsilon_{\mathrm{LR}}'}{2}, \quad \eta_1 = \min \left( \frac{\varepsilon_{\mathrm{LR}}'}{4}, \frac{\eta'}{2} \right), \quad \eta_2 = \frac{\eta'}{2} .
    \end{equation}
    and taking the corresponding $\mathcal{E}$ and $\mathcal{F}$, we can ensure that
    \begin{align}
        \| \tau -\mathcal{T} (\sigma) \|_1 & \leq \eta_1 + \eta_2 \leq \eta', \\
        \| \tilde{\mathcal{T}}_{\ell, x} \circ \mathcal{T}_{\ell, x} \circ \mathcal{T}_{\leq \ell - 1} (\sigma) -\mathcal{T}_{\leq \ell - 1} (\sigma) \|_1 & \leq \varepsilon_{\mathrm{LR}, 2} + 2 \eta_1 \leq \varepsilon_{\mathrm{LR}}' . 
    \end{align}
    This completes the proof.
\end{proof}

\section{Mixed State Phases via Conditional Mutual Information}
\label{sec:cmi_decay}

For the tripartition $\Lambda=ABC$ of $\rho$ shown in Figure \ref{fig:partition}a, we let with the thickness of $B$ being $\mathrm{dist}(A, C)=r$. 
Now, suppose any channel that acts locally on $C$, the \textit{twirled Petz map} acting on $BC$ is defined as follows:
\begin{equation}
    \tilde{\mathcal{E}}_{BC} (X) : = \int \dd \theta f(\theta) \rho^{\frac{1-i \theta}{2}} \mathcal{N}_C^{\dagger}(\mathcal{N}_C (\rho)^{\frac{-1+i \theta}{2}} X \mathcal{N}_C (\rho)^{\frac{-1-i \theta}{2}}) \rho^{\frac{1+i \theta}{2}}. \label{eq:TPRM}
\end{equation}
where $f (\theta) = \pi / (2 \cosh (\pi \theta) + 2)$ is a probability distribution function of angles $\theta$. 

According to the Fawzi-Renner inequality \cite{fawzi_2015_quantum, junge_2018_universal, mark_2016_quantum}, the twirled Petz map $\tilde{\mathcal{E}}_{B C}$ has recovery property:
\begin{equation}
    \frac{1}{2 \ln 2} \| \tilde{\mathcal{E}}_{B C} \circ \mathcal{E}_C (\rho_{A B C}) - \rho_{A B C} \|_1^2 \leq I_\rho (A : C |  B) ,
\end{equation}
which tells us that the conditional mutual information (CMI) controls the recovery error of any local noise by using a (quasi)-local channel acting on $B C$. 
This observation leads to the following definition of Markov length to study the phases of quantum mixed states \cite{sang_2025_statbility, sang_2025_mixed}:
\begin{definition}[Markov length]
    The Markov length $\xi (\rho)$ of $\rho$ is the smallest length $\xi$ such that:
    \begin{equation}
        I (A : C |  B) \leq \gamma e^{- \mathrm{dist} (A, C) / \xi} .
    \end{equation}
    where $\gamma = \mathrm{poly} (n)$.
\end{definition}

This gives the mixed state phases used in quantum many-body physics (Definition 2 in Ref.\,\cite{sang_2025_mixed}):
\begin{definition}[Mixed state phases via finite Markov length]
    \label{def:MSP_via_FML}
    Given two mixed states $\rho$ and $\sigma$. Suppose there exists a local channel circuit $\mathcal{E}=\mathcal{E}_d \circ \cdots \mathcal{E}_2 \circ \mathcal{E}_1$ where each $\mathcal{E}_{\ell} = \prod_x \mathcal{E}_{\ell, x}$ is a layer of non-overlapping local channel gates $\mathcal{E}_{\ell, x}$) such that:
    \begin{equation}
        \rho =\mathcal{E} (\sigma),
    \end{equation}
    and, further, for any time-layer $\ell$ and any $\mathcal{C} \subseteq \mathcal{E}_{\ell}$ being a layer composed of a subset of channel gates in $\mathcal{E}_{\ell}$, $\rho'_{\ell} =\mathcal{C} \circ \mathcal{E}_{\ell - 1} \circ \cdots \circ \mathcal{E}_2 \circ \mathcal{E}_1 (\sigma)$ has finite Markov length, i.e., $\xi (\rho'_{\ell}) \leq \xi$ with some $\xi$ independent of $\ell$.
\end{definition}

As promised in Section \ref{sec:MSP_via_LR} of the main text, we now show that the mixed state phases via finite Markov length (Definition \ref{def:MSP_via_FML}) imply the mixed state phases via local reversibility (Definition \ref{def:MSP_via_LR}).

\begin{figure}
    \centering
    \includegraphics[width=\linewidth]{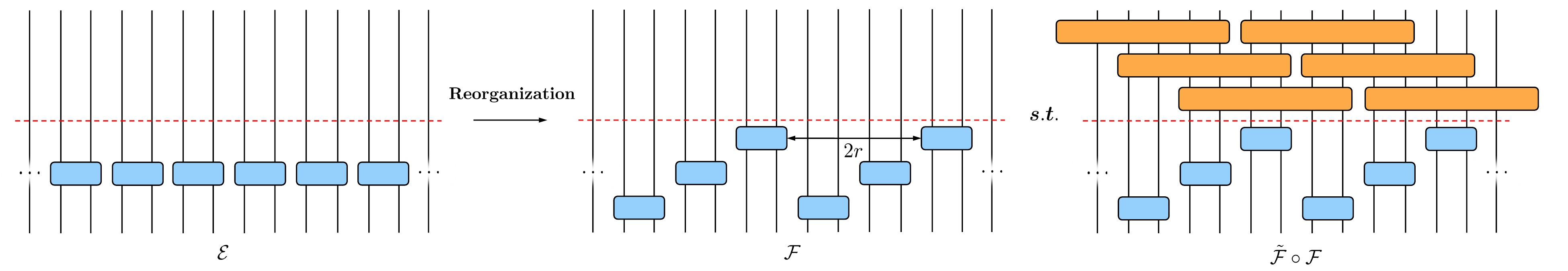}
    \caption{
    Reorganization trick that transfers from the original circuit $\mathcal{E}$ into a new circuit $\mathcal{F}$ (we only plot one layer for simplicity). All channel gates in each layer of $\mathcal{F}$ is separated by a distance $2r$, such that we can find a circuit $\tilde{\mathcal{F}}$ to locally reverse $\mathcal{F}$ with the same depth as $\mathcal{F}$.
    }
    \label{fig:reorganization}
\end{figure}

\begin{theorem}
    Given two mixed states $\rho$ and $\sigma$. Suppose $\rho$ can be prepared with a quasi-local channel circuit $\rho = \mathcal{E} (\sigma)$, which satisfies the condition in Definition \ref{def:MSP_via_FML} with a Markov length $\xi = O(1)$. 
    Assume $\mathcal{E}=\mathcal{E}_d \circ \cdots \mathcal{E}_2 \circ \mathcal{E}_1$ and each $\mathcal{E}_{\ell,x}$ has a support whose linear size is at most $l$.
    Then, there exists another quasi-local channel circuit $\mathcal{F}=\mathcal{F}_{d'} \circ \cdots \mathcal{F}_2 \circ \mathcal{F}_1$ such that $\rho = \mathcal{F} (\sigma)$ and each $\mathcal{F}_{\ell,x}$ has a support whose linear size is at most $l'$, where 
    \begin{equation}
        l' = l + 2r, \quad d' = d \cdot O(r^k),
    \end{equation}
    where $r = \xi \cdot O ( \mathrm{log} (n / \varepsilon_{\mathrm{LR}} ) )$, for arbitrary given $\varepsilon_{\mathrm{LR}} = O(1)$. 
\end{theorem}

\begin{proof}
    Let $\mathrm{Supp}(\mathcal{E}_{\ell, x})= C_{\ell, x}$. Now for each $\ell, x$, we define a tripartition $\Lambda=A_{\ell, x}B_{\ell, x}C_{\ell, x}$ such that the thickness of $B_{\ell, x}$ is denoted as $\mathrm{dist}(A_{\ell, x}, C_{\ell, x})=r$, where $r$ is the distance to be determined. 
    
    We use the reorganization trick (first introduced in Ref.\,\cite{sang_2025_statbility}) to rearrange $\mathcal{E}_{\ell, x}$ such that each $\mathcal{E}_{\ell, x}$ has been separated by a distance $2 r$ (see Figure \ref{fig:reorganization}).
    After the reorganization procedure, we can treat $\mathcal{E}_{\ell, x}$ as a new gate
    \begin{equation}
        \mathcal{F}_{\ell', x} := \mathcal{E}_{\ell, x} \otimes \mathcal{I}_{C_{\ell, x}(r) \backslash C_{\ell, x}},
    \end{equation}
    whose support is $\mathrm{Supp}(\mathcal{F}_{\ell', x})= C_{\ell, x}(r)$, where $\ell'$ is the new layer index of $\mathcal{E}_{\ell, x}$. 
    This procedure defines the new reorganized circuit $\mathcal{F}$
    \begin{equation}
        \mathcal{F} = \mathcal{F}_{d'} \circ \cdots \mathcal{F}_2 \circ \mathcal{F}_1.
    \end{equation}
    where $d' = d \cdot O(r^k)$. 
    
    Now, we can check the local reversibility from $\mathcal{F}(\rho)$ to $\rho$ along $\mathcal{F}$. 
    In fact, for each $\ell', x$, we just let $\tilde{\mathcal{F}}_{\ell', x} = \tilde{\mathcal{E}}_{\ell, x}$, where $\ell$ is the original layer index of $\mathcal{F}_{\ell', x}$.
    By the definition of the reorganization procedure, we can decompose $\mathcal{F}_{\leq \ell' - 1} = \mathcal{F}_{\ell' - 1} \circ \cdots \circ \mathcal{F}_{2} \circ \mathcal{F}_{1}$ into $\mathcal{F}_{\leq \ell' - 1} = \mathcal{C} \circ \mathcal{E}_{\ell} \circ \cdots \circ \mathcal{E}_2 \circ \mathcal{E}_1 (\sigma)$ where $\mathcal{C}$ is a sub-circuit of $\mathcal{E}_{\ell'-1}$.

    According to the Markov length definition of mixed-state phase, $\mathcal{C} \circ \mathcal{E}_{\ell} \circ \cdots \circ \mathcal{E}_2 \circ \mathcal{E}_1 (\sigma)$ exhibits exponential decay of CMI for tripartition $\Lambda=A_{\ell, x}B_{\ell, x}C_{\ell, x}$, and hence
    \begin{align}
        & \| \tilde{\mathcal{F}}_{\ell' , x} \circ \mathcal{F}_{\ell', x} \circ \mathcal{F}_{\leq \ell' - 1} (\sigma) -\mathcal{F}_{\leq \ell' - 1} (\sigma) \|_1 \nonumber\\
        =~& \| \tilde{\mathcal{E}}_{\ell, x} \circ \mathcal{E}_{\ell, x} (\mathcal{C} \circ \mathcal{E}_{\ell} \circ \cdots \circ \mathcal{E}_2 \circ \mathcal{E}_1 (\sigma)) - \mathcal{C} \circ \mathcal{E}_{\ell} \circ \cdots \circ \mathcal{E}_2 \circ \mathcal{E}_1 (\sigma) \|_1 \nonumber\\
        \leq~& \left( 2 \ln 2 \cdot I (A_{\ell, x} : C_{\ell, x} |  B_{\ell, x}) \right)^{\frac{1}{2}} \leq \sqrt{2 \gamma \ln 2} \cdot e^{- \frac{r}{2\xi}}.
    \end{align}
    As a result, we just need to take
    \begin{equation}
        r = \xi \cdot O ( \mathrm{log} (\gamma / \varepsilon_{\mathrm{LR}} ) ) = \xi \cdot O ( \mathrm{log} (n / \varepsilon_{\mathrm{LR}} ) ),
    \end{equation}
    which ensures $\| \tilde{\mathcal{F}}_{\ell , x} \circ \mathcal{F}_{\ell, x} \circ \mathcal{F}_{\leq \ell - 1} (\sigma) -\mathcal{F}_{\leq \ell - 1} (\sigma) \|_1 \leq \varepsilon_{\mathrm{LR}}$.
\end{proof}
We remark here that even if the new circuit $\mathcal{F}$ has a depth $d' = d \cdot O(r^k)$, the backward lightcone has a distance at most $s' = c d + r$.

\section{Classical Diffusion Models}
\label{sec:diffusion}

Suppose a $k$-dimensional lattice $\Lambda$, where each site supports a random variable from a sample space $\mathcal{X}$. A probability distribution is a function $P (y)$ where $y \in \mathcal{X}^n$. For simplicity, we assume a discrete sample space $\mathcal{X}$. When there is no confusion, we represent the function $P$ by a vector $P = (P_1, \cdots, P_{| \mathcal{X} |^n})$, where $\sum_i P_i = 1$ and $P_i \in [0, 1]$ for each $i \in \mathcal{X}^n$.

We embed the probability distribution into a diagonal mixed state
\begin{equation}
    \rho =\mathrm{diag} (P_1, \cdots, P_{| \mathcal{X} |^n}) = \sum_i P_i \ket{i}\!\bra{i} .
\end{equation}
A \textit{classical noisy channel} (or simply \textit{noisy channel}) is given by the probability transition matrix $T \in \mathbb{R}^{| \mathcal{X} |^n \times | \mathcal{X} |^n}$ such that
\begin{equation}
    T (P)_i = \sum_j P_j T_{j i} . \label{eq:TP}
\end{equation}
Then, we can also embed the classical noisy channel into the quantum channel by defining
\begin{equation}
    \mathcal{E} (X) = \sum_{i j} K_{i j} X K^{\dagger}_{i j}, \quad K_{i j} = T_{j i}^{\frac{1}{2}} \ket{i}\!\bra{j}, \label{eq:P_to_rho}
\end{equation}
where $K_{i j} \in \mathbb{R}^{| \mathcal{X} |^n \times | \mathcal{X} |^n}$ are the Kraus operators satisfying the normalization relation $\sum_{i j} K^{\dagger}_{i j} K_{i j} =\mathbb{I}_{| \mathcal{X} |^n}$. We can verify that
$\mathcal{E}$ acting on $\rho$ indeed realize Eq.(\ref{eq:TP}):
\begin{equation}
    \mathcal{E} (\rho) = \sum_{i j k} T_{j i}^{\frac{1}{2}} \ket{i}\!\bra{j} \left( P_k \ket{k}\!\bra{k} \right) T_{j i}^{\frac{1}{2}} \ket{j}\!\bra{i} = \sum_{i j} P_j T_{j i}  \ket{i}\!\bra{i} = \sum_i T (P)_i \ket{i}\!\bra{i} .
\end{equation}
\begin{definition}[Diffusion model]
    \label{def:DM}
    Suppose two distributions $P$ and $Q$. If there is a process (called \textnormal{diffusion}) generating $P$ from $Q$ through a continuous time evolution, described by the master equation (for continuous sample space $\mathcal{X}$, it is the Fokker-Planck equation)
    \begin{equation}
            \frac{\partial \tilde{p}_t}{\partial t} = \tilde{L}_t [\tilde{p}_t], \quad \tilde{p}_{t = 0} = P, \quad \tilde{p}_{t = 1} = Q, \quad \forall t \in [0, 1] .
    \end{equation}
    where $\tilde{L}_t$ is a linear operator. Then, there is another process (called \textnormal{denoising}) $L_t$ such that (notice that the initial and final distributions are inverted now)
    \begin{equation}
        \frac{\partial p_t}{\partial t} = L_t [p_t], \quad p_t = \tilde{p}_{1-t}, \quad \forall t \in [0, 1] .
    \end{equation}
    The pair $(\tilde{L}_t, L_t)$ is a \textnormal{diffusion model} that generates $P$ from $Q$ along the noise path $p_t$. Especially, $Q$ is usually selected as the trivial distribution $Q (0^n) = 1$.
\end{definition}

In practice, the dynamics are always time-discrete, which means we can give a time-discrete version of diffusion models: there exist classical noisy channels $T$ and $\tilde{T}$ with $d$ layers such that
\begin{equation}
    Q = \tilde{T} (P) =: \tilde{T}_1 \circ \tilde{T}_2 \circ \cdots \circ \tilde{T}_d (P), \quad P = T (Q) = T_d \circ \cdots \circ T_2 \circ T_1 (Q) , \label{eq:DM_def_1}
\end{equation}
and $Q$ is layer-wise reversible from $P$ along $T$, namely 
\begin{equation}
    \tilde{T}_{\ell} \circ T_{\ell} \circ T_{\leq \ell-1} (Q) = T_{\leq \ell-1} (Q),
\end{equation}
where $T_{\leq \ell-1} := T_{\ell-1} \circ \cdots \circ T_{2} \circ T_{1}$ is the first $(\ell-1)$ layer of $T$.
According to Ref.\,\cite{hu_2025_local}, it is intuitive to suppose that each $T_{\ell}$ and $\tilde{T}_{\ell}$ has finer structures:
\begin{equation}
    T_{\ell} = \prod_x T_{\ell, x}, \quad \tilde{T}_{\ell} = \prod_x \tilde{T}_{\ell, x}, \quad \text{$ T_{\ell, x} $ and $ \tilde{T}_{\ell, x} $ are local} . \label{eq:DM_def_2}
\end{equation}
Similar to Definition \ref{def:MSP_via_LR}, we denote $| \mathrm{Supp} (T_{\ell, x}) |, | \mathrm{Supp} (\tilde{T}_{\ell, x}) | \leq c$, and we also define $s = c d$.

Since both $P$ and $Q$ are classical distributions, we can always embed them into diagonal mixed states $\rho =\mathrm{diag} (P)$ and $\sigma =\mathrm{diag} (Q)$ using Eq.\,(\ref{eq:P_to_rho}). 
Then the classical limit of mixed state phases via local reversibility (Definition \ref{def:MSP_via_LR}) naturally induces the phases of classical distributions via local reversibility \cite{hu_2025_local, schuster_2025_hardness}. 

\begin{definition}[Phases of distributions via local reversibility]
    \label{def:DP_via_LR}
    Suppose two distributions $P$ and $Q$. We say $P$ and $Q$ are in the same phases of distribution, if for any $\eta > 0$ and $\varepsilon_{\mathrm{LR}} > 0$, there exists $T = T_d \circ \cdots \circ T_2 \circ T_1$ where $T_{\ell} = \prod_x T_{\ell, x}$ and each $T_{\ell, x}$ has a local support, such that:
    \begin{enumerate}
        \item The total variation distance $\frac{1}{2}| P - T (Q) |_1 \leq \eta$,
        \item $Q$ is \textnormal{$\varepsilon_{\mathrm{LR}}$-locally reversible} from $T(Q)$ along $T$: for each $T_{\ell, x}$, there exists another channel gate $\tilde{T}_{\ell, x}$ with the same spatial support as $T_{\ell, x}$ such that
        \begin{equation}
            \left| \tilde{T}_{\ell,x} \circ T_{\ell,x} \circ T_{\leq \ell - 1} (Q) - T_{\leq \ell - 1} (Q) \right|_1 \leq \varepsilon_{\mathrm{LR}}.
        \end{equation}
        where $T_{\leq \ell - 1} := T_{\ell - 1} \circ \cdots \circ T_2 \circ T_1$ represents the first $\ell - 1$ layers of T.
    \end{enumerate} 
    Especially, we say a distribution $P$ is in the trivial phase if $Q$ is the trivial distribution $Q (0^n) = 1$.
\end{definition}

% The phase structures for classical distribution are also abundant.
% In the classical case, for a trivial phase distribution $P$, it is also possible that there exists another shallow circuit $T'$ preparing $P$ that passes through an intermediate distribution violating local reversibility -- for example, a uniform mixture over all closed loop configurations with trivial homology on a torus \cite{hu_2025_local}.

On the other hand, our Theorem \ref{thm:AM}, Theorem \ref{thm:LE}, and Main Theorem \ref{thm:main} immediately imply the following crucial properties for arbitrary trivial phase distributions:

\begin{corollary}[Learning trivial phase distribution]
    \label{cor:learning_TPD}
    Suppose a trivial phase distribution $P$, it satisfies the following properties: given arbitrary $\varepsilon > 0$ and $\delta > 0$,
    \begin{enumerate}
        \item Approximate Markovianity: for any tripartition $\Lambda = A B C$ with $\mathrm{dist} (A, C) \geq 2 s$, there exists classical local recovery map $\Psi_{B \rightarrow B C}$ such that
        \begin{equation}
            | P - \Psi_{B \rightarrow B C} (P_{A B}) |_1 \leq \varepsilon .
        \end{equation}
        
        \item Local extendibility: for any pentapartition $\Lambda = A B C D E$ with $\mathrm{dist} (A, C) \geq 2 s$ and $\mathrm{dist} (B, D) \geq 2 s$, there exists classical local extension map $\Phi_{B E \rightarrow B C}$ such that
        \begin{equation}
            | P_{A B C} - \Phi_{B E \rightarrow B C} (P_{A B E}) |_1 \leq \varepsilon .
        \end{equation}
    \end{enumerate}
    Furthermore, there exists an algorithm that runs time $T$, learns from $M$ many samples of $P$, and generates $P$ using a $(k + 1)$-layer local noisy channel circuit $W$, where (recall $\mathcal{X}$ is a discrete sample space)
    \begin{equation}
        T = n \cdot | \mathcal{X} |^{O (c d)^k} \log \left( \frac{n}{\varepsilon} \right) + O (M), \quad M = \frac{n^2 \cdot | \mathcal{X} |^{O (c d)^k}}{\varepsilon^2} \log \left( \frac{n}{\delta} \right) .
    \end{equation}
    such that $\frac{1}{2}| W(Q) - P |_1 \leq \varepsilon$ with success probability at least $1 - \delta$.
\end{corollary}

\textbf{Remark}. For a continuous sample space $\mathcal{X}$, the approximate Markovianity and local extendibility are still true, following similarly to the proof of Theorem \ref{thm:AM} and Theorem \ref{thm:LE}. But for a continuous sample space, the semidefinite programming algorithm does not work. However, our approximate Markovianity and local extendibility, together with the covering schemes described in Section \ref{sec:cover_scheme}, still imply a more efficient learning paradigm of trivial phase distribution $P$, by learning recovery maps and extension maps locally and stitching the learned local pieces consistently.

% \section{Algorithm}

\end{document}